%% file: bare_jrnl_compsoc.tex
\pdfoutput=1
\documentclass[journal]{IEEEtran}
\usepackage{amssymb,amsmath}
\usepackage[lined,boxed,commentsnumbered,ruled]{algorithm2e}
\usepackage{algorithmic} 
\usepackage{graphicx}
\usepackage{stmaryrd}
\SetSymbolFont{stmry}{bold}{U}{stmry}{m}{n}
\usepackage{multirow,url}
\usepackage{color,cite}
\usepackage{cite}
\usepackage{tikz,pgf}
\usetikzlibrary{positioning}
\usetikzlibrary{fadings}
\usepackage{makecell}
\usepackage{array}
\usepackage{colortbl,textcomp,booktabs}
\usepackage{overpic} 

\newtheorem{lemma}{Lemma}

\newtheorem{proposition}{Proposition}
\newtheorem{definition}{Definition}
\newtheorem{theorem}{Theorem}
\newtheorem{observation}{Observation}

\ifodd 1
\newcommand{\com}[1]{\textbf{\color{red} (COMMENT: #1)}} 
\newcommand{\response}[1]{\textbf{\color{green} (RESPONSE: #1)}} 
\newcommand{\del}[1]{}
\else

\newcommand{\com}[1]{}
\newcommand{\comg}[1]{}
\newcommand{\response}[1]{}
\newcommand{\del}[1]{}
\fi

\definecolor{DGreen}{rgb}{0.16,0.38,0.27}
\definecolor{mygray}{gray}{.9}
\definecolor{shadecolor}{gray}{0.9}

\def\Ex{\mathbb{E}}
\def\M{\mathcal{M}}

\def\N{\mathcal{N}}

\begin{document}
%
\title{Optimal Pricing and Admission Control for Heterogeneous Secondary Users}

\author{Changkun~Jiang,~\IEEEmembership{Student Member,~IEEE,}
        Lingjie~Duan,~\IEEEmembership{Member,~IEEE,}
        and~Jianwei~Huang,~\IEEEmembership{Fellow,~IEEE}
\thanks{Part of the results appeared in IEEE WiOpt, May 2014 \cite{JDH-WiOpt}.}
\thanks{Changkun Jiang and Jianwei Huang (corresponding author) are with the Network Communications and Economics Lab, Department of Information Engineering, The Chinese University of Hong Kong
(E-mail: \{jc012, jwhuang\}@ie.cuhk.edu.hk); Lingjie Duan is with the Engineering Systems and Design Pillar, Singapore University of Technology and Design, 8 Somapah Road, Singapore (E-mail: lingjie\_duan@sutd.edu.sg).}
}

\IEEEcompsoctitleabstractindextext{\vspace{-6mm}%
\begin{abstract}
This paper studies how to maximize a spectrum database operator's expected revenue in sharing spectrum to secondary users, through joint pricing and admission control of spectrum resources. A unique feature of our model is the consideration of the stochastic and heterogeneous nature of secondary users' demands. We formulate the problem as a stochastic dynamic programming problem, and present the optimal solutions under both static and dynamic pricing schemes. In the case of static pricing, the prices do not change with time, although the admission control policy can still be time-dependent. In this case, we show that a stationary (time-independent) admission policy is in fact optimal under a wide range of system parameters. In the case of dynamic pricing, we allow both prices and admission control policies to be time-dependent. We show that the optimal dynamic pricing can improve the operator's revenue by more than $30\%$ over the optimal static pricing, when secondary users' demands for spectrum opportunities are highly elastic.
\end{abstract}
}

\maketitle

\IEEEdisplaynotcompsoctitleabstractindextext

\IEEEpeerreviewmaketitle

\input{Introduction}
\input{Section-Model}
\input{Section-Optimal}
\input{Section-Static}
\input{Section-DynamicPricing}
\input{Section-Extension}
\input{Simulation}
\input{Section-Conclusion}

\appendices
\section{Impact of Wireless Channel Conditions}
In this section, we incorporate the wireless channel conditions into our model. We generalize our model based on such a new consideration. We show that as the average channel gain increases, SUs are more willing to buy the spectrum for any given price, and accordingly motivates the operator to increase the price to increase its revenue.

Now we introduce the detailed channel model.
We assume a block-fading channel, where an SU's channel condition remains fixed during a time slot, and can change independently across time slots.
We assume that different SUs' channel processes are statistically identical, but different SUs will still have different channel realizations. Let $\M_l$ and $\M_h$ be the sets of light-traffic and heavy-traffic SUs, respectively.
We denote the channel realization of SU $m$ in time slot $n$ as $Z_m[n]$, where $m\in\M_l\cup\M_h$ and $n\in\N$.

Given the channel realization $Z_m[n]$ in time slot $n$ and the length of each time slot $T$, the achievable total transmitted data (in nats) of a light-traffic SU $m\in\M_l$ in time slot $n$ is
\begin{equation}
\textstyle \delta_m=Tw\ln\left(1+\frac{P^{\max}Z_m[n]}{n_0w}\right), \forall m\in\M_l.
\end{equation}
Here $w$ is the channel bandwidth, $P^{\max}$ is SU's maximum transmission power, $n_0$ is the noise power per unit bandwidth.
Similarly, given the channel realizations $Z_m[n]$ and $Z_m[n+1]$ in time slots $n$ and $n+1$, the achievable total transmitted data (in nats) of a heavy-traffic SU $m\in\M_h$ by occupying time slots $n$ and $n+1$ is
\begin{equation}
\begin{aligned}
\delta_m &=\textstyle Tw\ln\left(1+\frac{P^{\max}Z_m[n]}{n_0w}\right)\\
&~~\textstyle +Tw\ln\left(1+\frac{P^{\max}Z_m[n+1]}{n_0w}\right), \forall m\in\M_h.
\end{aligned}
\end{equation}

We further denote the light-traffic SU's utility in time slot $n$ as
\begin{equation}
U_m(Z_m[n])=\eta \cdot\delta_m, m\in\M_l,
\end{equation}
where $\eta$ is SUs' evaluation for per unit data. A heavy-traffic SU's utility in time slots $n$ and $n+1$ is\footnote{This is due to the consecutive occupancy of a heavy-traffic SU if admitted.}
\begin{equation}
U_m(Z_m[n],Z_m[n+1])=\eta \cdot\delta_m, m\in\M_h.
\end{equation}

Next we consider the impact of channel conditions on SUs and the operator, respectively. An SU's wireless channel condition will influence its achievable data volume and utility, hence will determine whether he is willing to accept the price and request spectrum access. Meanwhile, the demand probability of each type of SUs (i.e., those who are willing to accept the prices) will influence the operator's expected revenue. Hence, the operator needs to consider SUs' evaluation of the channel conditions and their aggregate demand when optimizing the pricing. Next we discuss these issues in details.

From an SU's point of view, it will choose to request the spectrum access if and only if its utility is no less than the price ($r_l$ or $r_h$) charged by the operator in time slot $n$. More specifically, for light-traffic and heavy-traffic SUs, the conditions for them to request the channel access are
\begin{equation}
U_m(Z_m[n])\geq r_l, \forall m\in\M_l \text{~~and~~} U_m(Z_m[n])\geq r_h, \forall m\in\M_h,
\end{equation}
respectively. 
If there are two light-traffic SUs arriving in  the same time slot, it is possible that one SU requests access (as its channel is good) and the other SU does not request (as its channel is bad). This reflects the impact of channel conditions on the SUs' decisions.

From the operator's point of view, the operator computes the admission policy offline, based on the statistics of SUs' ergodic channel processes (which are assumed to be statistically identical across all SUs).
More specifically, in time slot $n$, each type of SUs' utilities follow some distribution, due to the random channel realizations of the SUs. Such a distribution is a function of the price. We thus define the demand probability as a function of the aggregate demand level and the price. 
That is, we denote the demand probability as $\phi_i\left(\Ex_{\{Z_{m}[n],\forall m\in\M_i\}}U_m(Z_m[n]),r_i\right),i\in\{l,h\}$, which is increasing in the average utility level $\Ex_{\{Z_{m}[n],\forall m\in\M_i\}}U_m(Z_m[n])$ and decreasing in the operator's price $r_i$. We use the parameter $\hat{k}_i$ to denote the price sensitivity with respect to the price $r_i$, and the parameter $k_i$ to capture the price sensitivity with respect to both the channel statistics and price. Similar to [1, Ch.3], our paper presents a linear demand probability function, i.e.,
\begin{equation}
\begin{aligned}
&\textstyle \phi_i\left(\Ex_{\{Z_{m}[n],\forall m\in\M_i\}}U_m(Z_m[n]),r_i\right)\\
&\textstyle =\frac{1}{\Ex_{\{Z_{m}[n],\forall m\in\M_i\}}U_m(Z_m[n])}\left(\Ex_{\{Z_{m}[n],\forall m\in\M_i\}}U_m(Z_m[n])\!-\!\hat{k}_ir_i\right)\\
&\textstyle =1\!-\!\frac{\hat{k}_i}{\Ex_{\{Z_{m}[n],\forall m\in\M_i\}}U_m(Z_m[n])}r_i=1\!-\!k_ir_i,i\in\{l,h\},m\in\M_i,
\end{aligned}
\end{equation}
Here, we use a proper scaling factor $1/\Ex_{\{Z_{m}[n],\forall m\in\M_i\}}U_m(Z_m[n])$ to ensure that the demand probability is 1 when the price $r_i$ is 0. (In other words, we consider the case where there are always demands when there is no cost of using the spectrum.)  As we can see, the demand elasticity $k_i$ has taken into account the impact of the channel statistics.

[1] A. Mas-Colell, M. D. Whinston, and J. R. Green, ``Microeconomic Theory,'' Oxford Univ. Press, New York, 1995.


Next, we show the simulation results on the impact of the wireless channel states. Recall that our theoretical results apply to the case where the arrival process of each type of SUs belongs to the class of processes with at least one SU arrival at the beginning of each time slot. In the simulations, we consider a special case of the deterministic arrival process, such that there is exactly one heavy-traffic SU and one light-traffic SU arriving at the beginning of each time slot. Each SU has a different channel realization and may or may not request the spectrum, depending on the relationship between its utility obtained by using the channel and the charged price. More specifically, the system works as follows. At the beginning of each time slot, SUs arrive and observe their channel states. Then each SU makes the request decision according to its own channel state and the charged price. If there is one or more SUs requesting to access the channel in a time slot, the operator makes the admission decision. In this numerical example, we fix the price ratio such that the \emph{stationary heavy-priority admission policy} is optimal for all time slots, and illustrate the results in Figure 1. Subfigure 1 shows the random channel realizations over time, which influence SUs' request decisions over time in Subfigure 2. Subfigure 3 shows the operator's admission decisions accordingly.  We can see that the wireless channel states indeed affect SUs' requests for the channel after arriving. By comparing the first two subfigures that have similar patterns, we can see that better channel conditions lead to more requests. The operator's admission decisions verify the optimality of the stationary heavy-priority admission policy, which admits a heavy-traffic SU whenever possible, and only admits a light-traffic SU when there is only a light-traffic SU.

\begin{figure}[ht]
\centering
\includegraphics[width=0.45\textwidth]{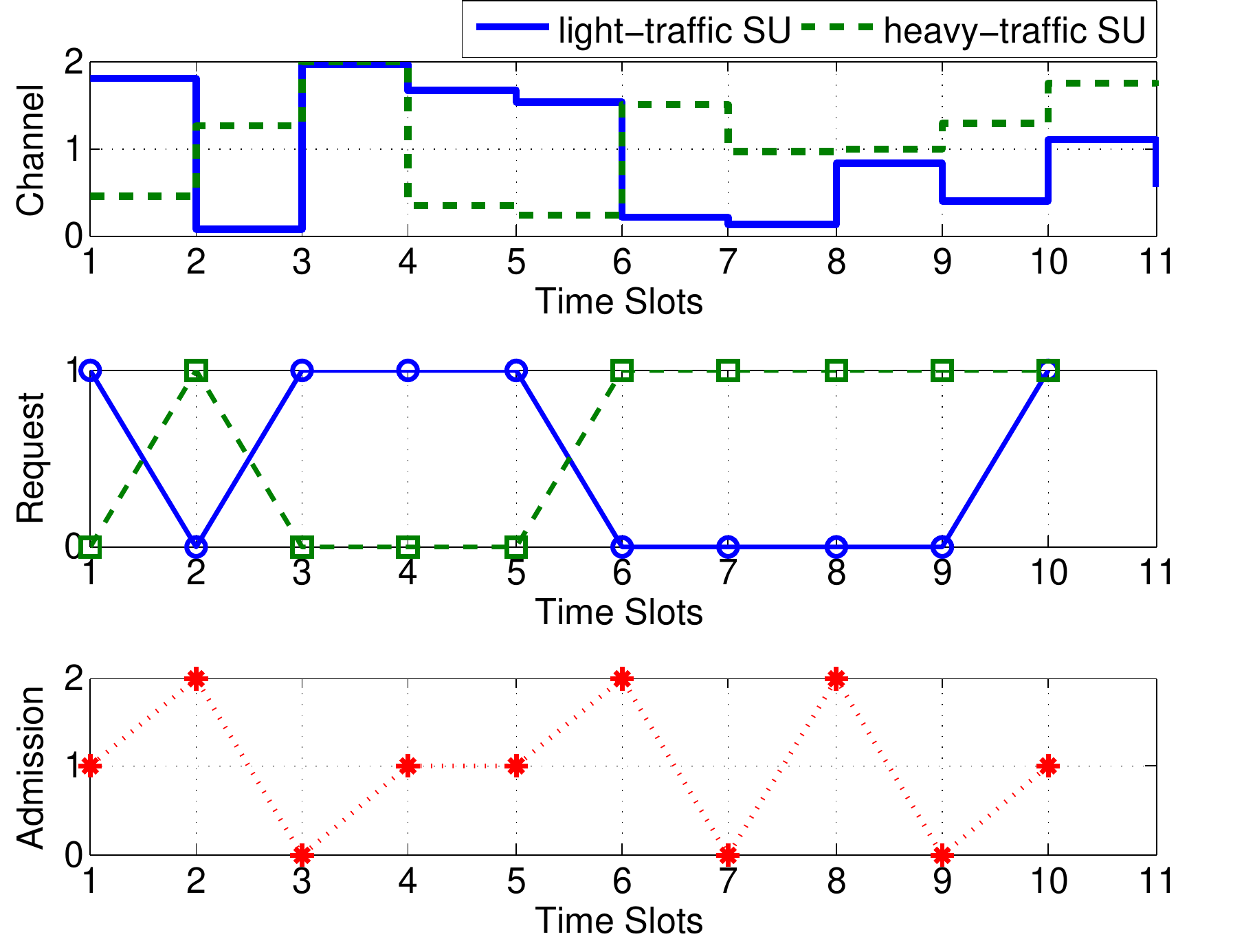}
\caption{Impact of the wireless channel states. The first subfigure shows the channel states of the two types of SUs, which is randomly generated and belongs to $[0,2]$. The second subfigure shows the request decisions after observing the channel states, where the heavy-traffic SU needs to consider the channel gains in two consecutive time slots. The last subfigure shows the operator's optimal admission decisions, where 2 represents admitting a heavy-traffic SU, 1 represents admitting a light-traffic SU, and 0 represents admitting nobody. Notice that if a heavy-traffic SU is admitted in a time slot, then the next time slot cannot admit any SU, due to the consecutive occupancy of two adjacent time slots by the admitted heavy-traffic SU.}
\label{fig:arch}
\end{figure}

\input{Section-Appendix}
\input{Section-Problem4}
\input{Section-Problem5}
\input{Section-Problem6}
\input{Section-ExtensionAppendix}
\section{Performance Comparison with a Related Study}
In this section, we compare our proposed scheme with the related studies in the literature. After checking the related studies on spectrum pricing and access control in the existing literature, it seems that we are the first to study the spectrum database operator's optimal pricing and dynamic admission control for heterogeneous secondary demands. In the more general area of dynamic pricing and admission control, we identify the following reference [1] ([27] in our manuscript) which studied a related (but not the same) problem in the operations research area.

[1] G. Y. Lin, Y. Lu, and D. D. Yao, ``The Stochastic Knapsack Revisited: Switch-Over Policies and Dynamic Pricing'' \emph{Operations Research}, vol. 56, no. 4, pp. 945-957, 2008.

In [1], the authors proposed the switch-over policies for the stochastic knapsack problem with dynamic pricing, where the policies start by accepting only demands associated with the highest price, and switch to accepting demands with lower prices as time goes by. This motivates us to consider the following heuristic switch-over admission policy in our context: the operator admits a heavy-traffic SU only if half of the price charged to heavy-traffic SUs is no smaller than the price charged to light-traffic SUs, i.e., $r_h/2 \geq r_l$. We define the revenue improvement of $R_1$ over $R_2$ as $(R_1-R_2)/R_2$. Figure 2 shows the revenue improvement of our optimal dynamic pricing and dynamic admission policy over the heuristic switch-over dynamic pricing and dynamic admission policy. We can see that the revenue improvement depends on the values of the demand elasticities. The smaller difference of the demand elasticities between the heterogeneous SUs will lead to larger revenue improvements. In general, our proposed scheme outperforms the policy proposed in [1] significantly (larger than $10\%$) in terms of the obtained revenue.

\begin{figure}[ht]
\centering
\vspace{10pt}
\begin{overpic}[scale=0.45]{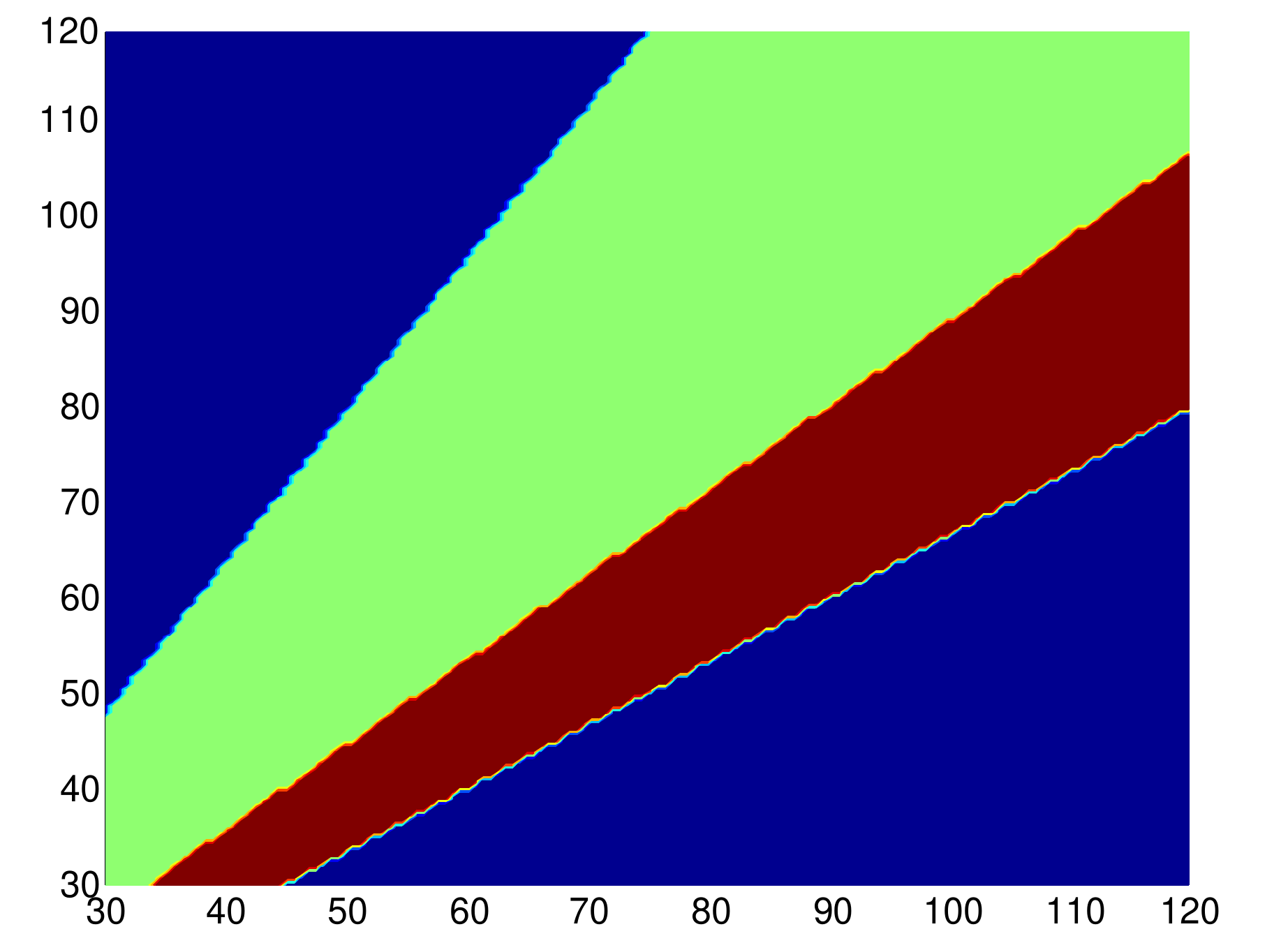}
\put(-1,13){\rotatebox{90}{{\scriptsize Heavy-traffic demand elasticity $k_h$}}}
\put(28,-3){\rotatebox{0}{{\scriptsize Light-traffic demand elasticity $k_l$}}}

\put(94,26){\rotatebox{90}{{\scriptsize Revenue Improvement $\geq30\%$}}}
\put(94,53){\rotatebox{35}{\vector(-1,0){15}}}

\put(0,73){\rotatebox{0}{{\scriptsize Revenue Improvement $<10\%$}}}
\put(18,73){\rotatebox{70}{\vector(-1,0){10}}}

\put(48,73){\rotatebox{0}{{\scriptsize $10\%\leq$ Revenue Improvement $<30\%$}}}
\put(70,73){\rotatebox{70}{\vector(-1,0){10}}}

\end{overpic}
\caption{Revenue improvement of our proposed optimal policy over the heuristic policy in [1].}
\label{fig:arch2}
\end{figure}

\end{document}

%% file: Introduction.tex
\section{Introduction}\label{sec:intro}


\IEEEPARstart{D}{atabase-assisted} spectrum sharing is a promising approach to improve the utilization of limited spectrum resources \cite{FCC,Ofcom}. In such an approach, primary licensed users (PUs) report their spectrum usage patterns to a spectrum database, which uses the primary activity records to coordinate the opportunistic spectrum access of secondary unlicensed users (SUs). Several government regulators, such as the FCC in the US and the Ofcom in the UK, strongly advocate such an approach (e.g., for the sharing of TV white space) due to its high reliability compared to sensing. Under such an approach, the database can effectively coordinate SUs' accesses, by mitigating these SUs' mutual usage conflicts and controlling the potential conflicts with PUs. Though researchers have made significant research progress in addressing various technical issues of spectrum database (e.g., database system management and spectrum allocation \cite{Feng,FengZhang,LuoGao,LuoGaoWiOpt}), very few studies looked at the economic issue of spectrum database (e.g., \cite{LuoGao,FengZhang,LuoGaoWiOpt}). Without a proper economic mechanism, the database operator may not have enough incentives to coordinate the spectrum sharing process. This motivates us to explore the revenue maximization problem for a spectrum database, in particular, the admission control of SUs and pricing of idle spectrum resources.

There are two key challenges when considering such a revenue maximization problem for the database operator. First, SUs' demands can be \emph{heterogeneous} in terms of spectrum occupancy. For example, a large file (e.g., video) downloading takes minutes or even hours to finish (hence we call heavy-traffic), while sending a short text message or accessing location-based services can be completed in seconds (hence we call light-traffic). Second, SUs' demands are often \emph{randomly} generated over time. The heterogeneity and randomness make it difficult for the operator to accurately predict future demands and make proper resource allocation decisions. 
However, these two issues have not been fully considered in the previous literature (e.g., \cite{LuoGao,FengZhang,LuoGaoWiOpt}).

To address the above two challenges,
we propose a joint spectrum pricing and admission control scheme for the database operator to maximize its expected revenue. The optimization is over the time period during which the spectrum channel is available for SUs to opportunistically access due to the lack of PU activities.  
The period is divided into several time slots, and the database operator needs to determine the optimal prices for different types of SUs (e.g., in heavy- and light-traffic types) in each time slot. These prices can be fixed (static pricing) or vary over time (dynamic pricing), and will affect how SUs request to access the spectrum. However, pricing alone may not be enough to mitigate the conflicts between multiple SUs who want to access to the limited spectrum at the same time. The operator also needs to determine the optimal admission control policy to control the total demand. The pricing and admission decisions need to be jointly optimized in order to achieve the maximum performance.

To our best knowledge, this is the first work that jointly prices and allocates the spectrum resource in a \emph{dynamic} setting to serve heterogeneous and stochastic SU demands. We formulate the operator's revenue maximization problem as a stochastic dynamic programming problem, which is in general challenging to solve. Our main results and key contributions are summarized as follows.
\begin{itemize}
\item \emph{Optimal static pricing and dynamic admission policies}.
We first constrain ourselves to the simple and widely used approach of static pricing, meanwhile allow dynamic time-dependent admissions. We show that the complex optimal dynamic admission policy often degenerates to a threshold-based stationary (time-independent) policy under a wide range of system parameters. For the scenario where a stationary policy is not optimal, we propose an algorithm to numerically compute the optimal admission policy.
\item \emph{Optimal dynamic pricing and dynamic admission policies}.
We further allow the prices to dynamically change over time based on different SUs' stochastic demands. Although the optimal prices and admission decisions are coupled, we are able to compute the optimal policy through a proper price-and-admission decomposition in each time slot. Similarly, we show that the optimal admission policy often degenerates to a stationary admission policy under a wide range of system parameters.
By comparing the optimal pricing and admission policies under both static and dynamic pricing schemes, we show that the dynamic pricing scheme can significantly improve the database operator's revenue (by more than $30\%$) when SU's demands are highly elastic. We also compare our dynamic pricing policy with a heuristic algorithm \cite{ORComparison} as the benchmark in terms of the revenue. 
\end{itemize}

We start our analysis by considering a simple case involving two types of SUs: a light-traffic SU who wants to access a channel for one time slot, and a heavy-traffic SU who needs to access a channel for two consecutive time slots. We further extend our analysis and optimal algorithm design to the more general case where (i) there are an arbitrary number of SU types, and (ii) each type of SU may access a channel for an arbitrary number of time slots.


The rest of the paper is organized as follows. We introduce the model and problem in Section \ref{sec:spectrumallocation}, assuming two SU types and a heavy-traffic SU will occupy the channel for two consecutive time slots. In Section~\ref{sunsec:pf}, we formulate and solve the optimal static pricing and dynamic admission control problem. In Section~\ref{sec:dynamicpricing}, we further consider the joint optimization of dynamic pricing and dynamic admission control problem. In Section \ref{sec:extension}, we extend our model and results to the case of an SU's arbitrary time slots occupancy and the case of multiple types of SUs. We show the simulation results in Section \ref{sec:simulation}.  Finally, we conclude the paper and discuss the future work in Section \ref{sec:conclusion}. 
\section{Related Work}
There are several recent results focusing on the spectrum pricing issues of a spectrum database (e.g., \cite{FengZhang,LuoGao,LuoGaoWiOpt}). These studies focused primarily on the static pricing with complete information in the spectrum database system, without considering the heterogeneous and stochastic SU demands as in our work. 
Besides the studies on spectrum database economics, there are also some recent related studies on secondary spectrum sharing and allocation (e.g., \cite{add1,add2,add3,add4}), which focused primarily on static pricing and static access control. The focus of our study is the pricing and dynamic admission control with time domain heterogeneity. This is a practically important issue that has not been explicitly addressed in existing studies.

There has been a rich literature on static or dynamic pricing of resources in Internet, communication networks, and transportation networks. Under static pricing, the pricing decisions do not change over time (e.g.,\cite{FengZhang,LuoGao,LuoGaoWiOpt,leiyang,XWang,GizelisVergados,Walrand,Niyato2,add1,add2,add3,add4}). 
The literature on dynamic pricing focuses primarily on dynamic pricing decisions of selling a given stock of items by a deadline (e.g., \cite{DanZhang,Guillermo}), and in particular, pricing decisions of airline seats and hotel rooms booking (e.g., \cite{PSYOU,Hotel}). However, in our work spectrum has its unique features to be priced and used. Unlike a traditional product, the unused spectrum resource cannot be stored and is immediately wasted. SUs' demands are also heterogeneous over time. The time-sensitive feature of spectrum and the demand heterogeneity make our model and analysis fundamentally different from prior studies. 

The literature on dynamic pricing of wireless resources only emerged recently (e.g., \cite{YuguangFang,JoeWongChiang,QianHuang}). Song \emph{et al.} in \cite{YuguangFang} studied the network revenue maximization problem by using dynamic pricing for multiple flows in a wireless multi-hop network. Ha \emph{et al.} in \cite{JoeWongChiang} proposed time-dependent pricing to decrease customers' congestion cost. Ma \emph{et al.} in \cite{QianHuang} proposed time and location based pricing for mobile data traffic. However, none of these prior studies focused SUs' stochastic and random spectrum demands. We tackle this issue by jointly considering admission control and pricing, and characterize the conditions under which the often complicated optimal pricing and admission decisions degenerate to the stationary pricing and admission schemes.

%% file: Section-Model.tex
\section{System Model and Problem Formulation}\label{sec:spectrumallocation}
We consider a database operator who records PUs' activities and knows a channel that will not used by PUs during a set $\mathcal{N}=\{1,\cdots,N\}$ of consecutive time slots, similar as in \cite{FCC,Ofcom,Feng}.\footnote{In this paper we focus on the optimal joint pricing and admission control of a single channel, which is already quite mathematically challenging. We may use our results in a multi-channel scenario, if the admission and pricing of each channel is done independently. The more general case of the joint consideration of multiple channels will be left as a future work.} The database operator wants to maximize its revenue through selling the temporary spectrum opportunities to the SUs.
The duration of this whole time period depends on the type of PU traffic, and is known in advance as the PUs need to register all traffic with the database (e.g., \cite{FCC,Ofcom,Feng}).

SUs randomly arrive and request channel access at the beginning of each time slot. To gain clear insights into the admission policies of SUs, we first assume that there are two types of SUs depending on the length of the channel access time. A \emph{light-traffic} SU only needs to use the channel for one time slot, and a \emph{heavy-traffic} SU needs to occupy two consecutive time slots.  In Section \ref{sec:extension}, we will extend our analysis to the case where a heavy-traffic SU occupies more than two consecutive time slots, and we will show that our main results do not change. We will further consider the general case of multiple types of SUs in terms of spectrum occupancy.

If an SU is admitted in $n\in\mathcal{N}$, the database operator will charge the SU either $r_l(n)$ or $r_h(n)$, depending on whether it is a light- or heavy-traffic SU. SUs are price-sensitive, and their demand probabilities of requesting the spectrum after arriving are non-increasing in the prices.\footnote{We in fact consider two different arrival processes. The first process describes how the SUs arrive at the system, which can be any process such that there is at least one arrival at the beginning of each time slot. One example of such processes is the deterministic arrival in each time slot. The second process characterizes how the arrived SUs request spectrum access from the database. Such a process depends on the prices set by the operator, as a higher price will reduce the demand from SUs.} Since we consider a single channel case, the database operator can admit at most one SU (light- or heavy-traffic) in any time slot. Once an SU's service request is rejected by the database, it will leave the system without waiting. This corresponds to the case where SUs have delay-intolerant applications such as VoIP and video conferencing.

Fig. \ref{system} summarizes the database's operations in our model. At the beginning of each time slot $n\in\mathcal{N}$, the database operator first announces prices $r_l(n)$ and $r_h(n)$ for the light- and heavy-traffic SU types, respectively. Then SUs observe the price update and randomly arrive with the probabilities affected by the prices. Finally, the database operator admits at most one SU to the channel (if the channel is available) and rejects the other SUs (if any).
After these three phases, the admitted SU will transmit data over the channel during the rest of the slot.\footnote{We assume that the signaling overhead is small in each time slot, since only the price and admitted SU's information is sent from the database to SUs, and SUs only send 1-bit binary demand type information to the database.}

In addition to the prices charged by the operator, the wireless channel condition will also influence each SU's request for the spectrum. From each individual SU's perspective, an SU's wireless channel condition will influence its achievable transmittd data volume and utility, hence will determine whether he is willing to accept the price and request spectrum access. From the database operator's perspective, the demand probability of each type of SUs (i.e., those who are willing to accept the prices) will influence the operator's expected revenue. Hence, the operator needs to consider SUs' evaluation of the channel condition and their \emph{average} demand when optimizing the pricing. In Sections \ref{sunsec:pf} and \ref{sec:dynamicpricing}, we will consider a demand probability function that decreases in prices, where the demand elasticity parameter incorporates the impact of the channel statistics. Due to page limit, we put the detailed channel modeling, analysis, and simulation results into Appendix A. 

\begin{figure}
   \centering
   \pgfdeclarelayer{background}
   \pgfdeclarelayer{foreground}
   \pgfsetlayers{background,main,foreground}
\begin{tikzpicture}[
  level 1/.style={sibling distance=40mm},
  edge from parent/.style={->,draw},
  >=latex]

\tikzstyle{line} = [draw, -latex']
    \begin{scope}[thick,font=\scriptsize]
    \draw [->,thick] (-4,0) -- (4,0) node [above left]  {Time Horizon};
    \draw (1.6,-3pt) -- (1.6,3pt)   node [above] {$N$};
    \draw (-1,-3pt) -- (-1,3pt) node [above] {$n$};
    \draw (-3.6,-3pt) -- (-3.6,3pt) node [above] {$1$};
    \draw (-2.3,-3pt)  (-2.3,3pt) node [above] {$\cdots$};
    \draw (0.3,-3pt)  (0.3,3pt) node [above] {$\cdots$};

\filldraw[fill=green!20] (1.6,-2) rectangle (3.2,-0.3);
    \node [rotate=90] at (3.0,-1.15){SU Access};

    \path[draw,fill=red!20](1.6,-0.3)--(2.0,-0.3)--(2.0,-2)--(1.6,-2)--cycle;
    \node [rotate=90] at (1.8,-1.15) {Pricing};
    \path[draw,fill=blue!20](2,-0.3)--(2.4,-0.3)--(2.4,-2)--(2,-2)--cycle;
    \node [rotate=90] at (2.2,-1.15) {SUs Demand};
    \path[draw,fill=DGreen!20](2.4,-0.3)--(2.8,-0.3)--(2.8,-2)--(2.4,-2)--cycle;
    \node (A) [rotate=90] at (2.6,-1.15) {Admission};
    \node (A1) at (2.4,-2.3) {$(S_N,X_N,Y_N)$};

\filldraw[fill=green!20] (-1,-2) rectangle (0.6,-0.3);
    \node [rotate=90] at (0.4,-1.15){SU Access};

    \path[draw,fill=red!20](-1,-0.3)--(-0.6,-0.3)--(-0.6,-2.0)--(-1,-2.0)--cycle;
    \node [rotate=90] at (-0.8,-1.15) {Pricing};
    \path[draw,fill=blue!20](-0.6,-0.3)--(-0.2,-0.3)--(-0.2,-2)--(-0.6,-2)--cycle;
    \node [rotate=90] at (-0.4,-1.15) {SUs Demand};
    \path[draw,fill=DGreen!20](-0.2,-0.3)--(0.2,-0.3)--(0.2,-2)--(-0.2,-2)--cycle;
    \node [rotate=90] at (-0,-1.15) {Admission};
    \node at (-0.2,-2.3) {$(S_n,X_n,Y_n)$};

\filldraw[fill=green!20] (-3.6,-2) rectangle (-2,-0.3);
    \node [rotate=90] at (-2.2,-1.15){SU Access};

    \path[draw,fill=red!20](-3.6,-0.3)--(-3.2,-0.3)--(-3.2,-2)--(-3.6,-2)--cycle;
    \node (P) [rotate=90] at (-3.4,-1.15) {Pricing};
    \path[draw,fill=blue!20](-3.2,-0.3)--(-2.8,-0.3)--(-2.8,-2)--(-3.2,-2)--cycle;
    \node [rotate=90] at (-3,-1.15) {SUs Demand};
    \path[draw,fill=DGreen!20](-2.8,-0.3)--(-2.4,-0.3)--(-2.4,-2)--(-2.8,-2)--cycle;
    \node [rotate=90] at (-2.6,-1.15) {Admission};
    \node at (-2.8,-2.3) {$(S_1,X_1,Y_1)$};

    \draw [->] (-3.6,0) -- (-3.6,-0.3);
    \draw [->] (-1,0) -- (-1,-0.3);
    \draw [->] (1.6,0) -- (1.6,-0.3);


    \node at (1.1,-1.15) {$\cdots$};
    \node at (-1.5,-1.15) {$\cdots$};
\end{scope}
\begin{pgfonlayer}{background}
        \path (P.west |- P.north)+(-0.8,1.7) node (a) {};
        \path (A.east |- A.east)+(+1.6,-2.2) node (c) {};

        \path[fill=yellow!20,rounded corners, draw=black!50, dashed]
            (a) rectangle (c);

\end{pgfonlayer}
\end{tikzpicture}
    \caption{Database operation in $N$ time slots. At the beginning of each time slot, the database operator announces prices for incoming SUs. After observing the realized demands, the database operator then makes the admission decision and inform the selected SU to access. The notation $(S_n,X_n,Y_n)$ denotes the resultant channel occupancy and two SU types' demand realizations (will be explained in Subsection \ref{subsec:A}).} \label{system}
\end{figure}

To maximize the expected revenue, the database operator wants to jointly optimize spectrum prices and admissions over all $N$ time slots. In this optimization problem, the database operator's decision of admitting a heavy-traffic SU will prevent admitting a light/heavy-traffic SU (if available) in the next time slot, hence the operation decisions over time are coupled. We will model the problem as a stochastic dynamic programming problem, and propose the optimal admission policies under static pricing in Section \ref{sunsec:pf} and under dynamic pricing in Section \ref{sec:dynamicpricing}, respectively. In both sections, we allow dynamic admission decisions over time. Notice that static pricing is a special case of dynamic pricing, and is widely used in industry due to its simplicity and low complexity. Hence, we are interested in exploring the benefits that the flexible dynamic pricing may bring beyond the simplified static pricing. Then we can provide insights into which pricing and admission scheme the database operator should choose and under what conditions.

\section{Optimal Static Pricing and Dynamic Admission}\label{sunsec:pf}
We first consider the case of static pricing, where prices do not change over time. It will serve as a benchmark and help us quantify the performance gain by using dynamic pricing in Section \ref{sec:dynamicpricing}.  With static pricing, the database only needs to optimize and announce prices once at the beginning of time slot 1, and keeps the prices fixed for the rest $N-1$ time slots, i.e., $r_l(n)=r_l$ and $r_h(n)=r_h$ for each time slot $n\in\mathcal{N}$.

We will formulate the revenue maximization problem with static pricing and dynamic admission as a stochastic dynamic programming problem. In Subsections \ref{subsec:A} to \ref{sec:static}, we will formulate and solve the optimal admission control problem through backward induction, given any \emph{fixed} prices. In Subsection \ref{sec:staticpricing}, we will optimize the static prices, considering the admission policies developed in Subsections \ref{subsec:A} to \ref{sec:static}.

\subsection{Admission Control Formulation under Fixed Prices}\label{subsec:A}

Given fixed prices $r_l$ and $r_h$, we now optimize the channel admission decision in each time slot. Such optimization not only considers the channel availability and SU demands in the current time slot, but also considers SU demands in future time slots. We will formulate it as a stochastic dynamic programming problem.

We first define the system state as follows.
\begin{definition}[System State]
The system state in time slot $n$ is $(S_n,X_n,Y_n)$. Here, $S_n$ denotes the number of \emph{remaining occupied} time slots at the beginning of time slot $n$. Since $S_n\in\{0,1\}$, $S_n$ also indicates the binary channel state, where $S_n=0$ denotes that the channel is available for admission in time slot $n$, and $S_n=1$ otherwise. The parameter $X_n=1$ denotes that at least one light-traffic SU arrives at the beginning of the time slot (and is willing to pay for price $r_l$), and $X_n=0$ otherwise. The parameter $Y_n$ is defined similarly as $X_n$ but for the heavy-traffic SUs. We define the SU demand probabilities in time slot $n$ as $p_l=\Pr\{X_n=1\}$ and $p_h=\Pr\{Y_n=1\}$, respectively. As prices are unchanged over time, $p_l$ and $p_h$ are the same for all time slots.
\end{definition}

The system state changes over time, depending on the channel admission decisions and SUs' demand realizations over time. 
The feasible set of admission actions in each time slot depends on the current system state. Formally, we define the state-dependent feasible admission action set as follows.
\begin{definition}[Admission Action Set]
The set of feasible admission actions in time slot $n$ is a state-dependent set $\mathcal{A}_n(S_n,X_n,Y_n)$. When $S_n=1$, i.e., the current time slot is not available for new admission as we are still serving the heavy-traffic SU from the last time slot, we have $\mathcal{A}_n(1,X_n,Y_n)=\{0\}$ for all possible $(X_n,Y_n)$. When $S_n=0$, the admission action set depends on which type of SUs' demands in the current time slot. If no SUs request in time slot $n$ (i.e., $(X_n,Y_n)=(0,0)$), the set of actions is still $\mathcal{A}_n=\{0\}$. If both light- and heavy-traffic SUs demand, i.e., $(X_n,Y_n)=(1,1)$, then we can either serve no SU, a light-traffic SU, or a heavy-traffic SU, and thus the set of actions is $\mathcal{A}_n=\{0,1,2\}$. To summarize, 
\begin{align}\label{actionsets}
\mathcal{A}_n(0,X_n,Y_n)=\left\{
\begin{array}{lcl}
          \{0\},             & \text{if }(X_n,Y_n)=(0,0),\\
\textstyle\{0,1\},           & \text{if }(X_n,Y_n)=(1,0),\\
\textstyle\{0,2\},           & \text{if }(X_n,Y_n)=(0,1),\\
\textstyle\{0,1,2\},           & \text{if }(X_n,Y_n)=(1,1).
\end{array} \right.
\end{align}
\end{definition}

We further define the specific admission decision in time slot $n$ as $a_n\in \mathcal{A}_n(S_n,X_n,Y_n)$. 


Now we are ready to introduce the state dynamics. When $S_n=1$, we will not admit any SU, hence in the next time slot $S_{n+1}=S_n-1=0$, as the remaining occupied time slots decreases by one. When $S_n=0$, the channel availability of the next time slot only depends on the action $a_n$. If we admit the light-traffic SU with $a_n=1$, then the channel is available in the next time slot (as the remaining occupied time slot is 0), i.e., $S_{n+1}=a_n-1=0$. If we admit the heavy-traffic SU with $a_n=2$, it will occupy two time slots (time slots $n$ and $n+1$). This means that at the beginning of time slot $n+1$, we will have the number of remaining occupied time slot to be 1, i.e., $S_{n+1}=a_n-1=1$. At the beginning of time slot $n+2$, there is no SU occupying the channel, hence $S_{n+2}=S_{n+1}-1=0$ and time slot $n+2$ is available for admission. To summarize, we derive the following state dynamics.
\begin{lemma}[State Dynamics]\label{lemma1}
The dynamics of the system state component $S_{n}$ for each time slot $n\in\mathcal{N}$ satisfies the following equation:
\begin{equation}\label{evolution}
\textstyle S_{n+1} = (S_{n} + a_{n}(1-S_{n})-1)^+,
\end{equation}
where $(x)^+:=\max(0,x)$, and $S_n\in\{0,1\}$ for each $n\in \mathcal{N}$.
\end{lemma}

Lemma \ref{lemma1} captures the change of remaining occupied time slots. The system state components $(X_n,Y_n)$ are the realizations of SU demands in the current time slot, and do not depend on the action $a_n$ in previous time slots. The key notations we introduced so far are listed in Table \ref{tab:notations}. 
\begin{table}[!t]
\setlength{\tabcolsep}{1pt}
\renewcommand{\arraystretch}{1.0}
\caption{Key Notations}\label{tab:notations}
\centering
\begin{tabular}{>{\scriptsize}c>{\scriptsize}c}
\toprule
{\bf Symbols} & {\bf Physical Meaning}\\
\midrule
$\mathcal{N}=\{1,\cdots,N\}$ & Set of time slots \\
\rowcolor{mygray}
$(S_n,X_n,Y_n)$ & System state in time slot $n$ \\
$a_n(S_n,X_n,Y_n)$ and $a_n$ & Admission action in time slot $n$\\
\rowcolor{mygray}
$\mathcal{A}_n(S_n,X_n,Y_n)$ & Set of feasible admission actions in time slot $n$\\
$r(a_n)$ & Immediate revenue by the admission action $a_n$\\
\rowcolor{mygray}
$R_n(S_n,X_n,Y_n,a_n)$ & Total revenue from time slot $n$ to $N$\\
$\mathbb{E}[R^{\ast}_{n}(S_{n},X_{n},Y_{n})]$ and $\bar{R}_n^{\ast}(S_n)$ & Optimal expected future revenue from $n$ to $N$ \\
\rowcolor{mygray}
$\pi^*(S_n,X_n,Y_n)$ & Optimal admission strategy in time slot $n$\\
$\boldsymbol{\pi}^*=\{\pi^{\ast}(S_n,X_n,Y_n),n\in\mathcal{N}\}$ & Optimal admission policy for all time slots\\
\rowcolor{mygray}
$R_n^{\ast}(r_l,r_h)$ & Total revenue from $n$ to $N$ as a function of prices\\
$r_l(n),r_h(n)$ & Price for light/heavy-traffic SUs in time slot $n$\\
\rowcolor{mygray}
$\bar{R}_n(r_l(n),r_h(n))$ & Expected future revenue from time slot $n$ to $N$\\
\multirow{2}{*}{$p_l(r_l(n)),p_h(r_h(n))$} & Probabilities of having at least one light- and\\
& heavy-traffic SU requesting spectrum in $n$\\
\rowcolor{mygray}
$k_l,k_h$ & Demand elasticity of light/heavy-traffic SUs\\
$\mathcal{I}=\{1,\cdots,I\}$ & Set of SUs' types in the multiple types case\\
\rowcolor{mygray}
$X_n^{(i)},\forall i\in\mathcal{I}$ & Demand of type-$i$ SUs in time slot $n$\\
\bottomrule
\end{tabular}
\end{table}

We are now ready to introduce the revenue maximization problem. We define a policy $\boldsymbol\pi=\{a_n(S_n,X_n,Y_n),\forall n\in\mathcal{N}\}$ as the set of decision rules for all possible states and time slots, and we let $\boldsymbol\Pi=\{\mathcal{A}_n(S_n,X_n,Y_n),\forall n\in\mathcal{N}\}$ be the feasible set of $\boldsymbol\pi$. Given \emph{all possible} system state vectors $\boldsymbol{S}=\{S_n,\forall n\in\mathcal{N}\}$, $\boldsymbol{X}=\{X_n,\forall n\in\mathcal{N}\}$, and $\boldsymbol{Y}=\{Y_n,\forall n\in\mathcal{N}\}$, the database operator aims to find an optimal policy $\boldsymbol\pi^{\ast}$ (from the set of all admission policies $\boldsymbol\Pi$) that maximizes the expected total revenue from time slot 1 to $N$. Formally, we define Problem \textbf{P1} as follows.
\begin{align}\label{eq:optproblem}
&\textstyle \mbox{\textbf{P1:} \emph{Revenue Maximization by Dynamic Admission}}\notag\\ 
&\textstyle \mathrm{maximize}~~ \mathbb{E}_{\boldsymbol{X},\boldsymbol{Y}}^{\boldsymbol\pi}[R(\boldsymbol{S},\boldsymbol{X},\boldsymbol{Y},\boldsymbol\pi)]\\
&\textstyle \mathrm{subject~to}~a_n(S_n,X_n,Y_n)\in\mathcal{A}_n(S_n,X_n,Y_n),\forall n\in\mathcal{N},\\
&\textstyle~~~~~~~~~~~~~S_{n+1} = (S_{n} + a_{n}(1-S_{n})-1)^+, \forall n\in\mathcal{N}\setminus\{N\},\\
&\textstyle \mathrm{variables}~~~\boldsymbol\pi=\{a_n(S_n,X_n,Y_n),\forall n\in\mathcal{N}\},
\end{align}
where the expectation in the objective function is taken over SUs' random requests ($\boldsymbol{X},\boldsymbol{Y}$).

We proceed to analyze Problem \textbf{P1} by using backward induction \cite{Bertsekas}. After SUs' demands $X_n$ and $Y_n$ are realized in time slot $n$, the operator makes the admission action $a_n$ to maximize the total revenue by considering future SU demands. We define the total revenue from time slot $n$ to $N$ as $R_n(S_n,X_n,Y_n,a_n)$. The total revenue computed in time slot $n$ has two parts: i) the {\emph{immediate revenue}} $r(a_n)$ for the current admission action $a_n$, where $r(a_n)=0$, $r_l$, or $r_h$ if $a_n=0$, $1$, or $2$, respectively; and ii) the {\emph{expected future revenue}} from time slot $n+1$ to $N$, i.e., $\mathbb{E}[R_{n+1}(S_{n+1},X_{n+1},Y_{n+1})]$, where the expectation is taken over the SUs' possible demands in the next time slot $n+1$, i.e., $(X_{n+1},Y_{n+1})$.\footnote{In this paper, the expectation $\mathbb{E}[R_{n}^{\ast}(S_{n},X_{n},Y_{n})]$ is always taken over SU requests $(X_{n},Y_{n}),\forall n\in \mathcal{N}$, unless otherwise mentioned.} Then the optimization problem of time slot $n$ in the backward induction process is
\begin{equation}\label{valueiteration}
\textstyle R_n^\ast (S_n,X_n,Y_n) = \max_{a_n \in \mathcal{A}_n} R_n (S_{n}, X_n, Y_n, a_n),
\end{equation}
where the revenue's dynamic recursion is
\begin{equation}\label{ref:recursion}
R_n(S_{n},\!X_n,\!Y_n,\!a_n)\!=\!r({a_n})\!+\mathbb{E}[R_{n+1}^{\ast}(\!S_{n+1},\!X_{n+1},\!Y_{n+1}\!)].
\end{equation}
As a boundary condition in the last time slot $N$, we have $R^{\ast}_{N}(S_{N},X_N,Y_N,a_N)=r(a_N)$, as there is no future spectrum opportunity and revenue collection after time slot $N$.

The maximum expected revenue from time slot $n$ to $N$ is denoted by $\mathbb{E}_{X_{n},Y_{n}}[R^{\ast}_{n}(S_n,X_{n},Y_{n})]$, which is a part of the revenue and will be utilized for admission decision-making in previous time slots. Since the expectation $\mathbb{E}_{X_{n},Y_{n}}[R^{\ast}_{n}(S_n,X_{n},Y_{n})]$ is taken over all possible SU demand combinations $(X_n,Y_n)$, we rewrite it as $\bar{R}_n^{\ast}(S_n),\forall n\in \mathcal{N}$ for simplicity.
We derive the expected total revenue $R_n(S_n,X_n,Y_n,a_n)$ by adding the immediate revenue as a result of action $a_n$ and the corresponding expected future revenue $\bar{R}_{n+1}^{\ast}(S_{n+1})$ (if $a_n=0$ or $1$, i.e., no admission or admitting a light-traffic SU) or  $\bar{R}_{n+2}^\ast(S_{n+2})$ (if $a_n=2$, i.e., admitting a heavy-traffic SU), considering all possible SU demands $(X_{n},Y_{n})$ in time slot~$n$:
\begin{align} \label{expected n satge}
&~~R_n(S_n,X_n,Y_n,a_n)=(1-p_l)(1-p_h)[0+\bar{R}_{n+1}^{\ast}(S_{n+1})]\notag\\
&~~~~~~~~~~~~~~~~~~+p_l(1-p_h)[r_{l}+\bar{R}_{n+1}^{\ast}(S_{n+1})]\notag\\
&~~~~~~~~~~~~~~~~~~+(1-p_l)p_h[(0+\bar{R}_{n+1}^{\ast}(S_{n+1}))\cdot\boldsymbol{1}_{\{a_n=0\}}\notag\\
&~~~~~~~~~~~~~~~~~~+(r_h+\bar{R}_{n+2}^{\ast}(S_{n+2}))\cdot\boldsymbol{1}_{\{a_n=2\}}]\notag\\
&~~~~~~~~~~~~~~~~~~+p_lp_h[(r_l+\bar{R}_{n+1}^{\ast}(S_{n+1}))\cdot\boldsymbol{1}_{\{a_n=1\}}\notag\\
&~~~~~~~~~~~~~~~~~~+(r_h+\bar{R}_{n+2}^{\ast}(S_{n+2}))\cdot\boldsymbol{1}_{\{a_n=2\}}],
\end{align}
which can be computed according to (\ref{valueiteration}) and (\ref{ref:recursion}) recursively and backwardly from time slot $N$ to~$n$. Later, we will calculate $R_n(S_n,X_n,Y_n,a_n)$ by setting the specific values of $a_n$ in the \emph{last two terms} of (\ref{expected n satge}) according to the different admission strategies for time slot~$n$.

Next, we will solve the dynamic programming problem using (\ref{valueiteration})-(\ref{expected n satge}).

%% file: Section-Optimal.tex
\subsection{Optimal Dynamic Admission Control}\label{sec:DP}

By using backward induction \cite{Bertsekas}, we start with the final time slot $N$ and derive the optimal decisions slot by slot back. In time slot $n$, the admission decision is made by comparing the corresponding total revenues $R_n(S_n,X_n,Y_n,a_n)$ for different admission $a_n$ in time slot $n$. 

Based on the above discussions, we propose the \emph{optimal} dynamic admission control policy in Algorithm \ref{algorithm}. More specifically, this control policy $\boldsymbol\pi^*(S_n,X_n,Y_n)$ is developed by solving Problem \textbf{P1} using standard backward induction mentioned earlier. In the following Cases I-III, we formally compare the immediate revenue plus the expected future revenue to make admission decisions (i.e., $r_h+\bar{R}_{n+2}^{\ast}(S_{n+2})$, $r_l+\bar{R}_{n+1}^{\ast}(S_{n+1})$, and $0+\bar{R}_{n+1}^{\ast}(S_{n+1})$):
\begin{itemize}
\item In Case I (lines 5-6) of Algorithm \ref{algorithm}, it is more beneficial for the operator to admit a heavy-traffic SU (if it exists) than a light-traffic SU.
\item In Case II (lines 7-8) of Algorithm \ref{algorithm}, it is more beneficial for the operator to admit a light-traffic SU (if it exists) than a heavy-traffic SU.
\item In Case III (lines 9-10) of Algorithm \ref{algorithm}, it is more beneficial for the operator to only admit a light-traffic SU (if it exists).
\end{itemize}
By the principle of optimality \cite{Bertsekas}, $\boldsymbol\pi^*=\{\pi^{\ast}(S_n,X_n,Y_n),n\in\mathcal{N}\}$ is the optimal solution, as shown in the following proposition.
\begin{proposition}\label{prop:optad}
Algorithm \ref{algorithm} solves Problem \textbf{P1} and computes the optimal admission policy~$\boldsymbol\pi^*$.
\end{proposition}

The proof of Proposition \ref{prop:optad} is given in Appendix B. Note that the optimal policy $\boldsymbol\pi^*$ is a contingency plan, which contains the optimal admission policy in each time slot $ n\in \mathcal{N} $ for any system state. After deriving the optimal policy, we can implement the policy forwardly from time slots 1 to $N$, after observing SUs' demand realizations. Furthermore, the optimal admission policy $\boldsymbol\pi^*$ may change over time, since the revenue values of Cases I-III depend on both the prices and the price-dependent demand probabilities.

\begin{algorithm} [t]
\caption{Optimal Admission Control Policy}\label{algorithm}
\begin{algorithmic}[1]
\algsetup{linenosize=\scriptsize}
\scriptsize
\STATE Set $n=N,\bar{R}_{N+1}^*=0$
\STATE The optimal admission for $N$ is $ a_N^*=X_N $ and $\bar{R}_N^*=p_lr_l$
\FOR{$n=N-1,\cdots,2,1$}
   \STATE Calculate $\bar{R}_{n+1}^{\ast}(S_{n+1})$ using (\ref{expected n satge}).
   \IF {$r_h+\bar{R}_{n+2}^{\ast}(S_{n+2})\geq r_l+\bar{R}_{n+1}^{\ast}(S_{n+1})$}
   \STATE if $Y_n=1$, then $a_n=2$; if $Y_n=0,X_n=1$, then $a_n=1$; otherwise $a_n=0$.
   \ELSIF {$\bar{R}_{n+1}^{\ast}(S_{n+1})<r_h +\bar{R}_{n+2}^{\ast}(S_{n+2})<r_l+\bar{R}_{n+1}^{\ast}(S_{n+1})$}
     \STATE if $X_n=1$, then $a_n=1$; if $X_n=0,Y_n=1$, then $a_n=2$; otherwise $a_n=0$.
   \ELSE
    \STATE  if $X_n=1$, then $a_n=1$; otherwise $a_n=0$.
   \ENDIF
\ENDFOR
\RETURN the optimal admission policy $\boldsymbol\pi^*$
\end{algorithmic}
\end{algorithm}

%% file: Section-Static.tex
\subsection{Stationary Admission Policies}\label{sec:static}
The optimal admission control solution in Algorithm~\ref{algorithm} does not have a closed-form characterization and the system still needs to check a huge-size table created from the algorithm after knowing the realizations of SU random demands. This motivates us to focus on a class of low complexity stationary admission policies, where the admission rules do not change over time (while the actual admission decisions might change over time). We will characterize the conditions under which these stationary admission policies are optimal.

Recall that there are three possible admission strategies in each time slot, depending on the values of $r_h+\bar{R}_{n+2}^{\ast}(S_{n+2})$, $r_l+\bar{R}_{n+1}^{\ast}(S_{n+1})$, and $0+\bar{R}_{n+1}^{\ast}(S_{n+1})$. For a particular time slot $n$, for example, if $r_h+\bar{R}_{n+2}^{\ast}(S_{n+2})>r_l+\bar{R}_{n+1}^{\ast}(S_{n+1})>0+\bar{R}_{n+1}^{\ast}(S_{n+1})$, we prefer to serve the heavy-traffic SU type rather than the light-traffic one or not serving anyone (i.e., the admission priority follows $\Lambda(2)>\Lambda(1)>\Lambda(0)$). Here, we define the function $\Lambda(a_n)$ to capture the priority order of the admission action $a_n\in\{0,1,2\}$. Due to the fact $r_l+\bar{R}_{n+1}^{\ast}(S_{n+1})>0+\bar{R}_{n+1}^{\ast}(S_{n+1})$ and serving a light-traffic SU is better than serving no one, there are a total of three reasonable admission priority orders, i.e., $\Lambda(2)>\Lambda(1)>\Lambda(0)$, $\Lambda(1)>\Lambda(2)>\Lambda(0)$, and $\Lambda(1)>\Lambda(0)>\Lambda(2)$. We discuss them one by one next.

Table \ref{Case I} shows the three stationary policies that we will discuss. Recall that when $S_n=1$ (i.e., channel is still occupied in the current time slot), we have $a_n^\ast=0$ (not admitting any SU) for any values of $X_n$ and $Y_n$. Table \ref{Case I} only focuses on the case of $S=0$. The three rows/sub-tables, namely, Tab.\ref{Case I}\textendash$HP$: $a_n^{HP\ast}$, Tab.\ref{Case I}\textendash$LP$: $a_n^{LP\ast}$, and Tab.\ref{Case I}\textendash$LD$: $a_n^{LD\ast}$, represent the Heavy-Priority (i.e., $\Lambda(2)>\Lambda(1)>\Lambda(0)$), Light-Priority (i.e., $\Lambda(1)>\Lambda(2)>\Lambda(0)$), and Light-Dominant (i.e., $\Lambda(1)>\Lambda(0)>\Lambda(2)$) admission policies, respectively. 
For each policy, we will derive the conditions of the static prices $r_l$ and $r_h$, under which the policy achieves the optimality of Problem~\textbf{P1}.
\begin{table}[t]
\renewcommand{\arraystretch}{1.1}
\caption{Three Stationary Admission Policies}\label{stationarypolicy}
\label{Case I}
\centering
\begin{tabular}{|c|c|c|c|c|}
\hline
\multirow{2}{*}{Admission Policies} 
&\multicolumn{4}{c|}{System states $(S_{n},X_n,Y_n)$} \\
\cline{2-5}
&$(0,0,0)$ & $(0,0,1)$ & $(0,1,0)$ & $(0,1,1)$ \\ 
\hline
\hline
Tab.II\textendash$HP$: $a_n^{HP\ast}$ & $0$ & $2$ & $1$ & $2$\\ 
\hline
\hline
Tab.II\textendash$LP$: $a_n^{LP\ast}$ & $0$ & $2$ & $1$ & $1$\\ 
\hline
\hline
Tab.II\textendash$LD$: $a_n^{LD\ast}$ & $0$ & $0$ & $1$ & $1$\\ 
\hline
\end{tabular}
\end{table}

We first analyze the Heavy-Priority admission policy (in Tab.\ref{Case I}\textendash$HP$: $a_n^{HP\ast},\forall n\in \mathcal{N}$). Under this policy, we will serve a heavy-traffic SU ($a_n=2$) whenever possible  ($Y_n=1$), and only serve a light-traffic SU ($a_n=1$) when there is only a light-traffic SU ($X_n=1$ and $Y_n=0$).\footnote{Notice that such discussion is only meaningful for time slot 1 to $N-1$, as in the last time slot $N$ we will always admit a light-traffic SU whenever possible.}  Such a stationary policy is optimal if the following two conditions hold for each and every time slot $ n\in \{1,\cdots,N-1\} $,
\begin{gather}\label{eq:con1}
r_h+\bar{R}_{n+2}^{\ast}(0)\geq 0+\bar{R}_{n+1}^{\ast}(0),\\
\label{eq:con2}
r_h +\bar{R}_{n+2}^{\ast}(0)\geq r_l+\bar{R}_{n+1}^{\ast}(0).
\end{gather}
Inequality (\ref{eq:con1}) shows that serving a heavy-traffic SU who occupies two consecutive time slots leads to a higher expected total revenue than serving no SU in the current time slot. Inequality (\ref{eq:con2}) shows that serving a heavy-traffic SU leads to a higher expected total revenue than serving a light-traffic SU in the current time slot. Since (\ref{eq:con2}) ensures (\ref{eq:con1}), we only need to consider (\ref{eq:con2}). 

Similarly, we can derive the condition under which the Light-Priority admission policy (in Tab.\ref{Case I}\textendash$LP$: $a_n^{LP\ast}$) is optimal, i.e.,  $0+\bar{R}_{n+1}^{\ast}(0)<r_h+\bar{R}_{n+2}^{\ast}(0)<r_l+\bar{R}_{n+1}^{\ast}(0)$ for all $n\in\{1,\cdots,N-1\}$. Under this policy, we will admit a light-traffic SU whenever possible ($X_n=1$), and admit a heavy-traffic SU otherwise ($X_n=0$ and $Y_n=1$). Finally, we can derive the condition under which the Light-Dominant admission policy (in Tab.\ref{Case I}\textendash$LD$: $a_n^{LD\ast}$) is optimal, i.e.,  $r_h+\bar{R}_{n+2}^{\ast}(0)\leq 0+\bar{R}_{n+1}^{\ast}(0)$ for all $n\in\{1,\cdots,N-1\}$. Under this policy, we will choose to admit a light-traffic SU ($a_n=1$) whenever possible ($X_n=1$), and will never admit any heavy-traffic SU, as it leaves no room to accept a light-traffic SU in the next time slot.

To summarize the above analysis, we have the following theorem. Recall that $r_h/r_l$ denotes the ratio between the prices charged to the heavy-traffic and the light-traffic SUs, and $p_l$ and $p_h$ are the demand probabilities defined in Subsection \ref{subsec:A}.
\begin{theorem}\label{theorem:allocation} 
A stationary admission policy becomes the optimal policy to solve Problem \textbf{P1} if one of the following condition is true:
\begin{itemize}
\item The Heavy-Priority admission policy $a_n^{HP\ast}$ in Tab.\ref{Case I}\textendash$HP$ for all $n\in\mathcal{N}$ is optimal if $r_h/r_l\geq2p_l+(1-p_l)/(1-p_h)$.
\item The Light-Priority admission policy $a_n^{LP\ast}$ in Tab.\ref{Case I}\textendash$LD$ for all $n\in\mathcal{N}$ is optimal if $p_l\leq r_h/r_l\leq 1+p_l$.
\item The Light-Dominant admission policy $a_n^{LD\ast}$ in Tab.\ref{Case I}\textendash$LP$ for all $n\in\mathcal{N}$ is optimal if $r_h/r_l<p_l$.
\end{itemize}
\end{theorem}

The proof of Theorem \ref{theorem:allocation} is given in Appendix C. The theorem shows that each of the three stationary policies is optimal within a particular range of the price ratio $r_h/r_l$.
Fig. \ref{big picture} illustrates the results of Theorem \ref{theorem:allocation} graphically. In this figure, we divide the feasible range of the price ratio $r_h/r_l$ into four regimes, among which in three regimes (I, II, and IV) the stationary policies are optimal. We are able to further characterize the closed-form optimal total revenues for these three regimes, and the details can be found in Appendix C. 
It is clear that a larger value of $r_h/r_l$ gives a higher preference to the admission of a heavy-traffic SU. In regime III, we have to use Algorithm \ref{algorithm} to compute the optimal admission policy.

\begin{figure}
   \centering
\usetikzlibrary{arrows}
\usetikzlibrary{decorations.pathreplacing}
\begin{tikzpicture}
\tikzstyle{line} = [draw, -latex']
    \begin{scope}[thick,font=\scriptsize]
    \draw [->,thick] (-4,0) -- (4.3,0) node [below] {$\frac{r_h}{r_l}$} node [above] {$\infty$};

    \draw (1.8,-3pt) -- (1.8,3pt)   node [above] {$2p_l+\frac{1-p_l}{1-p_h}$};
    \draw (-1.8,-3pt) -- (-1.8,3pt) node [above] {$p_l$};
    \draw (0.3,-3pt) -- (0.3,3pt) node [above] {$1+p_l$};
    \draw (-4,-3pt) -- (-4,3pt) node [above] {$0$};
    \end{scope}
\draw[decorate, decoration={brace, mirror}, very thick, red] (-4.0,-0.3) -- (-1.8,-0.3);
\draw (-2.3,-0.6) node {{\tiny \textcolor{red}{I: Light-Dominant Admission Policy}}};
\draw[decorate, decoration={brace}, very thick, blue] (-1.8,0.5) -- (0.3,0.5);
\draw (-1,0.8) node {{\tiny \textcolor{blue}{II: Light-Priority Admission Policy}}};
\draw[decorate, decoration={brace, mirror}, very thick, black] (0.3,-0.3) -- (1.8,-0.3);
\draw (1.1,-0.6) node {{\tiny III: Algorithm \ref{algorithm}}};
\draw[decorate, decoration={brace}, very thick, DGreen] (1.8,0.7) -- (4.3,0.7);
\draw (3,1) node {{\tiny \textcolor{DGreen}{IV: Heavy-Priority Admission Policy}}};
\end{tikzpicture}
\caption{Optimal stationary admission policies for all price ratio $r_h/r_l$ values (regimes I, II, and IV). 
} \label{big picture}
\end{figure}

After analyzing the optimal admission control decisions from time slot $N$ to 1 in the backward induction, we now optimize the initial pricing decision at the beginning of time slot 1.

\subsection{Optimal Static Pricing}\label{sec:staticpricing}
Under static pricing, the database operator optimizes and announces the prices $r_h$ and $r_l$ in time slot 1, and do not change these prices for the remaining $N-1$ time slots. As explained in Section \ref{sec:spectrumallocation}, we consider the general case where prices will affect SU demands during each time slot. 
As a concrete example, we consider the widely used linear demand function in economics \cite{mas-colell}, where the probability of an SU of type $i\in\{l,h\}$ requesting the spectrum resource in a time slot is
$p_i(r_i)=1-k_ir_i$, where $0\leq r_i\leq r_i^{\max}=1/k_i$.\footnote{Changing to some common nonlinear functions are unlikely to change the key results. This is because the optimal static pricing can be solved in Proposition 2, even for nonlinear demand functions, we can still search the optimal static pricing.} The parameters $k_l$ and $k_h$ characterize the demand elasticity of the light-traffic and the heavy-traffic SUs, respectively, and larger values of $k_l$ and $k_h$ reflect higher price sensitivities.\footnote{In practice, the price elasticity parameters can be estimated according to the market survey or historical data about demand responses (e.g., \cite{estimation}). By doing multiple independent repeated trials, the operator can estimate the demand elasticities.}

By using the three stationary admission policies in Theorem \ref{theorem:allocation}, we are able to derive three closed-form objective ${R}_{n}^{\ast},n\in\mathcal{N}$ as a function of prices $r_l$ and $r_h$. Next we optimize the prices that maximize the total revenue $R_1^{\ast}$ in Problem \textbf{P1}.
\begin{proposition}\label{theoremstatic}
Consider the case $r_h/r_l\geq2p_l+(1-p_l)/(1-p_h)$, in which the heavy-priority admission policy is optimal as shown in Theorem \ref{theorem:allocation}. The optimal static pricing $(r_l^{\ast},r_h^{\ast})$ is the optimal solution to the following problem
\begin{align}
&\mathrm{maximize}~~{R}_{1}^{\ast}(r_l,r_h),\\
&\mathrm{subject~to}~~ r_h/r_l\geq2p_l+(1-p_l)/(1-p_h),\\
&~~~~~~~~~~~~~~  0\leq r_l\leq r_l^{\max},0\leq r_h\leq r_h^{\max},\\
&\mathrm{variables}~~~r_l,r_h,
\end{align}
where 
\begin{align}\label{eq:closedformrevenue}
&\textstyle R_1^{\ast}(r_l,r_h)=N\left(\frac{p_lr_l+(1-p_l)p_hr_h}{1-(p_lp_h-p_h)}\right)\notag\\
&\textstyle~~~~~~+\left(\frac{(p_lp_h-p_h)(r_h-p_lr_l)}{1-(p_lp_h-p_h)}\right)\frac{(p_lp_h-p_h)(1-(p_lp_h-p_h)^{N})}{1-(p_lp_h-p_h)}.
\end{align}
\end{proposition}

The proof of Proposition \ref{theoremstatic} is given in Appendix D. The same conclusion holds for the other two cases shown in  Theorem \ref{theorem:allocation}, and the details are provided in Appendix C. The function ${R}_{1}^{\ast}(r_l,r_h)$ turns out to be non-convex in general, and the optimal prices cannot be solved in closed form. However, notice that the key benefit of static pricing is that it does not need to be recomputed and updated frequently over time, thus we can compute the optimal static prices \emph{offline} once and the high computational complexity is not a major practical issue.

%% file: Section-DynamicPricing.tex
\section{Optimal Dynamic Pricing and Dynamic Admission}\label{sec:dynamicpricing}
In Section \ref{sunsec:pf}, we have considered the static pricing and dynamic admission control problem. Now we consider the case of dynamic pricing, where the prices vary over time. In the following, we will formulate the dynamic pricing and dynamic admission control problem, aiming at deriving the optimal dynamic pricing and admission policies.
\subsection{Dynamic Pricing-and-Admission Problem Formulation}
Now we further study the general case of dynamic pricing, where the database operator has the flexibility of changing prices over time. 
The database operator's goal is to compute the optimal prices $\boldsymbol{r}_l^{\ast}=\{r_l^{\ast}(n),n\in\mathcal{N}\}$ and $\boldsymbol{r}_h^{\ast}=\{r_h^{\ast}(n),n\in\mathcal{N}\}$, and the optimal admission policy $\boldsymbol{\pi}^{\ast}=\{a_n^{\ast}(S_n,X_n,Y_n),n\in\mathcal{N}\}$ for all time slots and system states to maximize its expected revenue, i.e.,
\begin{align}
&\textstyle \mbox{\textbf{P2:} \emph{Joint Dynamic Pricing and Dynamic Admission}}\notag\\
&\mathrm{maximize}~~ \mathbb{E}_{\boldsymbol{X},\boldsymbol{Y}}^{\boldsymbol\pi}[R(\boldsymbol{S},\boldsymbol{X},\boldsymbol{Y},\boldsymbol\pi,\boldsymbol{r}_l,\boldsymbol{r}_h)]\\
&\textstyle 
\mathrm{subject~to}~a_n(S_n,X_n,Y_n)\!\in\!\mathcal{A}_n(S_n,X_n,Y_n),\!\forall n\in\mathcal{N},\\
&\textstyle~~~~~~~~~~~~~S_{n+1} = (S_{n} + a_{n}(1-S_{n})-1)^+,\forall n\in\mathcal{N}\setminus\{N\},\\
&~~~~~~~~~~~~~ 
0\leq r_l(n)\leq r_l^{\max}, \forall n\in\mathcal{N},\\
&~~~~~~~~~~~~~0\leq r_h(n)\leq r_h^{\max}, \forall n\in\mathcal{N},\\
&\mathrm{variables}~~\{\boldsymbol\pi,\boldsymbol{r}_l,\boldsymbol{r}_h\}.
\end{align}

We can again use backward induction to solve Problem \textbf{P2} in each time slot. Different from Section~\ref{sunsec:pf}, we need to jointly determine the prices and the admission decisions from time slot $N$ to 1.
The subproblem in each time slot $n\in\mathcal{N}$ is
\begin{align}
&\textstyle \mbox{\textbf{P3:} \emph{Pricing-and-Admission Subproblem in time slot $n$}}\notag\\
&\mathrm{maximize}~~~R_n\big(r_l(n),r_h(n),a_n(S_n,X_n,Y_n)\big)\\
&\mathrm{subject~to}~~~a_n(S_n,X_n,Y_n)\in\mathcal{A}_n(S_n,X_n,Y_n),\\
&~~~~~~~~~~~~~~~0\leq r_l(n)\leq r_l^{\max},\\
&~~~~~~~~~~~~~~~0\leq r_h(n)\leq r_h^{\max},\\
&\mathrm{variables}~~~~\{a_n(S_n,X_n,Y_n),r_l(n),r_h(n)\}.
\end{align}
The expression of $R_n$ can be similarly derived by using the derivation procedure of (\ref{expected n satge}), except that the demand probabilities $p_l$ and $p_h$ are functions the prices $r_l(n)$ and $r_h(n)$, i.e., $p_l(r_l(n))$ and $p_h(r_h(n))$, respectively. The key challenge of solving Problem \textbf{P3} is the coupling between the pricing and admission decisions in each time slot. Next we will propose a decomposition approach for the two decisions that helps us solve Problem  \textbf{P3} in each time slot $n$.
\subsection{Decomposition of Pricing and Admission in Each Time Slot}

\begin{table}[t]
\setlength{\tabcolsep}{0.5pt}
\renewcommand{\arraystretch}{1.1}
\caption{Three Admission Strategies in Time Slot $n$}
\label{table: allocations}
\centering
\begin{tabular}{cc}
\toprule
{\bf Admission Strategies in time slot $n$}
&{\bf Conditions}\\
\midrule
Heavy-Priority Strategy (HP): &\\
$a_{n}=(2-X_{n})Y_{n}+X_{n}$ & \multirow{-2}*{$r_h(n)+\bar{R}_{n+2}^{\ast}\!\geq\! r_l(n)+\bar{R}_{n+1}^{\ast}$}\\
\hline
Light-Priority Strategy (LP): $a_{n}=$ &\\
$X_{n}\cdot\mathbf{1}_{\{Y_{n}=0\}}+(2\!-\!X_{n})\cdot\mathbf{1}_{\{Y_{n}=1\}}$ & \multirow{-2}*{$0\!+\!\bar{R}_{n+1}^{\ast}\!<\!r_h(n)\!+\!\bar{R}_{n+2}^{\ast}\!<\! r_l(n)\!+\!\bar{R}_{n+1}^{\ast}$}\\
\hline
Light-Dominant Strategy (LD): $a_{n}\!\!=\!\! X_{n}$ &$r_h(n)+\bar{R}_{n+2}^{\ast}\leq 0+\bar{R}_{n+1}^{\ast}$\\
\bottomrule
\end{tabular}
\end{table}

First we want to clarify the difference between an \emph{admission strategy} and an \emph{admission policy}. An admission strategy specifies the admission actions for a particular time slot $n$, while an admission policy applies to all time slots in $\mathcal{N}$ (e.g., those in Table~ \ref{stationarypolicy}). Here we will focus on the admission strategy, as we only study Problem \textbf{P3} for a particular time slot $n$.

Next we consider all possible admission control strategies for a time slot $n$, as shown in Table~\ref{table: allocations}. In this table, HP stands for heavy-priority strategy, LP stands for light-priority strategy, and LD stands for light-dominant strategy. Each strategy is accompanied by a condition of the total revenue from time slot $n$ to $N$. The strategy is optimal for time slot $n$ if the corresponding condition holds.

Let us take the heavy-priority strategy (HP) as an example to explain our decomposition approach. In this strategy, we will serve a heavy-traffic SU ($a_n=2$) whenever possible  ($Y_n=1$), and only serve a light-traffic SU ($a_n=1$) if there is no heavy-traffic SU ($X_n=1$ and $Y_n=0$). Summarizing these cases together, the decision under the heavy-priority strategy can be written as $a_{n}=(2-X_{n})Y_{n}+X_{n}$. The corresponding condition for the heavy-priority strategy in Table \ref{table: allocations} shows that the total revenue of admitting a heavy-traffic SU is no less than that of admitting a light-traffic SU. The conditions for the other two admission strategies (LP and LD) can be derived similarly.

Using the result in Table \ref{table: allocations}, we can solve Problem \textbf{P3} in the following two steps:
\begin{itemize}
\item \emph{Price optimization under a chosen admission strategy:} Assume  that one of the three admission strategies in Table \ref{table: allocations} will be used in time slot $n$, we optimize prices $r_l(n)$ and $r_h(n)$ to maximize the expected total revenue.
\item \emph{Admission strategy optimization:} Compare the maximized expected total revenues (from slot $n$ to $N$) under the three admission strategies with the optimized prices, and pick the best admission strategy and pricing combination that leads to the largest revenue.
\end{itemize}

Notice that the above decomposition method is for each time slot $n\in\mathcal{N}$. The above decomposition procedure guarantees that we obtain the optimal solution of the joint problem \textbf{P3} for the following reason. First, the three possible admission strategies in each time slot are exhaustive and mutually exclusive, in the sense that the optimal pricing decision in time slot $n$ guarantees that there is only one strategy that is optimal to adopt in this time slot, depending on the conditions in Table \ref{table: allocations}. Second, if one-out-of-the-three admission strategies is optimal to adopt in time slot $n$, there must exist an associated optimal pricing accordingly that maximizes the total revenue. We thus conclude that the two-step decomposition procedure is guaranteed to solve Problem \textbf{P3} optimally.

Next, we will derive the closed-form optimal pricing under each of the three admission strategies, respectively. We will conduct the admission strategy optimization in Subsection \ref{dynamicpricingalg}.

\subsubsection{~~~~~~Optimal Pricing under Heavy-Priority Strategy}\label{sec:HPS}

Given HP strategy chosen in time slot $n$, we derive the \emph{expected total revenue} $R_n^{HP}\big(r_l^{HP}(n),r_h^{HP}(n)\big)$ by setting $a_n=2$ and $a_n=2$ in the last two terms of (\ref{expected n satge}), respectively,
where the probabilities $p_l(r_l^{HP}(n)$ and $p_h(r_h^{HP}(n))$ can also be modeled as the linear function in Subsection \ref{sec:staticpricing}.\footnote{The linear function only contributes to the explicit expression of the optimal pricing solution. From the proof of Proposition 3, we know that the optimal pricing can be derived with general demand functions, but not in terms of closed form.} Notice that the database operator may not always use HP in future time slots.

In order to optimize the prices, the database operator needs to solve the following problem.
\begin{align}
   &\textstyle \mbox{\textbf{P4:} \emph{Optimal Pricing for time slot $n$ under HP}}\notag\\
   &\mathrm{maximize}~~  R_n^{HP}\big(r_l^{HP}(n),r_h^{HP}(n)\big)\\
   &\mathrm{subject~to}~~  r_h^{HP}(n)+\bar{R}_{n+2}^{\ast}\geq r_l^{HP}(n)+\bar{R}_{n+1}^{\ast},\label{P4con1}\\
   &~~~~~~~~~~~~~~0\leq r_l^{HP}(n)\leq r_l^{\max},\\
   &~~~~~~~~~~~~~~0\leq r_h^{HP}(n)\leq r_h^{\max},\\
   &\mathrm{variables}~~~  r_l^{HP}(n), r_h^{HP}(n).
\end{align}
Constraint (\ref{P4con1}) guarantees that the heavy-priority admission strategy is optimal in time slot $n$, where $\bar{R}_{n+2}^{\ast}$ and $\bar{R}_{n+1}^{\ast}$ are determined by the optimal solutions to Problem \textbf{P3} in time slots $n+2$ and $n+1$. Since the optimization problem \textbf{P4} is a continuous function over a compact feasible set, the maximum is guaranteed to be attainable. It is easy to show that Problem \textbf{P4} is not a convex optimization problem due to the three-order polynomial objective function. Thus, a solution satisfying KKT conditions may be either a local optimum or a global optimum of Problem \textbf{P4}. Hence, we need to find all solutions satisfying KKT conditions, and then compare these solutions to find the global optimum.

\begin{figure}[!t]
\centering
\begin{tikzpicture}[scale=1.0]
    \draw [<->,thick] (0,2.6) node (yaxis) [above] {$r_h^{HP}(n)$}
        |- (5.0,0) node (xaxis) [right] {$r_l^{HP}(n)$};
    \draw (1,0) coordinate (a_1) -- (1,0) coordinate (a_2);
    \draw (0,0.3) coordinate (a_1) -- (1.5,1.8) coordinate (b_2) node[right, text width=14em] {{\scriptsize $r_h^{HP}(n)=r_l^{HP}(n)+\bar{R}^{\ast}_{n+1}-\bar{R}^{\ast}_{n+2}$}};
    \draw (0,1.4) coordinate (a_2) -- (2,1.4) coordinate (b_3) node[right, text width=20em]{}; 
    \draw[thick,dashed] (0,2.2) coordinate (a_2) -- (2,2.2) coordinate (b_3) node[right, text width=20em]{}; 
    \draw (1.3,0) coordinate (a_2) -- (1.3,2.6) coordinate (b_3) node[right, text width=10em] {};
    \draw (0,1.4)  node[left] {{\scriptsize small} $r_h^{\max}$};
    \draw (0,2.2)  node[left] {\textcolor{blue}{{\scriptsize large} $r_h^{\max}$}};
    \draw (1.3,0)  node[below]  {$r_l^{\max}$};
    \draw (0,0)  node[below] {$0$};
    \draw (0.5,1.1) node {$\mathcal{F}$};
    \path[draw,thick,fill=DGreen](0,0.3)--(0,1.4)--(1.1,1.4)--cycle;
    \path[draw,thick,fill=blue!60](0,1.4)--(0,2.2)--(1.3,2.2)--(1.3,1.6)--(1.1,1.4)--cycle;
      \fill[red] (0,0.3) circle (2pt);
      \fill[red] (0,1.4) circle (2pt);
      \fill[red] (1.1,1.4) circle (2pt);
      \fill[yellow!70] (0.3,0.9) circle (2pt);
      \fill[black] (0.56,1.7) circle (2pt);
      \fill[red] (0,2.2) circle (2pt);
      \fill[red] (1.3,2.2) circle (2pt);
      \fill[red] (1.3,1.6) circle (2pt);
      \node(c) at (2.0,0.8) [circle,draw,fill=blue!20] {$\mathcal{F}$};
      \draw[yellow!70,<-,thick] (0.3,0.9) .. controls +(1,1) and +(-1,-1) .. (c);
      \draw[black,<-,dashed,thick] (0.56,1.7) .. controls +(1,1) and +(-1,-1) .. (c);
\end{tikzpicture}
\caption{The feasible region $\mathcal{F}$ of Problem \textbf{P4}. There are two possible shapes according to the values of $r_h^{\max}$. The interior points, the red dots, and the line segments between them are all possible solutions.}
\label{feasibleregion}
\end{figure}

We will first examine the feasible region of Problem \textbf{P4} based on any possible prices $r_l^{HP}(n)$ and $r_h^{HP}(n)$. It turns out that the feasible region is a polyhedron in a two-dimensional plane. Fig. \ref{feasibleregion} shows the feasible region. According to the value of $r_h^{\max}$, the feasible region has two possible cases. The optimal solution can only be either the interior points inside the feasible region or the extreme points on the boundary. As such, we only need to check whether all the possible extreme points and the interior points satisfying KKT conditions are local optima. We skip the details (which can be found in Appendix E) due to space limit, and summarize the optimal pricing results in the following proposition.
\begin{proposition}\label{prop:HTA}
The optimal pricing in time slot $n$ under the HP strategy is summarized in Table \ref{tab:longdynamic}, which depends on the values of $\bar{R}_{n+1}^{\ast}-\bar{R}_{n+2}^{\ast}$ and $k_h/k_l$. The closed-form optimal pricing solutions in Table \ref{tab:longdynamic} are given as follows, respectively,
\begin{align}
I_0^{HP}:\left\{
\begin{array}{l}
\textstyle r_l^{HP}(n)=\frac{1}{2k_l}\\
\textstyle r_h^{HP}(n)=\frac{\left(\frac{1}{4k_l}+\frac{1}{k_h}+\bar{R}^{\ast}_{n+1}-\bar{R}^{\ast}_{n+2}\right)}{2}
\end{array} \right.,\notag
\end{align}
\begin{align}
E_2^{HP}:\left\{
\begin{array}{l}
\textstyle r_l^{HP}(n)=\frac{1}{k_h}+\bar{R}^{\ast}_{n+2}-\bar{R}^{\ast}_{n+1}\\
\textstyle r_h^{HP}(n)=\frac{1}{k_h}
\end{array} \right.,\notag
\end{align}
\begin{align}
\text{ and } E_1^{HP}:
\left\{
\begin{array}{l}
\textstyle r_l^{HP}(n)=\frac{-(\bar{R}^{\ast}_{n+1}-\bar{R}^{\ast}_{n+2})+\sqrt{(\bar{R}^{\ast}_{n+1}-\bar{R}^{\ast}_{n+2})^2+\frac{3}{k_lk_h}}}{3}\\
\textstyle r_h^{HP}(n)=\frac{2(\bar{R}^{\ast}_{n+1}-\bar{R}^{\ast}_{n+2})+\sqrt{(\bar{R}^{\ast}_{n+1}-\bar{R}^{\ast}_{n+2})^2+\frac{3}{k_lk_h}}}{3}
\end{array} \right..\notag
\end{align}
\end{proposition}

The proof of Proposition \ref{prop:HTA} is given in Appendix E. In Table \ref{tab:longdynamic}, $I_0^{HP}$, $E_1^{HP}$, and $E_2^{HP}$ represent the unique optimal solution in different cases (i.e., one interior point solution and two extreme point solutions). 
``N/A'' represents the cases where the combinations of conditions are infeasible.
For example, when $4/3\leq k_h/k_l<3$, it follows that $\bar{R}_{n+1}^{\ast}-\bar{R}_{n+2}^{\ast}>(4k_l-3k_h)/(4k_hk_l)$; when $k_h/k_h\geq3$, we have $\bar{R}_{n+1}^{\ast}-\bar{R}_{n+2}^{\ast}\geq(2-\sqrt{1+k_h/k_l})/k_h$. Hence, the corresponding cell is labeled as ``N/A''.

Tables \ref{tab:longdynamic} shows the optimal dynamic pricing in each time slot $n$ under the HP strategy. Given the demand elasticities $k_l$ and $k_h$, the solution will be uniquely given by one of the three cases of $\bar{R}^{\ast}_{n+1}-\bar{R}^{\ast}_{n+2}$ regimes. In Subsection \ref{dynamicpricingalg}, we will propose an algorithm to compute $\bar{R}_{n+1}^\ast - \bar{R}_{n+2}^\ast$ iteratively for all time slots.

\begin{table}[!t]
\setlength{\tabcolsep}{2pt}
\renewcommand{\arraystretch}{1.1}
\caption{Optimal Pricing under Heavy-Priority Strategy}
\label{tab:longdynamic}
\centering
\begin{tabular}{|c|c|c|c|}
\hline
\multirow{2}{*}{} &\multicolumn{3}{c|}{$\bar{R}^{\ast}_{n+1}-\bar{R}^{\ast}_{n+2}$}\\
\cline{2-4}
 &$\leq\frac{4k_l-3k_h}{4k_hk_l}$ & $\left(\frac{4k_l-3k_h}{4k_hk_l},\frac{2-\sqrt{1+\frac{k_h}{k_l}}}{k_h}\right)$ & $\geq\frac{2-\sqrt{1+\frac{k_h}{k_l}}}{k_h}$\\
\hline
$\frac{k_h}{k_l}<\frac{4}{3}$ & $I_0^{HP}$ & $E_1^{HP}$ & $E_2^{HP}$\\
\hline
$\frac{4}{3}\leq \frac{k_h}{k_l}<3$ & N/A  & $E_1^{HP}$ & $E_2^{HP}$\\
\hline
$\frac{k_h}{k_l}\geq3$ & N/A & N/A  & $E_2^{HP}$\\
\hline
\end{tabular}
\end{table}
\subsubsection{~~~~~~Optimal Pricing under Light-Priority Strategy}\label{sec:dyLPS}
Given LP strategy chosen in time slot $n$, we derive the expected total revenue $R_n^{LP}\big(r_l^{LP}(n),r_h^{LP}(n)\big)$ by setting $a_n=2$ and $a_n=1$ in the last two terms of (\ref{expected n satge}), respectively.
The database operator needs to solve the following problem.
   \begin{align}
   &\textstyle \mbox{\textbf{P5:} \emph{Optimal Pricing for time slot $n$ under LP}}\notag\\
   &\mathrm{maximize}~~~ R_n^{LP}\big(r_l^{LP}(n),r_h^{LP}(n)\big)\\
   &\mathrm{subject~to}~~~r_h^{LP}(n)+\bar{R}^{\ast}_{n+2}\leq r_l^{LP}(n)+\bar{R}^{\ast}_{n+1},\label{P5con1}\\
   &~~~~~~~~~~~~~~~r_h^{LP}(n)+\bar{R}^{\ast}_{n+2}\geq \bar{R}^{\ast}_{n+1},\label{P5con2}\\
   &~~~~~~~~~~~~~~~0\leq r_l^{LP}(n)\leq r_l^{\max},\\
   &~~~~~~~~~~~~~~~0\leq r_h^{LP}(n)\leq r_h^{\max},\\
   &\mathrm{variables}~~~~r_l^{LP}(n),r_h^{LP}(n).
   \end{align}
Constraints (\ref{P5con1}) and (\ref{P5con2}) guarantee that the light-priority strategy is optimal in time slot $n$.

The analysis for Problem \textbf{P5} is similar to that for Problem \textbf{P4}, due to the similar structures of the two problems. We thus have Proposition \ref{prop:MTS} as follows.
\begin{proposition}\label{prop:MTS}
The optimal solution to Problem \textbf{P5} can also be summarized in a table as in Table \ref{tab:longdynamic}, only with different conditions in the rows and the columns and expressions of $I_0^{LP}$, $E_1^{LP}$ and $E_2^{LP}$.
\end{proposition}

Due to space limitation, the proof of Proposition~\ref{prop:MTS} and the detailed solutions can be found in Appendix F.
\setlength{\tabcolsep}{2pt}

\subsubsection{~~~~~~Optimal Pricing under Light-Dominant Strategy}
Given LD strategy chosen in time slot $n$, we derive the expected total revenue $R_n^{LD}\big(r_l^{LD}(n),r_h^{LD}(n)\big)$ by setting $a_n=0$ and $a_n=1$ in the last two terms of (\ref{expected n satge}), respectively.
The database operator needs to solve the following problem.
   \begin{align}
   &\textstyle \mbox{\textbf{P6:} \emph{Optimal Pricing for time slot $n$ under LD}}\notag\\
   &\mathrm{maximize}~~~ R_n^{LD}\big(r_l^{LD}(n),r_h^{LD}(n)\big)\\
   &\mathrm{subject~to}~~~r_h^{LD}(n)+\bar{R}^{\ast}_{n+2}\leq     0+\bar{R}^{\ast}_{n+1},\\
   &~~~~~~~~~~~~~~~0\leq r_l^{LD}(n)\leq r_l^{\max},\\
   &~~~~~~~~~~~~~~~0\leq r_h^{LD}(n)\leq r_h^{\max},\\
   &\mathrm{variables}~~~~ r_l^{LD}(n),r_h^{LD}(n).
   \end{align}

Unlike the HP and the LP cases, we can derive the optimal prices under LD in closed-form.
\begin{proposition}\label{prop:LTS}
The optimal prices in time slot $n$ under the LD strategy are given by the interior point solution $I_0^{LD}:$
\begin{equation}\label{SpecialinteriorS}
 r_l^{LD}(n)=\frac{1}{2k_l},r_h^{LD}(n)=\min(\bar{R}^{\ast}_{n+1}-\bar{R}^{\ast}_{n+2}, r_h^{\max}).
\end{equation}
\end{proposition}

The proof of Proposition \ref{prop:LTS} is given in Appendix~G. We have analyzed the price optimization under any chosen admission strategy. Next, we will compare the expected total revenues $R^{HP\ast}_{n}$, $R^{LP\ast}_{n}$, and $R^{LD\ast}_{n}$ to pick the optimal pricing-admission strategy.

\subsection{Optimal Dynamic Pricing and Admission Policies}\label{dynamicpricingalg}

After deriving the optimal prices under each admission strategy, we can now compare the corresponding revenues and choose the best admission strategy for time slot $n$. We need to do this for each of the $N$ time slots. We show this process in Algorithm \ref{algorithm4}, which involves the previous solutions (Table \ref{tab:longdynamic}, Proposition \ref{prop:MTS}, and Equation (\ref{SpecialinteriorS})). More specifically, the algorithm iteratively computes the prices and revenues under the three admission strategies, respectively, and then selects the optimal prices and the corresponding admission strategy which lead to the largest revenue (lines 3 to 15).
The complexity of Algorithm \ref{algorithm4} is low and in the order of the total time slots $\mathcal{O}(N)$, as it only needs to check the tables and Equation (\ref{SpecialinteriorS}) we derived. We summarize the optimality result as follows.
\begin{theorem}\label{theorem2} 
The dynamic prices $\boldsymbol{r}^{\ast}=\{\boldsymbol{r}^{\ast}(n),\forall n\in\mathcal{N}\}$ and the dynamic admission policy $\boldsymbol{\pi}^{\ast}=\{a_n^{\ast}(S_n,X_n,Y_n),\forall n\in \mathcal{N}\}$ derived in Algorithm \ref{algorithm4} are the unique optimal solution to Problem \textbf{P2}.
\end{theorem}

The proof of Theorem \ref{theorem2} is given in Appendix H. 
Note that the optimal prices and admission policy form a \emph{contingency} plan that contains information about the optimal prices and admission decisions at all the possible system states $(S_n, X_n, Y_n)$ in any time slots $n\in\mathcal{N}$. To implement the optimal policy from time slot 1 to $N$, the database operator needs to decide the actual admission actions according to the realizations of random demands and the transition of system states. More specifically, at the beginning of each time slot $n$, the operator first announces prices $\boldsymbol{r}^{\ast}(n)$ according to $\boldsymbol{r}^{\ast}$ and checks the actual demands $(X_n,Y_n)$. Then, the admission decisions are determined by checking the optimal policy~$\boldsymbol{\pi}^{\ast}$ and the state component $S_n$ is updated accordingly.

\begin{algorithm} [t]
\caption{Optimal Dynamic Pricing and Admission Policy}\label{algorithm4}
\begin{algorithmic}[1]
\algsetup{linenosize=\scriptsize}
\scriptsize
\STATE Set $n=N+1$, $\bar{R}^{\ast}_{N+1}=0$
\STATE Set $r^{\ast}_l(N),r^{\ast}_h(N)$ by (\ref{SpecialinteriorS}) and $\bar{R}^{\ast}_{N}$ by $\bar{R}^{LD}_N(r^{\ast}_l(N),r^{\ast}_h(N))$.
\FOR{$n=N-1,\cdots,2,1$}
\STATE  Derive $r^{HP\ast}_l(n),r^{HP\ast}_h(n)$, $R^{HP\ast}_{n}$ by Table \ref{tab:longdynamic}.
\STATE Derive $r^{LP\ast}_l(n),r^{LP\ast}_h(n)$, $R^{LP\ast}_{n}$ by Prop. \ref{prop:MTS}.
\STATE Derive $r^{LD\ast}_l(n),r^{LD\ast}_h(n)$ and $R^{LD\ast}_{n}$ by (\ref{SpecialinteriorS}).
\STATE $\bar{R}_n^{\ast}\leftarrow \max\{{R}^{HP\ast}_{n},{R}^{LP\ast}_{n},{R}^{LD\ast}_{n}\}$ and $r^{\ast}_l(n),r^{\ast}_h(n)\gets \arg\max\{{R}^{HP\ast}_{n},{R}^{LP\ast}_{n},{R}^{LD\ast}_{n}\}$.
\IF {$r^{\ast}_l(n),r^{\ast}_h(n)=r^{HP\ast}_l(n),r^{HP\ast}_h(n)$}
 \STATE The heavy-priority strategy is optimal.
\ELSIF {$r^{\ast}_l(n),r^{\ast}_h(n)=r^{LP\ast}_l(n),r^{LP\ast}_h(n)$}
 \STATE The light-priority strategy is optimal.
\ELSE
 \STATE The light-dominant strategy is optimal.
\ENDIF
\ENDFOR
\RETURN Pricing-Admission policy $\boldsymbol{r}^{\ast}$ and $\boldsymbol{\pi}^{\ast}$.
\end{algorithmic}
\end{algorithm}

%% file: Section-Extension.tex
\section{Extensions}\label{sec:extension}
The analysis of the simplified case in Sections \ref{sec:spectrumallocation} to \ref{sec:dynamicpricing} paves the way for the analysis of the general case of multiple types of SUs. Next, we will first consider the case of arbitrary spectrum occupancies of two SU types, and then the general case of more than two SU types.
\subsection{Extension to Arbitrary Spectrum Occupancies of Two SU Types}\label{sec:extensionA}
In Sections \ref{sec:spectrumallocation} to \ref{sec:dynamicpricing}, we have assumed that a heavy-traffic SU occupies 2 consecutive time slots. Now we proceed to consider the general case where a heavy-traffic SU occupies $M$ consecutive time slots. The channel occupancy of a light-traffic SU is still normalized to a unit time slot. Naturally, we have $2\leq M\leq N$. Following similar notations as in Section~\ref{sec:spectrumallocation}, in order to characterize the spectrum occupancy information over time, we define $S_n$ as the number of \emph{remaining occupied} time slots \emph{before} making the admission action $a_n$ in time slot $n$, where $S_n\in\{0,1,\cdots,M-1\}$. At the beginning of time slot $n$, we first check the SU occupancy of the current time slot, i.e.,
\begin{equation}
S_n=\left\{
  \begin{aligned}
   &1,\cdots,M-1, \text{ if time slot $n$ is occupied},\\
   &0, \text{ ~~~~~~~~~~~~~~~if time slot $n$ is idle}.\\
  \end{aligned}
  \right.
\end{equation}
For example, if $M=3$ and we start admitting a heavy-traffic SU in time slot $n$, then $S_{n+1} = 2, S_{n+2} = 1$, and $S_{n+2}= 0$. If we define the possible admission action as $a_n=0$ (admitting no SU), $a_n=1$ (admitting a light-traffic SU), and $a_n = M$ (admitting a heavy-traffic SU), then the dynamics of the system state in (\ref{evolution}) still holds here, i.e., $S_{n+1} = (S_{n} + a_{n}(1-S_{n})-1)^+,\forall n\in \{1,\cdots,N-1\}$, and we define the whole system state in time slot $n$ as $(S_n, X_n, Y_n)$ similarly as in Section \ref{sunsec:pf}. The problem formulation turns out to be the same as Problem \textbf{P1}. As a result, the optimal admission policy can also be computed similarly as Algorithm \ref{algorithm}.

\subsubsection{Stationary Admission Policy under Static Pricing}
When we analyze the static pricing for this general case, a new challenge  is to understand that under which combination of system parameters the stationary admission policies are optimal, which is different from  those in Subsection \ref{sec:static}. Next we take the ``Heavy-Priority Admission Policy'' as an example, and derive the condition of the parameters $p_l$, $p_h$, $r_l$, and $r_h$, under which the stationary admission policy is optimal under static pricing.
\begin{proposition}\label{prop:general}
The optimal policy for solving the revenue maximization Problem \textbf{P1} degenerates to the heavy-priority stationary admission policy when price ratio between the heavy-traffic SU and the light-traffic SU is larger than a threshold $\theta_{th}^{HP}(p_l,p_h)$, i.e.,
\begin{equation}\label{generalthreshold}
r_h/r_l>\theta_{th}^{HP}(p_l,p_h),
\end{equation}
where the threshold ratio $\theta_{th}^{HP}(p_l,p_h)$ can be determined by solving the following:
\begin{equation}\label{generalcondition}
r_h +\bar{R}^{\ast}_{n+M}=r_l+\bar{R}^{\ast}_{n+1},\forall n\in\{1,2,\cdots,N-M+1\}.
\end{equation}
\end{proposition}

The proof of Proposition \ref{prop:general} is given in Appendix~I. We give the proof sketch as follows. First, we derive the expected revenue $\bar{R}^{\ast}_{n}$ as a function of $r_l,r_h,p_l,p_h$, given the heavy-priority stationary admission policy. Second, we determine $r_h/r_l$ in terms of $p_l$, $p_h$, and $n$, by plugging $\bar{R}^{\ast}_{n+1}$ and $\bar{R}^{\ast}_{n+M}$ into the condition (\ref{generalcondition}), i.e., $r_h/r_l=f(p_l,p_h,n)$. Third, we denote $f(p_l,p_h,n)$ as $\theta_{th}^{HP}(p_l,p_h,n)$, and derive the final threshold $\theta_{th}^{HP}(p_l,p_h)$ by optimizing $\theta_{th}^{HP}(p_l,p_h,n)$ over $n\in\{1,2,\cdots,N-M+1\}$. It thus follows that the heavy-priority stationary admission policy is optimal to solve the operator's revenue maximization problem if (\ref{generalthreshold}) holds. Proposition~\ref{prop:general} shows that our analysis in Section \ref{sunsec:pf} also applies to the general case. We can also derive the threshold condition for the light-priority admission policy by considering
$\bar{R}^{\ast}_{n+1}\leq r_h +\bar{R}^{\ast}_{n+M}\leq r_l+\bar{R}^{\ast}_{n+1}$, and the light-dominant admission policy by considering $r_h +\bar{R}^{\ast}_{n+M}<\bar{R}^{\ast}_{n+1}$ similarly. The related analysis are similar to Theorem \ref{theorem:allocation}. We skip the detailed analysis due to space constraints.



\subsubsection{Dynamic Pricing and Performance Evaluation}
The analysis under dynamic pricing is also similar to that in Section \ref{sec:dynamicpricing}, where we decompose the problem into three subproblems in each time slot. We show the main result in the following proposition, by focusing on the heavy-priority strategy for the illustration purpose.
\begin{proposition}\label{prop:generaldynamic}
Given an arbitrary value of spectrum occupancy $M$, the optimal dynamic pricing under the heavy-priority strategy is the same as that in Proposition \ref{prop:HTA} and Table \ref{tab:longdynamic}, once we replace $\bar{R}_{n+1}^{\ast}-\bar{R}_{n+2}^{\ast}$ by $\bar{R}_{n+1}^{\ast}-\bar{R}_{n+M}^{\ast}$.
\end{proposition}

The proof of the proposition is given in Appendix~J. Proposition \ref{prop:generaldynamic} shows that the previous analysis for dynamic pricing can be directly extended to the arbitrary occupancy case.







\begin{figure*}
\begin{align} \label{GeneralExpectation}
&R_n(S_n,X_n^{(1)},\cdots,X_n^{(I)},a_n)\notag\\
&\textstyle=\prod\limits_{i=1}^{I}(1-p_i)[0+\bar{R}_{n+1}^{\ast}(S_{n+1})]+\sum\limits_{i=1}^{I}p_i\prod\limits_{j\neq i}^{I}(1-p_j)[(0+\bar{R}_{n+1}^{\ast}(S_{n+1}))\cdot\boldsymbol{1}_{\{a_n=0\}}+(r_i+\bar{R}_{n+i}^{\ast}(S_{n+i}))\cdot\boldsymbol{1}_{\{a_n=i\}}]\notag\\
&\textstyle~~~+\sum\limits_{i=1}^{I}\sum\limits_{j\neq i}^{I}p_ip_j\prod\limits_{k\neq i,j}^{I}(1-p_k)[(0+\bar{R}_{n+1}^{\ast}(S_{n+1}))\cdot\boldsymbol{1}_{\{a_n=0\}}+(r_i+\bar{R}_{n+i}^{\ast}(S_{n+i}))\cdot\boldsymbol{1}_{\{a_n=i\}}+(r_j+\bar{R}_{n+j}^{\ast}(S_{n+j}))\cdot\boldsymbol{1}_{\{a_n=j\}}\notag]\\
&\textstyle~~~+\cdots+\prod\limits_{i=1}^{I}p_i\left((0+\bar{R}_{n+1}^{\ast}(S_{n+1}))\cdot\boldsymbol{1}_{\{a_n=0\}}+\sum\limits_{i=1}^{I}(r_i+\bar{R}_{n+i}^{\ast}(S_{n+i}))\cdot\boldsymbol{1}_{\{a_n=i\}}\right).
\end{align}
\hrulefill
\end{figure*}

\subsection{Extension to Multiple Types of SUs}\label{sec:extensionB}
\subsubsection{Model and Problem Formulation}
In this subsection, we further extend the analysis in Sections \ref{sunsec:pf} to \ref{sec:dynamicpricing} and Subsection \ref{sec:extensionA} to the case with a total of $I$ types of SUs seeking for spectrum access, including one type of light-traffic SUs and $I-1$ types of heavy-traffic SUs who occupy $2,3,\cdots,I$ consecutive time slots, respectively. We use $\mathcal{I}=\{1,2,\cdots,I\}$ to denote the set of SU types. To analyze the stationary admission policy, we need to compare a total of $I+1$ admission choices (including no admission) as in the analysis in Section \ref{sunsec:pf} and Subsection \ref{sec:extensionA}. The difference is that there are two revenue constraints for each policy in Section \ref{sunsec:pf} and Subsection \ref{sec:extension} (e.g., (\ref{eq:con1}) and (\ref{eq:con2})), while there are $I+1$ revenue constraints here. We continue the procedure and derive the associated thresholds, then determine the stationary admission policy by comparing the price relations with those thresholds.

More specifically, we define the prices charged to all types of SUs as $\mathcal{R}=\{r_i,\forall i\in\mathcal{I}\}$, where $r_i$ is the price charged to a type-$i$ SU for using the spectrum resource. Let the demand probabilities of all types of SUs be $\mathcal{P}=\{p_i,\forall i\in\mathcal{I}\}$, and the realizations of all types of SUs' demands in time slot $n$ be $X_n^{(i)},\forall i\in\mathcal{I},n\in\mathcal{N}$. Given $r_i\in\mathcal{R}$ and $p_i\in\mathcal{P}$, the expected total revenue in time slot $n$ is the summation of the immediate revenue (as a result of the immediate action $a_n$) and the expected future revenue $\bar{R}_{n+1}^{\ast}(S_{n+1})$ (if $a_n=0$ with no admission) or  $\bar{R}_{n+i}^\ast(S_{n+i})$ (if $a_n=i$, admitting a type-$i$ SU), considering all possible SU demands $(X_n^{(1)},\cdots,X_n^{(I)})$ in time slot~$n$. The detailed expression is given in (\ref{GeneralExpectation}).

At the beginning of time slot $n$, we determine the optimal admission decision by comparing the total revenue of admitting a particular type of SU, which involves both the immediate revenue $r_i$ and the maximum expected future revenue $\bar{R}_{n+i}^{\ast}(S_{n+i})$. Given SUs' demands in time slot $n$, if the optimal decision is no admission ($a_n=0$) due to a more profitable type of SU in the next time slot, the total revenue in time slot $n$ is $0+\bar{R}_{n+1}^{\ast}(S_{n+1})$. To summarize, the optimal decision in time slot $n$ is
\begin{align}\label{GeneralDecision}
a_n^{\ast}=&\textstyle\arg\max\limits_{a_n\in\mathcal{I}\cup\{0\}}\left\{0+\bar{R}_{n+1}^{\ast}(S_{n+1}),\right.\notag\\
&\textstyle ~~~~~~~~~~(r_i+\bar{R}_{n+i}^{\ast}(S_{n+i}))\cdot\boldsymbol{1}_{\{X_n^{(i)}=1\}},\forall i\in\mathcal{I}\}.
\end{align}
The above argument reveals a backward induction algorithm of determining the optimal admission decision in each time slot, which is similar to Algorithm \ref{algorithm}. We are interested in the optimality of the stationary admission policies as discussed in Subsection \ref{sec:static}.

\subsubsection{Stationary Admission Policies under Static Pricing}
We first consider a type-$i$ and a type-$j$ SU ($i>j>1$) who seek to occupy arbitrarily consecutive time slots $i$ and $j$, respectively. In this case, the priority of admitting a particular type of SUs depends on the values of $r_i+\bar{R}_{n+i}^{\ast}$, $r_j+\bar{R}_{n+j}^{\ast}$, and $0+\bar{R}_{n+1}^{\ast}$. For a particular time slot $n$, for example, if $r_i+\bar{R}_{n+i}^{\ast}>r_j+\bar{R}_{n+j}^{\ast}>0+\bar{R}_{n+1}^{\ast}$, we prefer to serve the type-$i$ SU type rather than the type-$j$ SU (i.e., the admission priority follows $\Lambda(i)>\Lambda(j)>\Lambda(0)$). By specifying the values of $a_n$ according to this admission priority in (\ref{GeneralExpectation}), we determine the differences $\bar{R}_{n+j}^{\ast}-\bar{R}_{n+i}^{\ast}$ and $\bar{R}_{n+1}^{\ast}-\bar{R}_{n+j}^{\ast}$ similarly as Theorem \ref{theorem:allocation} and Proposition \ref{prop:general}. The threshold that guarantees the condition $r_i+\bar{R}_{n+i}^{\ast}>r_j+\bar{R}_{n+j}^{\ast}>0+\bar{R}_{n+1}^{\ast}$ can be derived by solving this condition. Further, by optimizing the derived threshold over all time slots $n\in\mathcal{N}$, we derive the final threshold that guarantees the optimality of the admission priority $\Lambda(i)>\Lambda(j)>\Lambda(0)$ for all time slots. Hence, this  admission priority becomes one of the stationary admission policies. Similarly, the thresholds for the other five admission priorities can be determined by solving the corresponding revenue conditions.

The above discussions can be generalized to the case of multiple types of SUs as follows. For a particular time slot $n$, for example, if the revenue conditions satisfy $r_I+\bar{R}_{n+I}^{\ast}>r_{I-1}+\bar{R}_{n+I-1}^{\ast}>\cdots>r_1+\bar{R}_{n+1}^{\ast}>0+\bar{R}_{n+1}^{\ast}$, the admission priority follows $\Lambda(I)>\Lambda(I-1)>\cdots>\Lambda(1)>\Lambda(0)$. By specifying the values of $a_n$ according to this admission priority in (\ref{GeneralExpectation}), we determine the difference $\bar{R}_{n+j}^{\ast}-\bar{R}_{n+I}^{\ast},\forall j\in\{1,\cdots,I-1\}$ similarly as Theorem \ref{theorem:allocation} and Proposition \ref{prop:general}, respectively. We then proceed to derive the thresholds such that the revenue conditions hold for all time slots. These thresholds guarantee that the admission priority $\Lambda(I)>\Lambda(I-1)>\cdots>\Lambda(1)>\Lambda(0)$ is optimal for all time slots, and hence it becomes a stationary admission policy.

\begin{proposition}\label{prop:GeneralPolicy}
Given the set $\mathcal{I}$ of $I$ types of SUs, there are $(I+1)!$ admission priorities. For each admission priority, there exist thresholds of the price ratios such that the optimal admission priority for a time slot is optimal for all time slots (corresponding to an optimal stationary admission policy).
\end{proposition}

The proof of Proposition \ref{prop:GeneralPolicy} is given in Appendix K. Proposition \ref{prop:GeneralPolicy} shows that the threshold-based stationary policy still holds in the general scenario, and there exist $(I+1)!$ thresholds\footnote{To determine the specific admission strategy (priority) in each time slot, we need to sort the $I+1$ revenues in (\ref{GeneralDecision}) to the corresponding order. Hence, we have a $I+1$ permutation of $I+1$, which involves $(I+1)!$ admission strategies (priorities).} for all types of SUs $\mathcal{I}$, which are completely determined by the values of $\{0+\bar{R}_{n+1}^{\ast},r_i+\bar{R}_{n+i}^{\ast},\forall i\in\mathcal{I}\}$ in each time slot. Recall that in Subsection \ref{sec:static}, we should have $(2+1)!$ stationary admission policies. However, due to the fact $r_1+\bar{R}_{n+1}^{\ast}>0+\bar{R}_{n+1}^{\ast}$, finally we have a total of $\frac{(2+1)!}{2!}=3$ stationary admission policies.


\subsubsection{Optimal Dynamic Pricing and Dynamic Admission}
In the dynamic pricing setting, the joint pricing and admission problem in time slot $n$ can be formulated similarly as Problem \textbf{P3} in Section \ref{sec:dynamicpricing}, by changing the objective function to (\ref{GeneralExpectation}). Since there are $I+1$ possible revenues in (\ref{GeneralDecision}) and we need to determine their value orders, there are $(I+1)!$ admission strategies (admission priorities) as in Proposition \ref{prop:GeneralPolicy}. We follow the same pricing-admission decomposition procedure to transform the joint problem into $(I+1)!$ subproblems corresponding to the $(I+1)!$ admission strategies in this time slot.
As such, we can also derive the optimal pricing for maximizing the revenue in each time slot by solving those subproblems as we did in Section \ref{sec:dynamicpricing}, and then choose the admission strategy that leads to the largest revenue as shown in Algorithm \ref{algorithm4}. The analysis procedure is identical with that in the previous scenario. The only difference is that there are $I$ rather than two constraints (revenue conditions) in each optimization problem when assuming a particular admission strategy, hence it will be more complicated to optimize the prices in each subproblem.

%% file: Simulation.tex
\section{Simulation Results}\label{sec:simulation}
In this section, we provide the simulation results to illustrate our key insights regarding the performances of the dynamic admission control under both static pricing and dynamic pricing. We first illustrate the stationary admission policies for the dynamic admission control under static pricing and dynamic pricing, respectively. We then compare the revenue improvement of dynamic pricing over static pricing under a wide range of system parameters. 

\subsection{Optimal Static Pricing and Stationary Admission Policy}
In Subsection \ref{sec:staticpricing}, we derived the optimal static pricing by first assuming that one of the stationary admission policies is optimal. Recall that the three conditions in Theorem \ref{theorem:allocation} are characterized by the price ratio $r_h/r_l$. Given any demand elasticities $k_l$ and $k_h$ (hence any $r_h/r_l$ relation with respect to $p_l$ and $p_h$), it is natural to ask whether the optimal static pricing satisfies one of the conditions in Theorem \ref{theorem:allocation}, so that it is indeed optimal to choose a stationary admission policy after we optimize the static prices. Fig. \ref{staticfig1} illustrates the corresponding result, showing when a stationary admission control policy is optimal under the optimal static prices for particular system parameters $k_l$ and $k_h$. As we can see, except the small brown (Nonstationary Policy) regime which corresponds to regime III in Fig. \ref{big picture}, the stationary policies are optimal in most cases.

\begin{figure}[t]
\centering
\begin{overpic}[scale=0.4]{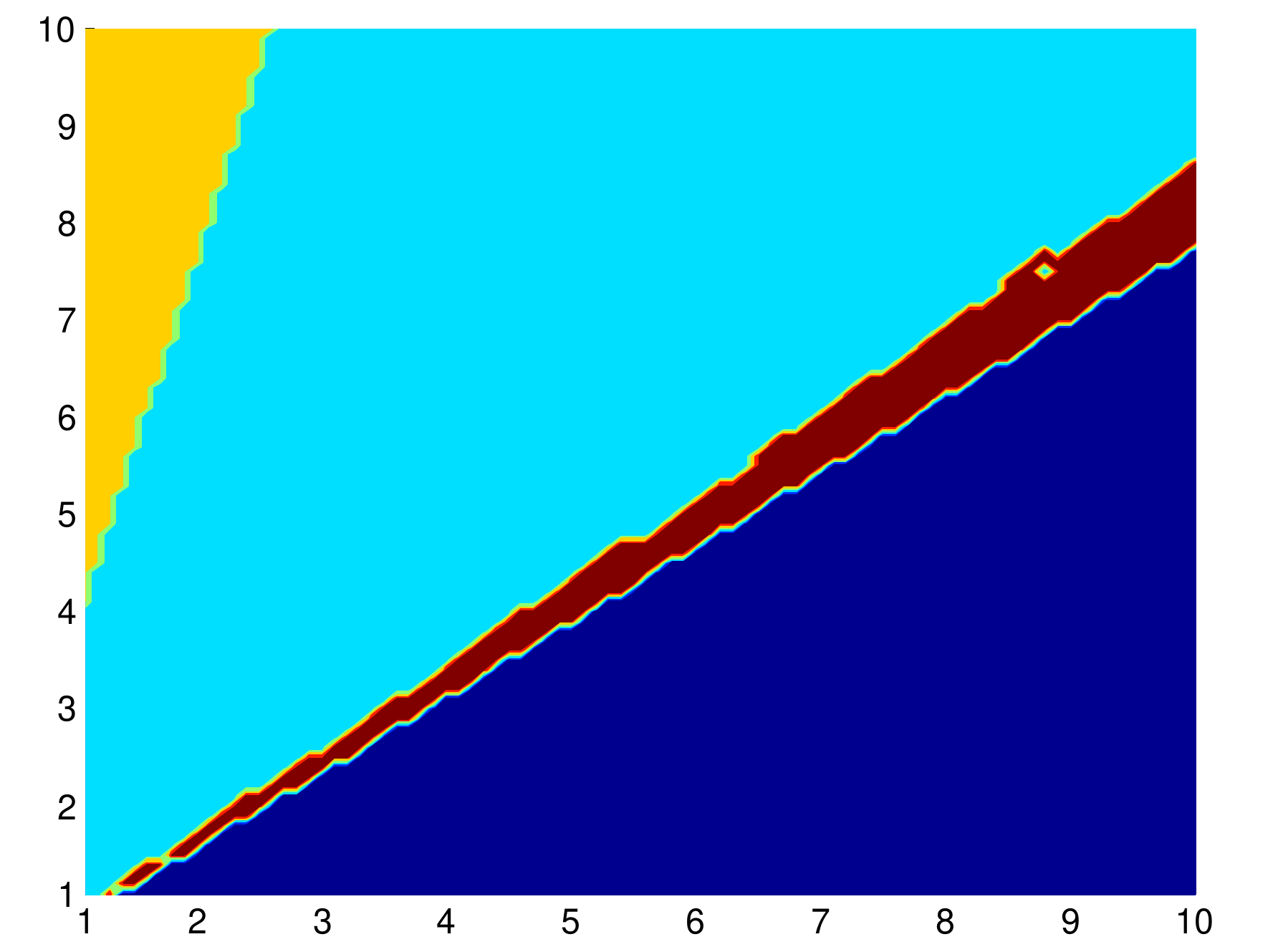}
\put(-1,15){\rotatebox{90}{{\tiny Heavy-traffic demand elasticity $k_h$}}}
\put(28,-3){\rotatebox{0}{{\tiny Light-traffic demand elasticity $k_l$}}}

\put(97,50){\rotatebox{90}{{\tiny Nonstationary Policy}}}
\put(97,61){\rotatebox{35}{\vector(-1,0){15}}}

\put(3,76){\rotatebox{0}{{\tiny Light-Dominant Admission Policy}}}
\put(15,76){\rotatebox{70}{\vector(-1,0){10}}}

\put(52,76){\rotatebox{0}{{\tiny Light-Priority Admission Policy}}}
\put(64,76){\rotatebox{70}{\vector(-1,0){10}}}

\put(97,2){\rotatebox{90}{{\tiny Heavy-Priority Admission Policy}}}
\put(97,24){\rotatebox{0}{\vector(-1,0){15}}}
\end{overpic}
\caption{Optimal choices of admission policies for different values of elasticity parameters $k_l$ and $k_h$. The yellow (Light-Dominant Admission Policy), cyan (Light-Priority Admission Policy), and blue (Heavy-Priority Admission Policy) regimes represent three stationary admission policies, i.e., I, II, and IV regimes in Fig. \ref{big picture}, respectively. The brown (Nonstationary Policy) regime requires ``Algorithm \ref{algorithm}'' to compute the optimal policy.}\label{staticfig1}
\end{figure}
\begin{figure}[t]
\centering
\begin{overpic}[scale=0.4]{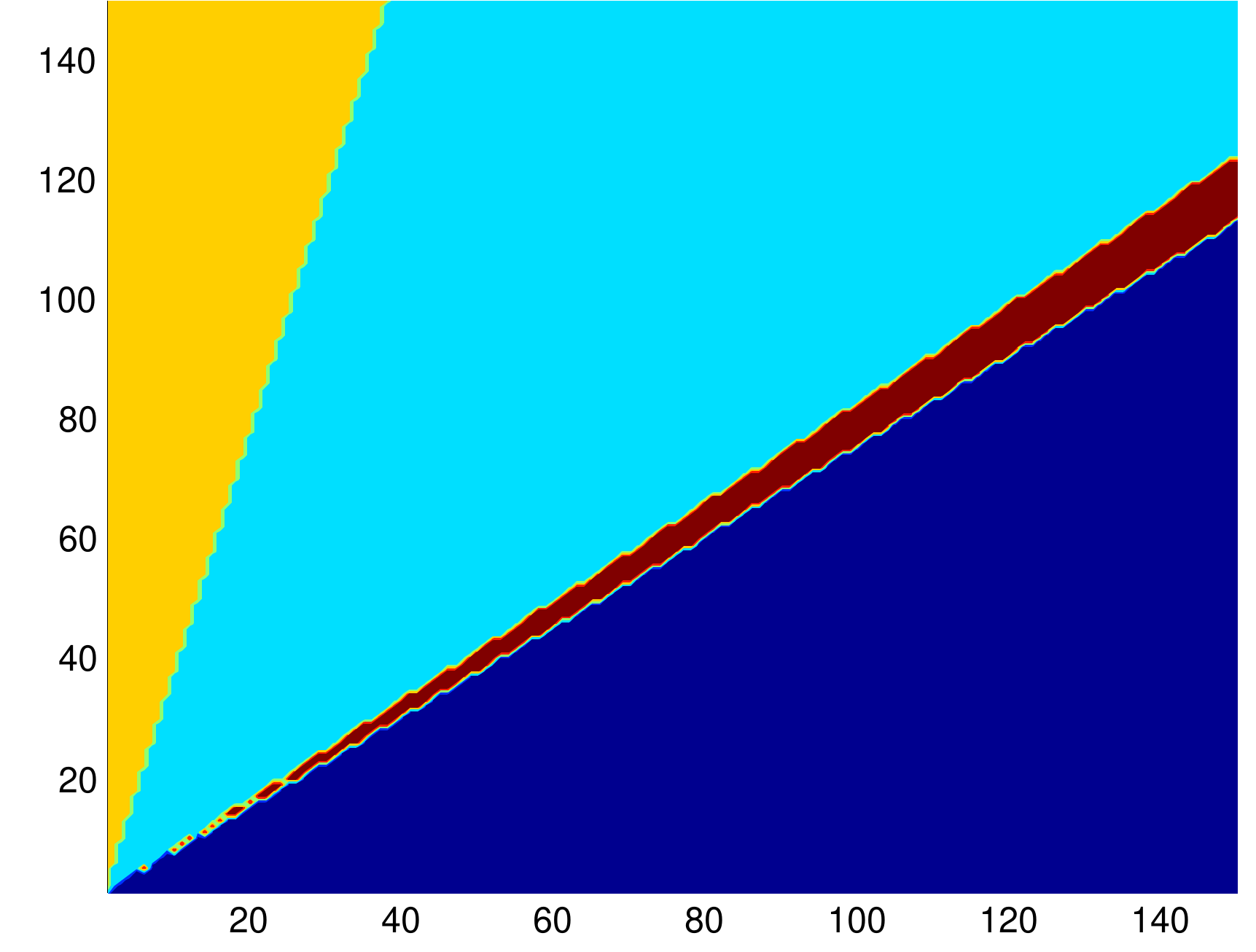}
\put(-1,15){\rotatebox{90}{{\tiny Heavy-traffic demand elasticity $k_h$}}}
\put(30,-3){\rotatebox{0}{{\tiny Light-traffic demand elasticity $k_l$}}}

\put(102,50){\rotatebox{90}{{\tiny Nonstationary Policy}}}
\put(102,63){\rotatebox{35}{\vector(-1,0){15}}}

\put(3,79){\rotatebox{0}{{\tiny LD Policy and Dynamic Pricing}}}
\put(17,79){\rotatebox{70}{\vector(-1,0){10}}}

\put(53,79){\rotatebox{0}{{\tiny LP Policy and Dynamic Pricing}}}
\put(67,79){\rotatebox{70}{\vector(-1,0){10}}}

\put(102,2){\rotatebox{90}{{\tiny HP Policy and Dynamic Pricing}}}
\put(102,24){\rotatebox{0}{\vector(-1,0){15}}}
\end{overpic}
\caption{Optimal admission policies under dynamic pricing over $N=100$ time slots. The yellow (LD Policy and Dynamic Pricing), cyan (LP Policy and Dynamic Pricing), and blue (HP Policy and Dynamic Pricing) regimes are the three stationary admission policies, i.e., I, II, and IV regimes in Fig. \ref{big picture}, respectively. The brown (Nonstationary Policy) regime requires Algorithm \ref{algorithm4} to compute the optimal policy.}\label{dypricingpolicy}
\end{figure}

\subsection{Optimal Dynamic Pricing and Stationary Admission Policy}

In Subsection \ref{dynamicpricingalg}, we have shown that in the most general case of dynamic pricing and dynamic admission control, the optimal admission strategies in different time slots may be different. On the other hand, it would be interesting to study under what system parameters the optimal admission decisions of different time slots (under dynamic pricing) will coincide with one of the stationary admission policies defined in Table~\ref{stationarypolicy}.

Recall that in our system model, as long as we adopt the linear demand functions, the system only has two parameters $k_l$ and $k_h$, and the other parameters (e.g., probabilities $p_l$ and $p_h$) are determined by $k_l$ and $k_h$. Fig. \ref{dypricingpolicy} illustrates the optimal admission and pricing decisions under dynamic pricing. We can see that the optimal admission strategies in Algorithm \ref{algorithm4} degenerate to stationary admission policies in most cases, and it is only optimal to switch between different admission strategies (HP, LD, and LP) in a small regime (the brown regime in Fig. \ref{dypricingpolicy}).
\begin{observation}
Under a wide range of system parameters $k_l$ and $k_h$, the optimal admission decisions developed in Algorithm \ref{algorithm4} (with the optimized optimal dynamic prices) degenerate to stationary admission policies over all time slots.
\end{observation}

When the stationary admission policy is optimal, we have the following claims.
\begin{itemize}
\item If light-traffic SUs are much more price-sensitive than heavy-traffic SUs (i.e., $k_l$ is significantly larger than $k_h$), the optimal dynamic pricing degenerates to the heavy-priority admission policy which is stationary over time.
\item If heavy-traffic SUs are much more sensitive to prices than light-traffic SUs ($k_l$ is significantly less than $k_h$),  the optimal dynamic pricing degenerates to the light-dominant admission policy which is stationary over time.
\item If both light- and heavy-traffic SUs' sensitivities $k_l$ and $k_h$ are comparable, the optimal dynamic pricing degenerates to the light-priority admission policy which is stationary over time.
\end{itemize}

\subsection{Performance Comparison of Optimal Dynamic Pricing with Optimal Static Pricing}
In addition to the optimal pricing and admission policies, it is also important to compare the performance of dynamic pricing with that of static pricing. The key benefit of static pricing is that it does not change over time. Unlike static pricing, the advantage of dynamic pricing is to achieve the maximum operator revenue. However, dynamic pricing has a higher implementational complexity. Next, we compare the optimal revenue of optimal dynamic pricing obtained in Theorem \ref{theorem2} with that of optimal static pricing obtained in Subsection \ref{sec:staticpricing}. 
Fig. \ref{evalution} shows the revenue improvement of dynamic pricing over static pricing under different demand elasticity values ($k_l$ and $k_h$). Here, we set the total time slots $N=100$, so that the time horizon is long enough to approximate the time-average performance.
\begin{observation}\label{obs2}
As shown in Fig. \ref{evalution}, dynamic pricing outperforms static pricing by more than $30\%$ when both types of SUs are sensitive to prices (i.e., both $k_l$ and $k_h$ are high). When both types of SUs are not price-sensitive  (i.e., $k_l$ and $k_h$ are low), dynamic pricing only leads to limited revenue improvement (less than $10\%$) than static pricing, and it is better to adopt static pricing due to its low complexity. 
\end{observation}

\begin{figure}[t]
\centering
\begin{overpic}[scale=0.4]{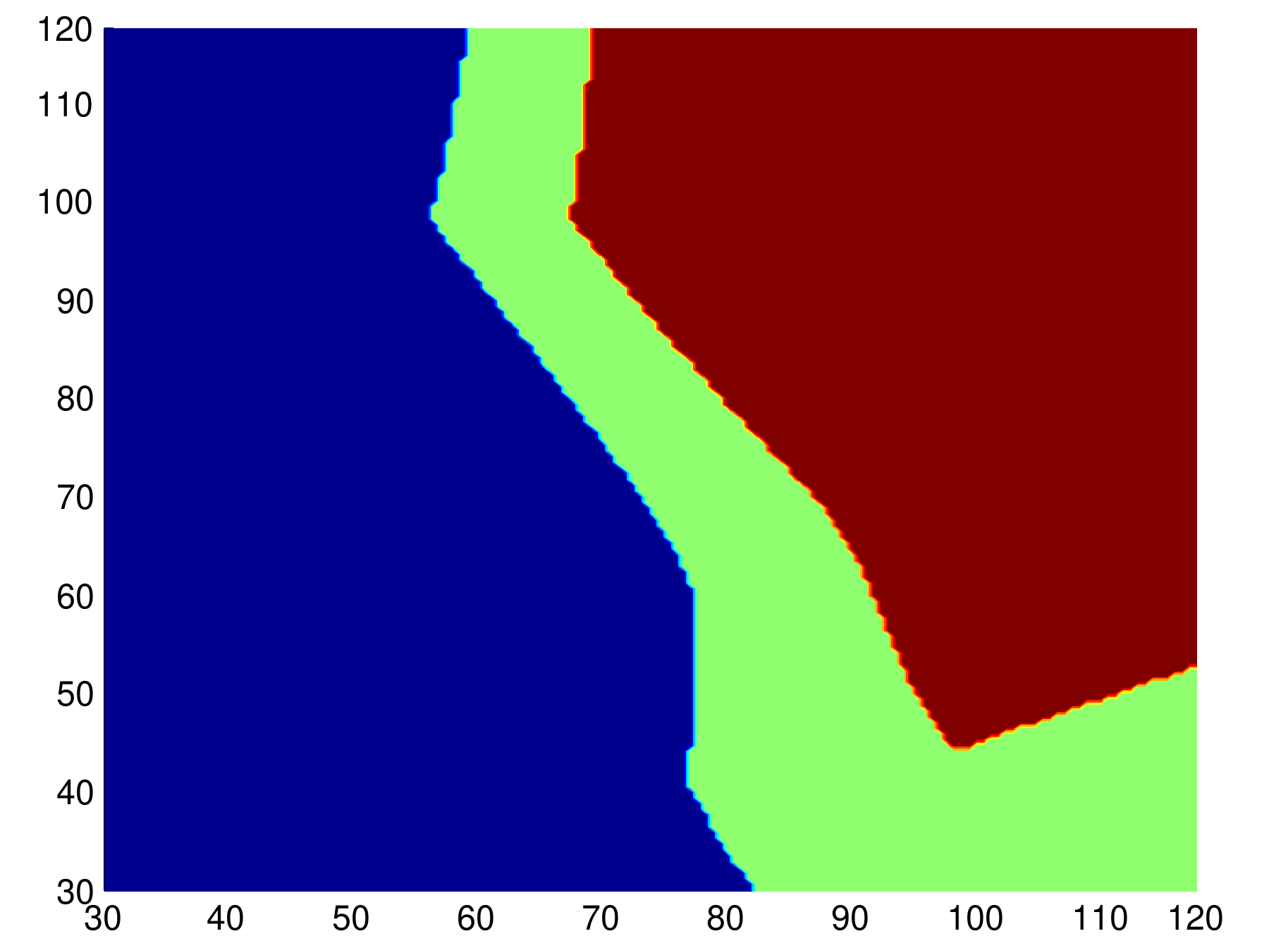}
\put(-1,15){\rotatebox{90}{{\tiny Heavy-traffic demand elasticity $k_h$}}}
\put(28,-3){\rotatebox{0}{{\tiny Light-traffic demand elasticity $k_l$}}}


\put(3,75){\rotatebox{0}{{\tiny Revenue Improvement $<10\%$}}}
\put(20,75){\rotatebox{0}{\vector(0,-1){15}}}

\put(53,75){\rotatebox{0}{{\tiny Revenue Improvement $\geq30\%$}}}
\put(72,75){\rotatebox{0}{\vector(0,-1){15}}}

\put(96,3){\rotatebox{90}{{\tiny $10\%\leq$ Revenue Improvement $<30\%$}}}
\put(96,15){\rotatebox{0}{\vector(-1,0){15}}}

\end{overpic}
\caption{The revenue improvement of dynamic pricing over static pricing for different $k_l$ and $k_h$ distributions.}\label{evalution}
\end{figure}

The above comparison is based on the assumption that heavy-traffic SUs request two consecutive time slots. In Section \ref{sec:extension}, we have extended the model to arbitrary spectrum occupancies. Hence, it is also interesting to show the comparison with more spectrum occupancies. Fig.~\ref{evalution2} shows the revenue improvement of dynamic pricing over static pricing with three consecutive time slots occupancy of heavy-traffic SUs ($M=3$). We can see that dynamic pricing significantly outperforms static pricing when SUs' demands are highly elastic, which is similar to Observation \ref{obs2}. Comparing with Fig. \ref{evalution} with $M=2$, the difference here is that a larger value of $M$ reduces the benefit of dynamic pricing. For example, when $k_l\in(90,120)$ and $k_h\in(60,70)$, the revenue improvement of dynamic pricing over static pricing is more than $30\%$ in Fig. \ref{evalution}, but is only around $10\%$ in Fig. \ref{evalution2}. The intuition is that a larger spectrum occupancy reduces the flexibility of dynamic pricing, since more slots will be occupied and cannot be dynamically allocated to new demands. Consider the extreme case $M=N$, then all slots will be occupied when admitting a heavy-traffic SU initially and dynamic pricing degenerates to static pricing. This implies that as the channel occupancy gap between the two SU types increases, it becomes increasingly attractive for the operator to choose the simple static pricing approach in order to achieve a close-to-optimal revenue.

\begin{figure}[t]
\centering
\begin{overpic}[scale=0.4]{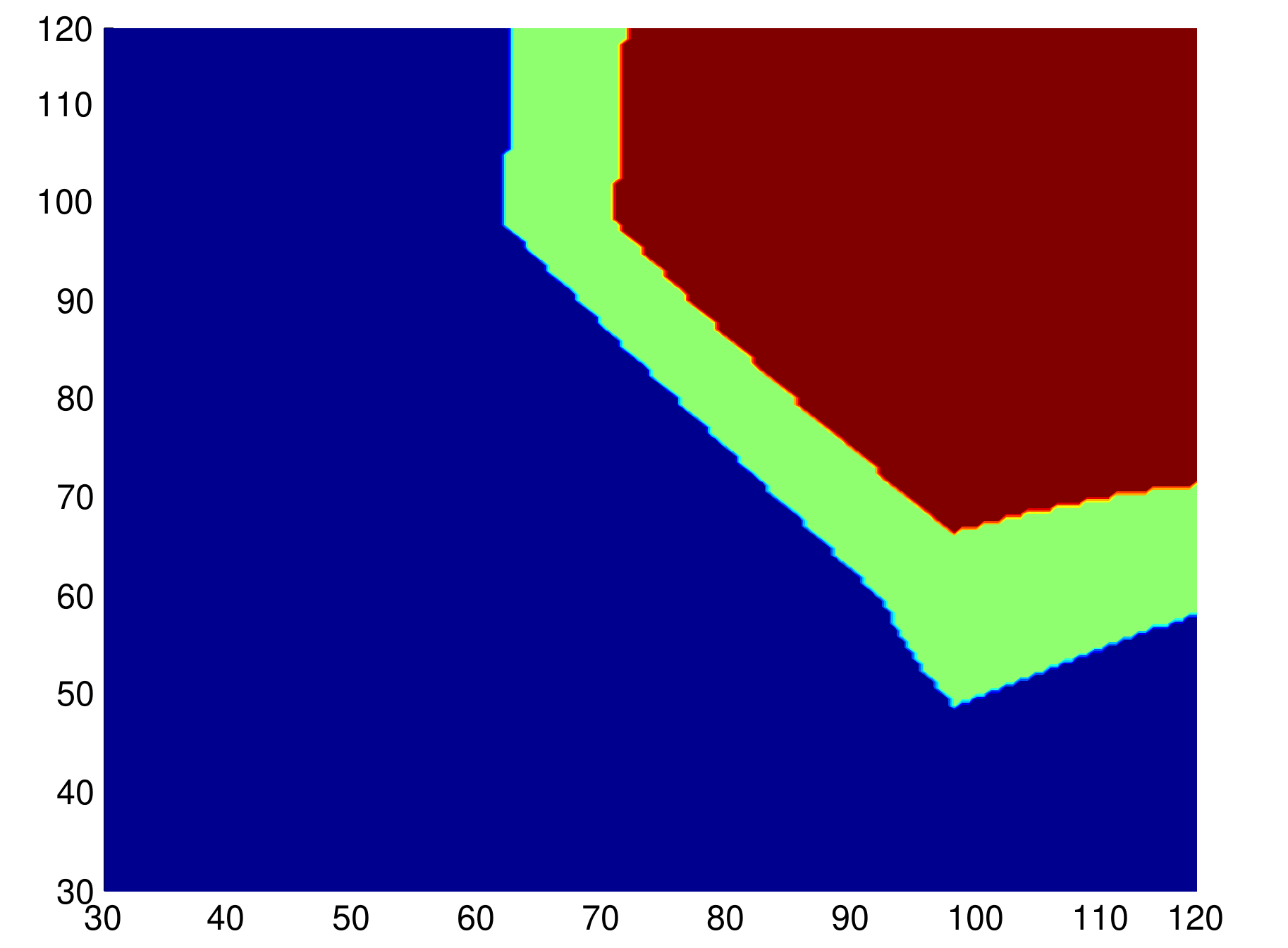}
\put(-1,15){\rotatebox{90}{{\tiny Heavy-traffic demand elasticity $k_h$}}}
\put(28,-3){\rotatebox{0}{{\tiny Light-traffic demand elasticity $k_l$}}}


\put(3,75){\rotatebox{0}{{\tiny Revenue Improvement $<10\%$}}}
\put(20,75){\rotatebox{0}{\vector(0,-1){15}}}

\put(53,75){\rotatebox{0}{{\tiny Revenue Improvement $\geq30\%$}}}
\put(72,75){\rotatebox{0}{\vector(0,-1){15}}}

\put(96,10){\rotatebox{90}{{\tiny $10\%\leq$ Revenue Improvement $<30\%$}}}
\put(96,30){\rotatebox{0}{\vector(-1,0){15}}}

\end{overpic}
\caption{The revenue improvement of dynamic pricing over static pricing for different $k_l$ and $k_h$ distributions. Here, the spectrum occupancy $M=3$, while $M=2$ in Fig. \ref{evalution}.}\label{evalution2}
\end{figure}

\subsection{Performance Comparison with a Related Study}
We numerically compare the revenue obtained by our policy with that obtained by the heuristic switch-over policy motivated by \cite{ORComparison}, which in our context admits a heavy-traffic SU only if half of the price charged to heavy-traffic SUs is no smaller than the price charged to light-traffic SUs, i.e., $r_h/2 \geq r_l$. The result (in Appendix L due to page limit) shows that the smaller difference of the demand elasticities between SUs will lead to larger revenue improvements. In general, our proposed scheme outperforms the policy in \cite{ORComparison} significantly in terms of the obtained revenue.
%
%
%

%% file: Section-Conclusion.tex
\section{Conclusion and Future Work}\label{sec:conclusion}
In this paper, we consider a spectrum database operator's revenue maximization problem through joint spectrum pricing and admission control. We incorporate the heterogeneity of SUs' spectrum occupancy and demand uncertainty into the model, and consider both the static and the dynamic pricing schemes. In static pricing, we show that stationary admission policies can achieve optimality in most cases. In dynamic pricing, we compute optimal pricing through a proper pricing-and-admission decomposition in each time slot. Furthermore, we show that dynamic pricing significantly improves revenue over static pricing when SUs are sensitive to prices change. Finally, we show that when the gap of the channel occupation length between two types of SUs increases, the gap between static pricing and dynamic pricing shrinks.

In the future work, we will consider the pricing and admission control of multiple channels. In this case, SUs may request different spectrum-time chunks in a two-dimensional time and frequency plane. One challenge is how to solve this Markov Decision Process (MDP), where the system state and state dynamics are much more complicated. We may further consider delay tolerant SUs who are willing to wait in queues if not admitted immediately, and use  the queueing based MDP to analyze the pricing and admission decisions for such a scenario. 



%

%% file: Section-Appendix.tex
%
\section{Proof for Proposition 1}
\setcounter{proposition}{0}
\begin{proposition}
Algorithm 1 solves Problem \textbf{P1} and computes the optimal admission policy  $\boldsymbol{\pi}^{\ast}$.
\end{proposition}
\begin{IEEEproof}
The optimal admission policy $\boldsymbol{\pi}^{\ast}$ is derived by using the standard backward induction algorithm for solving dynamic programming problems. According to the \emph{principle of optimality} in the standard textbook [24], we know that the backward induction algorithm in Cases I-III of Algorithm 1 is optimal for solving our dynamic programming problem.
\end{IEEEproof}

\section{Proof for Theorem 1}
\setcounter{theorem}{0}
\begin{theorem}[Optimality of Stationary Admission Policies]\label{theorem:allocationAppendix}
A stationary admission policy becomes the optimal policy to solve Problem \textbf{P1} if one of the following conditions is true:
\begin{itemize}
\item The Heavy-Priority admission policy $a_n^{HP\ast}$ in Tab.II\textendash$HP$ for all $n\in\mathcal{N}$ is optimal if $r_h/r_l\geq2p_l+(1-p_l)/(1-p_h)$.
\item The Light-Priority admission policy~$a_n^{LP\ast}$ in Tab.II\textendash$LP$ for all $n\in\mathcal{N}$ is optimal if $p_l\leq r_h/r_l\leq 1+p_l$.
\item The Light-Dominant admission policy $a_n^{LD\ast}$ in Tab.II\textendash$LD$ for all $n\in\mathcal{N}$ is optimal if $r_h/r_l<p_l$.
\end{itemize}
\end{theorem}

In the following, we will prove Theorem 1 item by item (from Item 1 up to Item 3).
\subsection{Proof for Item 1 in Theorem 1}
\textbf{Item 1:} The Heavy-Priority admission policy $a_n^{HP\ast}$ in Tab.II\textendash$HP$ for all $n\in\mathcal{N}$ is optimal if $r_h/r_l\geq2p_l+(1-p_l)/(1-p_h)$.

In order to prove Item 1, we start with the following lemma.
\setcounter{lemma}{0}
\begin{lemma}\label{proposition:cas1}
The optimal policy of solving Problem 1 is the stationary policy in (Tab.II\textendash$HP$: $a_n^{H\ast}$), if $\forall n=1, ... ,N-1$,
\begin{equation}
\frac{r_h}{r_l}>\theta_{th}^{HP}(n),
\end{equation}
where the threshold,
\begin{equation}
\theta_{th}^{HP}(n)=\frac{1+p_h+p_l(1-p_h)+2p_lp_h(-p_h)^{N-n-1}}{1+p_h(-p_h)^{N-n-1}}.
\end{equation}
\end{lemma}
\begin{IEEEproof}
In the last section, we have defined the maximized expected future revenue from time slot $n$ to $N$ as $\mathbb{E}_{X_{n},Y_{n}}[R^{\ast}_{n}(S_n,X_{n},Y_{n})]$. Since the expectation is taken over all the user arrival combinations $(X_n,Y_n)$, we rewrite it as $\bar{R}_n^{\ast}(S_n),\forall n\in \mathcal{N}$ for simplification. Besides, when the channel is occupied, i.e., $S_n=1$, we have $a_n^\ast=0$ for any value of $X_n$, $Y_n$, and $n$. So we will only focus on the case $(S,X,Y)=(0,X,Y)$ in the following. Thus we only need to consider $\bar{R}_n^{\ast}(0),\forall n\in \mathcal{N}$.

Note that the boundary condition in (4) is $\bar{R}_{N+1}^{\ast}(\cdot)=0$ for any values of $S_{N+1}$, $X_{N+1}$, and $Y_{N+1}$. Besides, $\bar{R}_N^{\ast}(0)=p_lr_l$ with respect to $X_N$ for any value of $Y_N$. Then for $1\leq n\leq N-1$, we have
\begin{align}
\bar{R}_n^{\ast}(0)&=(1-p_l)(1-p_h)(0+\bar{R}_{n+1}^{\ast}(0))\notag\\
&~~+(1-p_l)p_h(r_h+\bar{R}_{n+2}^{\ast}(0))\notag\\
&~~+p_l(1-p_h)(r_l+\bar{R}_{n+1}^{\ast}(0))+p_lp_h(r_h+\bar{R}_{n+2}^{\ast}(0))\notag\\
&=(1-p_h)\bar{R}_{n+1}^{\ast}(0)\notag\\
&~~+p_h\bar{R}_{n+2}^{\ast}(0)+p_lr_l(1-p_h)+p_hr_h.\label{expectationcas1}
\end{align}
Rearranging the terms of this equality yields
\begin{equation}\label{eq:con3prop}
\begin{aligned}
&\bar{R}_{n}^{\ast}(0)-\bar{R}_{n+1}^{\ast}(0)-\frac{p_lr_l(1-p_h)+p_hr_h}{1+p_h}\\
&=-p_h\Big(\bar{R}_{n+1}^{\ast}(0)-\bar{R}_{n+2}^{\ast}(0)-\frac{p_lr_l(1-p_h)+p_hr_h}{1+p_h}\Big)\\
&=-p_h^2\Big(\bar{R}_{n+2}^{\ast}(0)-\bar{R}_{n+3}^{\ast}(0)-\frac{p_lr_l(1-p_h)+p_hr_h}{1+p_h}\Big)\\
&=\cdots\\
&=(-p_h)^{N-n}\Big(\bar{R}_{N}^{\ast}(0)-\bar{R}_{N+1}^{\ast}(0)-\frac{p_lr_l(1-p_h)+p_hr_h}{1+p_h}\Big).
\end{aligned}
\end{equation}
From (\ref{eq:con3prop}), we derive the similar equations for $n+1,n+2,\cdots,N$ as follows.
\begin{align}
&\bar{R}_{n+1}^{\ast}(0)-\bar{R}_{n+2}^{\ast}(0)-\frac{p_lr_l(1-p_h)+p_hr_h}{1+p_h}\notag\\
&=(-p_h)^{N-n-1}\Big(\bar{R}_{N}^{\ast}(0)-\bar{R}_{N+1}^{\ast}(0)-\frac{p_lr_l(1-p_h)+p_hr_h}{1+p_h}\Big).\notag\\
&~~\vdots\notag\\
&\bar{R}_{N}^{\ast}(0)-\bar{R}_{N+1}^{\ast}(0)-\frac{p_lr_l(1-p_h)+p_hr_h}{1+p_h}\notag\\
&=(-p_h)^{0}\Big(\bar{R}_{N}^{\ast}(0)-\bar{R}_{N+1}^{\ast}(0)-\frac{p_lr_l(1-p_h)+p_hr_h}{1+p_h}\Big).\label{eq:congeneral}
\end{align}
By adding the above equations ((\ref{eq:con3prop}) and
(\ref{eq:congeneral})) and combining the terms, we can further derive the general expression of expected future revenue.
\begin{align}\label{eq:congenerexp}
\bar{R}_{n}^{\ast}(0)&=\bar{R}_{N+1}^{\ast}(0)+(N-n+1)\frac{p_lr_l(1-p_h)+p_hr_h}{1+p_h}\notag\\
&~~+\frac{2p_lr_lp_h-p_hr_h}{1+p_h}\frac{1-(-p_h)^{N-n+1}}{1+p_h}\notag\\
&=(N-n+1)\frac{p_lr_l(1-p_h)+p_hr_h}{1+p_h}\notag\\
&~~+\frac{2p_lr_lp_h-p_hr_h}{1+p_h}\frac{1-(-p_h)^{N-n+1}}{1+p_h}.
\end{align}

It is easy to check that when $ n=N $ and $ n=N+1 $, the expected revenue also satisfy (\ref{eq:congenerexp}), so (\ref{eq:congenerexp}) holds for all $ 1\leq n\leq N+1 $.

Substitute (\ref{eq:congenerexp}) into the heavy-priority admission condition (6), for $ 1\leq n\leq N-1 $, we have
\begin{equation}\label{eq:threshold}
\frac{r_h}{r_l} \geq \frac{1+p_h+p_l(1-p_h)+2p_lp_h(-p_h)^{N-n-1}}{1+p_h(-p_h)^{N-n-1}}.\\
\end{equation}
This completes the proof.
\end{IEEEproof}

With Lemma 1, we proceed to refine the threshold. The result is shown in the following lemma.
\begin{lemma}\label{theorem:cas1}
When the price ratio $r_h/r_l$ is larger than the threshold $\theta_{th}^{HP}=2p_l+\frac{1-p_l}{1-p_h}$, the benefit of serving a deterministic heavy-traffic SU is strictly larger than serving one deterministic and another possible one light-traffic SU for every time slot, i.e., the heavy-priority admission strategy holds for all time slots.
\end{lemma}
\begin{proof}
As shown in \emph{Lemma 1}, for $ 1\leq n\leq N-1 $,
$$\theta_{th}^{HP}(n)=\frac{1+p_h+p_l(1-p_h)+2p_lp_h(-p_h)^{N-n-1}}{1+p_h(-p_h)^{N-n-1}}.$$
We consider two regimes as follows, depending on whether $N-n-1$ is even or odd, respectively.
\begin{itemize}
\item \emph{Even-numbered Regime}: $ N-n-1=2k, k \in \mathbb{N} $, then
\begin{equation}\label{eq:cas1c11}
\begin{aligned}
\theta_{th}^{HP}(k)_e&=\frac{1+p_h+p_l(1-p_h)+2p_lp_h(-p_h)^{2k}}{1+p_h(-p_h)^{2k}}\\
&=2p_l+\frac{(1-p_l)(1+p_h)}{1+(p_h)^{2k+1}}.\\
\end{aligned}
\end{equation}
First, we consider the continuous version $ \theta_{th}^{HP}(x)_e $, and take the first-order derivative w.r.t. $ x $, i.e.,
\begin{equation}\label{eq:cas1c12}
(\theta_{th}^{HP})'(x)_e=-\frac{2(1-p_l)(1+p_h)p_h^{2x+1}\ln{p_h}}{(1+p_h^{2x+1})^2} \geq 0.\\
\end{equation}
So $ \theta_{th}^{HP}(x)_e $ is nondecreasing in $x$. So, let $ k \rightarrow \infty $\footnote{Since the available number of time slots is a finite number $N$, $k$ cannot be $\infty$ accordingly. However, we want to derive the upper bound of $\theta_{th}^{HP}(n)$ with an either big or small $N$. So, in order not to lose generality, we let $k \rightarrow \infty$ to get a big bound as much as possible. We will show that eventually the upper bound is attained at a finite number which only depends on $N$.}, we have $\theta_{th}^{HP}(k)_e^{\max}=2p_l+(1-p_l)(1+p_h) $.
Also note that when $k=0$, $\theta_{th}^{HP}(k)_e^{\min}=1+p_l $.
\item \emph{Odd-numbered Regime}: $ N-n-1=2k+1, k \in \mathbb{N} $, then
\begin{equation}\label{eq:cas1c13}
\begin{aligned}
\theta_{th}^{HP}(k)_o&=\frac{1+p_h+p_l(1-p_h)+2p_lp_h(-p_h)^{2k+1}}{1+p_h(-p_h)^{2k+1}}\\
&=2p_l+\frac{(1-p_l)(1+p_h)}{1-p_h^{2k+2}}.\\
\end{aligned}
\end{equation}
Similarly, consider the continuous version $ \theta_{th}^{HP}(x)_o $, and take the first-order derivative w.r.t. $ x $, i.e.,
\begin{equation}\label{eq:cas1c14}
(\theta_{th}^{HP})'(x)_o=\frac{2(1-p_l)(1+p_h)p_h^{2x+2}\ln{p_h}}{(1-p_h^{2x+2})^2} \leq 0.\\
\end{equation}
So $ \theta_{th}^{HP}(x)_o $ is non-increasing in $x$. One can see that when $ k=0 $, we have $ \theta_{th}^{HP}(k)_e^{\max}=2p_l+\frac{1-p_l}{1-p_h} $; and when $k \rightarrow \infty $, $\theta_{th}^{HP}(k)_e^{\min}=2p_l+(1-p_l)(1+p_h) $.
\end{itemize}
Since $ 2p_l+\frac{1-p_l}{1-p_h} \geq 2p_l+(1-p_l)(1+p_h) \geq 1+p_l $ for $p_l,p_h\in[0,1]$, we conclude that the upper bound for the threshold $ \theta_{th}^{HP}(n) $ is $ \theta_{th}^{HP}=2p_l+\frac{1-p_l}{1-p_h} $, as desired.

We can also see the upper bound is attained at $n=N-2$ (i.e., $k=0$ in the odd-numbered regime), which is completely determined by the number of time slots $N$.
\end{proof}

With Lemma 1 and Lemma 2,\footnote{Here the two lemmas are just used to prove Item 1 and only appear in here. They have nothing to do with the lemmas in the submission.} we conclude that Item 1 in Theorem 1 readily holds. This completes the proof for Item 1.
\subsection{Proof for Item 2 in Theorem 1}
\textbf{Item 2:} The Light-Priority admission policy~$a_n^{LP\ast}$ in Tab.II\textendash$LP$ for all $n\in\mathcal{N}$ is optimal if $p_l\leq r_h/r_l\leq 1+p_l$.
\begin{proof}
Also, the boundary condition in (4) is $\bar{R}_{N+1}^{\ast}(\cdot)=0$ for any values of $S_{N+1}$, $X_{N+1}$, and $Y_{N+1}$. Besides, $\bar{R}_N^{\ast}(0)=p_lr_l$ with respect to $X_N$ for any value of $Y_N$. Then for $1\leq n\leq N-1$, we have
\begin{equation}
\begin{aligned}
\bar{R}_n^{\ast}(0)&=(1-p_l)(1-p_h)(0+\bar{R}_{n+1}^{\ast}(0))\\
&~~+(1-p_l)p_h(r_h+\bar{R}_{n+2}^{\ast}(0))\\
&~~+p_l(1-p_h)(r_l+\bar{R}_{n+1}^{\ast}(0)+p_lp_h(r_l+\bar{R}_{n+1}^{\ast}(0))\\
&=(1-p_h+p_lp_h)\bar{R}_{n+1}^{\ast}(0)+(1-p_l)p_h\bar{R}_{n+2}^{\ast}(0)\\
&~~+p_lr_l+(1-p_l)p_hr_h\\
&=\bar{R}_{n+1}^{\ast}(0)+(-p_h+p_lp_h)(\bar{R}_{n+1}^{\ast}(0)-\bar{R}_{n+1}^{\ast}(0))\\
&~~+p_lr_l+(1-p_l)p_hr_h.\\
\end{aligned}
\end{equation}
Rearranging these terms yields
\begin{align}
&\bar{R}_{n}^{\ast}(0)-\bar{R}_{n+1}^{\ast}(0)\notag\\
&=(-p_h+p_lp_h)(\bar{R}_{n+1}^{\ast}(0)-\bar{R}_{n+2}^{\ast}(0))\notag\\
&~~+(1-p_h+p_lp_h)\frac{p_lr_l+(1-p_l)p_hr_h}{1-p_h+p_lp_h}.
\end{align}
Then,
\begin{equation}\label{eq:cas23}
\begin{aligned}
&\bar{R}_{n}^{\ast}(0)-\bar{R}_{n+1}^{\ast}(0)-\frac{p_lr_l+(1-p_l)p_hr_h}{1-(-p_h+p_lp_h)}\\
&\textstyle=(-p_h+p_lp_h)\left(\bar{R}_{n+1}^{\ast}(0)-\bar{R}_{n+2}^{\ast}(0)-\frac{p_lr_l+(1-p_l)p_hr_h}{1-(-p_h+p_lp_h)}\right)\\
&\textstyle=(-p_h+p_lp_h)^2\left(\bar{R}_{n+2}^{\ast}(0)-\bar{R}_{n+3}^{\ast}(0)-\frac{p_lr_l+(1-p_l)p_hr_h}{1-(-p_h+p_lp_h)}\right)\\
&\textstyle=\cdots\\
&\textstyle=(-p_h+p_lp_h)^{N-n}\left(\bar{R}_{N}^{\ast}(0)-\bar{R}_{N+1}^{\ast}(0)-\frac{p_lr_l+(1-p_l)p_hr_h}{1-(-p_h+p_lp_h)}\right)\\
&\textstyle=(-p_h+p_lp_h)^{N-n}\left(p_lr_l-\frac{p_lr_l+(1-p_l)p_hr_h}{1-(-p_h+p_lp_h)}\right)\\
&\textstyle=(-p_h+p_lp_h)^{N-n}\left(\frac{(-p_h+p_lp_h)(r_h-p_lr_l)}{1-(-p_h+p_lp_h)}\right).
\end{aligned}
\end{equation}
And we can further derive the aggregate expected future revenue, which is given by
\begin{align}\label{eq:cas26}
&\bar{R}_{n}^{\ast}(0)=\bar{R}_{N+1}^{\ast}(0)+(N-n+1)\frac{p_lr_l+(1-p_l)p_hr_h}{1-(-p_h+p_lp_h)}\notag\\
&~~+\left(\frac{(-p_h+p_lp_h)(r_h-p_lr_l)}{1-(-p_h+p_lp_h)}\right)\cdot\notag\\
&~~~~~~\frac{(-p_h+p_lp_h)(1-(-p_h+p_lp_h)^{N-n+1})}{1-(-p_h+p_lp_h)}.
\end{align}
Similarly, (\ref{eq:cas26}) holds for all $1\leq n\leq N+1$.

By substituting (\ref{eq:cas26}) into the light-priority admission condition $0+\bar{R}_{n+1}^{\ast}(0)<r_h+\bar{R}_{n+2}^{\ast}(0)<r_l+\bar{R}_{n+1}^{\ast}(0)$, we have
\begin{equation}\label{eq:cas24}
p_l \leq \frac{r_h}{r_l} < 1+p_l.
\end{equation}
This completes the proof.
\end{proof}
\subsection{Proof for Item 3 in Theorem 1}
\textbf{Item 3:} The Light-Dominant admission policy $a_n^{LD\ast}$ in Tab.II\textendash$LD$ for all $n\in\mathcal{N}$ is optimal if $r_h/r_l<p_l$.
\begin{proof}
The initial conditions are $\bar{R}_{N+1}^{\ast}(0)=0$ and $\bar{R}_{N}^{\ast}(0)=p_lr_l$. For $1\leq n\leq N-1$, we have
\begin{align}
\bar{R}_{n}^{\ast}(0)&=(1-p_l)(1-p_h)(0+\bar{R}_{n+1}^{\ast}(0))\notag\\
&~~+(1-p_l)p_h(0+\bar{R}_{n+1}^{\ast}(0))\notag\\
&~~+p_l(1-p_h)(r_l+\bar{R}_{n+1}^{\ast}(0))\notag\\
&~~+p_lp_h(r_l+\bar{R}_{n+1}^{\ast}(0))\notag\\
&=\bar{R}_{n+1}^{\ast}(0)+p_lr_l.
\end{align}
It follows that
\begin{equation}\label{eq:cas33}
\begin{aligned}
&\bar{R}_{n}^{\ast}(0)-\bar{R}_{n+1}^{\ast}(0)=p_lr_l.\\
\end{aligned}
\end{equation}
Further, we have
\begin{equation}\label{eq:cas36}
\begin{aligned}
\bar{R}_{n}^{\ast}(0)&=\bar{R}_{N+1}^{\ast}(0)+(N-n+1)p_lr_l=(N-n+1)p_lr_l.\\
\end{aligned}
\end{equation}
for all $1\leq n\leq N+1$.

Substitute (\ref{eq:cas36}) into the light-dominant admission condition $r_h+\bar{R}_{n+2}^{\ast}(0)\leq 0+\bar{R}_{n+1}^{\ast}(0)$, we have
\begin{equation}\label{eq:cas34}
\frac{r_h}{r_l} < p_l.
\end{equation}
This completes the proof.
\end{proof}

As such, Item 1, Item 2, and Item 3 in Theorem 1 all holds. We thus complete the proof for Theorem 1.

In the next three appendices, we will solve the problems and prove the lemmas and propositions in the dynamic pricing section. Notice that \emph{we keep the index of each lemma (resp. proposition) the same as that in the submission.}

\section{Proof for Proposition 2}
\begin{proof}
Given the case where the heavy-priority admission policy is optimal, we revisit the proof for Item 1 in Theorem 1 to show Proposition 2. By setting $n=1$ in (25), we immediately have the closed-form total revenue $R^{\ast}_1(r_l,r_h)$. Then we can optimize the static prices $r_l$ and $r_h$ with the expression $R^{\ast}_1(r_l,r_h)$. This completes the proof.
\end{proof}

%% file: Section-Problem4.tex
\section{The Solution to Problem 4}
\subsection{Problem 4: Optimal Pricing for time slot $n$ under HP}
The objective of Problem 4 (i.e., \textbf{P4}) may not be always concave and thus Problem 4 is not a convex optimization problem. Thus any solution satisfying KKT conditions may be only a local optimum of Problem 4. Despite this, we solve the problem analytically by following two procedures: we first find all local optimum solutions satisfying KKT conditions, and then compare these solutions to pick up the global optimum.

We introduce the dual variable  $\boldsymbol{\mu}=(\mu_a,\mu_l^{L},\mu_l^{R},\mu_h^{L},\mu_h^{R})$ ($\boldsymbol{\mu} \succeq \mathbf{0}$), and the associated Lagrangian $\mathcal{L}:\mathbb{R}\times\mathbb{R}\times\mathbb{R}^5\rightarrow\mathbb{R}$ is given by
\begin{equation}
\begin{aligned}
&\mathcal{L}(r_l^{HP}(n),r_h^{HP}(n),\boldsymbol{\mu})\\
&=\bar{R}_n^{HP}(r_l^{HP}(n),r_h^{HP}(n))\\
&~~+\mu_a\cdot\left(r_h^{HP}(n)+\bar{R}_{n+2}^{\ast}-r_l^{HP}(n)-\bar{R}_{n+1}^{\ast}\right)+\mu_{l}^{L}r_l^{HP}(n)\\
&~~+\mu_{l}^{R}(r_l^{\max}-r_l^{HP}(n))+\mu_{h}^{L}r_h^{HP}(n)+\mu_{h}^{R}(r_h^{\max}-r_h^{HP}(n)).\\
\end{aligned}
\end{equation}
The optimality conditions, \emph{i.e.}, Karush\textendash Kuhn\textendash Tucker conditions (KKT), are given by (\ref{eq:KKT1}), (\ref{eq:KKT2}), (\ref{eq:KKT3}), (\ref{eq:KKT4}), and (\ref{eq:KKT5}).
\begin{gather}\label{eq:KKT1}
\frac{\partial \mathcal{L}(r_l^{HP}(n),r_h^{HP}(n),\boldsymbol{\mu})}{\partial r_l^{HP}(n)}\left| _{(r_l^{{HP}\ast}(n),r_h^{{HP}\ast}(n),\boldsymbol{\mu}^*)}\right.=0,\\
\label{eq:KKT2}
\begin{aligned}
&\frac{\partial \mathcal{L}(r_l^{HP}(n),r_h^{HP}(n),\boldsymbol{\mu})}{\partial r_h^{HP}(n)}\left| _{(r_l^{{HP}\ast}(n),r_h^{{HP}\ast}(n),\boldsymbol{\mu}^*)}\right.=0,\\
\end{aligned}\\
\label{eq:KKT3}
\begin{aligned}
&\textstyle\mu_a^{\ast}\cdot\left(r_h^{{HP}\ast}(n)+\bar{R}_{n+2}^{\ast}-r_l^{{HP}\ast}(n)-\bar{R}_{n+1}^{\ast}\right)=0,\\
&\textstyle\mu_{l}^{L\ast}r^{{HP}\ast}_l(n)=0, \mu_{l}^{R\ast}(r_l^{\max}-r^{{HP}\ast}_l(n))=0,\\
&\mu_{h}^{L\ast}r^{{HP}\ast}_h(n)=0, \mu_{h}^{R\ast}(r_h^{\max}-r^{{HP}\ast}_h(n))=0,\\
\end{aligned}\\
\label{eq:KKT4}
\boldsymbol{\mu}^{\ast}\succeq \mathbf{0},\\
\label{eq:KKT5}
\begin{aligned}
&r^{{HP}\ast}_h(n)+\bar{R}_{n+2}^{\ast}\geq r^{{HP}\ast}_l(n)+\bar{R}_{n+1}^{\ast},\\
&0\leq r^{{HP}\ast}_l(n)\leq r_l^{\max}, 0\leq r^{{HP}\ast}_h(n)\leq r_h^{\max}.\\
\end{aligned}
\end{gather}
They are stationarity, complementary slackness, dual feasibility, and primal feasibility, respectively. With KKT conditions, we can search for local optimum solutions to the constrained optimization Problem 4. We still need to find the global optimum by examining all local optima.

To find out an optimal solution to Problem 4, we need to first examine the feasible region of this problem based on any possible parameter values. Then we only need to check whether all the possible extreme points and the interior points satisfying KKT conditions to be local optima. We leverage the proof by contradiction to exclude most extreme point solutions, either with one or two active constraints. The results can be shown in \emph{Lemma 3}.
\subsection{Proof for Lemma 3}\label{appendth2}
\begin{lemma}\label{lemma3}
There are only three possible solutions satisfying the KKT conditions of \textbf{P4}: the interior point solution $I_0^{HP}:\big(r^{{HP}I}_l(n),r^{{HP}I}_h(n)\big)$, and two extreme point solutions $E_1^{HP}:\big(r^{{HP}E}_l(n),r^{{HP}E}_h(n)\big)$ and $E_2^{HP}:\big(r^{{HP}E'}_l(n),r^{{HP}E'}_h(n)\big)$.
\end{lemma}
\begin{IEEEproof}
We first consider the interior point solution. In this case, no constraints are active, i.e., all constraints are strictly inequalities. The solution $r_l^{{HP}I}(n),r_h^{{HP}I}(n)$ can be given by
\begin{equation} \label{eq:interior}
\left\{
  \begin{aligned}
   &(1 \!-\! p_h(r_h^{{HP}\ast}(n)))(p'_l(r_l^{{HP}\ast}(n))r_l^{{HP}\ast}(n) \!+\! p_l(r_l^{{HP}\ast}(n)))\!=\!0,\\
   &-p'_h(r_h^{{HP}\ast}(n))\bar{R}_{n+1}^{\ast}+p'_h(r_h^{{HP}\ast}(n))\bar{R}_{n+2}^{\ast}\\
   &~~+p_l(r_l^{{HP}\ast}(n))r_l^{{HP}\ast}(n)(-p'_h(r_h^{{HP}\ast}(n)))\\
   &~~+p'_h(r_h^{{HP}\ast}(n))r_h^{{HP}\ast}(n)+p_h(r_h^{{HP}\ast}(n))=0,\\
   &\boldsymbol{\mu}^{\ast}=\mathbf{0}.\\
  \end{aligned}
\right.
\end{equation}
The condition is that all constraints are strictly inequalities.

Then we consider the extreme point solutions. In this case, at least one constraint is active (one or two here). We further look at the feasible region $\mathcal{F}$, which is composed of five line segments. Then $\mathcal{F}\neq \emptyset$ iff $r_h^{\max}\geq \bar{R}^{\ast}_{n+1}-\bar{R}^{\ast}_{n+2}\geq0$.

\begin{tikzpicture}[scale=1.3]
    \draw [<->,thick] (0,2.5) node (yaxis) [above] {$r_h^{HP}(n)$}
        |- (3.0,0) node (xaxis) [right] {$r_l^{HP}(n)$};
    \draw (1,0) coordinate (a_1) -- (1,0) coordinate (a_2);
    \draw (0,0.3) coordinate (a_1) -- (1.5,1.8) coordinate (b_2) node[right, text width=14em] {$r_h^{HP}(n)=r_l^{HP}(n)+\bar{R}^{\ast}_{n+1}-\bar{R}^{\ast}_{n+2}$};
    \draw (0,1.4) coordinate (a_2) -- (2,1.4) coordinate (b_3) node[right, text width=10em] {$r_h^{HP}(n)=r_h^{\max}$};
    \draw (1.3,0) coordinate (a_2) -- (1.3,2.5) coordinate (b_3) node[right, text width=10em] {$r_l^{HP}(n)=r_l^{\max}$};
    \draw (0,1.4)  node[left] {$r_h^{\max}$};
    \draw (1.3,0)  node[below]  {$r_l^{\max}$};
    \draw (0,0)  node[left] {$0$};
    \draw (0.5,1.1) node {$\mathcal{F}$};
    \path[draw,thick,fill=blue!80](0,0.3)--(0,1.4)--(1.1,1.4)--cycle;
      \fill[red] (0,0.3) circle (2pt);
      \fill[red] (0,1.4) circle (2pt);
      \fill[red] (1.1,1.4) circle (2pt);
      \fill[blue] (0.3,0.9) circle (2pt);
      \node(c) at (2.0,0.8) [circle,draw,fill=blue!20] {$\mathcal{F}$};
      \draw[blue,<-] (0.3,0.9) .. controls +(1,1) and +(-1,-1) .. (c);
\end{tikzpicture}

\begin{tikzpicture}[scale=1.3]
    \draw [<->,thick] (0,2.5) node (yaxis) [above] {$r_h^{HP}(n)$}
        |- (3.0,0) node (xaxis) [right] {$r_l^{HP}(n)$};
    \draw (1,0) coordinate (a_1) -- (1,0) coordinate (a_2);
    \draw (0,0.3) coordinate (a_1) -- (1.5,1.8) coordinate (b_2) node[right, text width=14em] {$r_h^{HP}(n)=r_l^{HP}(n)+\bar{R}^{\ast}_{n+1}-\bar{R}^{\ast}_{n+2}$};
    \draw (0,2.2) coordinate (a_2) -- (2,2.2) coordinate (b_3) node[right, text width=10em] {$r_h^{HP}(n)=r_h^{\max}$};
    \draw (1.3,0) coordinate (a_2) -- (1.3,2.5) coordinate (b_3) node[right, text width=10em] {$r_l^{HP}(n)=r_l^{\max}$};
    \draw (0,2.2)  node[left] {$r_h^{\max}$};
    \draw (1.3,0)  node[below]  {$r_l^{\max}$};
    \draw (0,0)  node[left] {$0$};
    \draw (0.6,1.5) node {$\mathcal{F}$};
    \path[draw,thick,fill=blue!80](0,0.3)--(0,2.2)--(1.3,2.2)--(1.3,1.6)--cycle;
        \fill[red] (0,0.3) circle (2pt);
        \fill[red] (0,2.2) circle (2pt);
        \fill[red] (1.3,2.2) circle (2pt);
        \fill[red] (1.3,1.6) circle (2pt);
        \fill[blue] (0.7,1.7) circle (2pt);
        \node(c) at (2.0,0.8) [circle,draw,fill=blue!20] {$\mathcal{F}$};
        \draw[blue,<-] (0.7,1.7) .. controls +(-1,-1) and +(-1,0) .. (c);
\end{tikzpicture}

(1) We claim that there exist at most four intersection points on the boundary. So there exist at most four extreme point solutions with two active constraints (Exactly two dual variables are strictly positive.).

\textbf{(i).} If $\bar{R}^{\ast}_{n+1}-\bar{R}^{\ast}_{n+2}< r_h^{\max}< r_l^{\max}+\bar{R}^{\ast}_{n+1}-\bar{R}^{\ast}_{n+2}$, there exist three intersection points. $(r^{{HP}E'}_l(n),r^{{HP}E'}_h(n))=(0,\bar{R}^{\ast}_{n+1}-\bar{R}^{\ast}_{n+2})$, $(0,r_h^{\max})$, and $(r_h^{\max}+\bar{R}^{\ast}_{n+2}-\bar{R}^{\ast}_{n+1},r_h^{\max})$, the corresponding dual variables are ($\mu_a^{\ast}>0,\mu_{l}^{L\ast}>0$), ($\mu_{l}^{L\ast}>0,\mu_{h}^{R\ast}>0$), and ($\mu_a^{\ast}>0,\mu_{h}^{R\ast}>0$).

For $(r^{{HP}E'}_l(n),r^{{HP}E'}_h(n))=(0,\bar{R}^{\ast}_{n+1}-\bar{R}^{\ast}_{n+2})$, we have $p_h(r^{{HP}\ast}_h(n))+\mu_a^{\ast}=0$ according to (\ref{eq:KKT2}), which is a contradiction.

For $(r^{{HP}E'}_l(n),r^{{HP}E'}_h(n))=(0,r_l^{\max})$, we have $(1-p_h(r^{{HP}\ast}_h(n)))p_l(r^{{HP}\ast}_l(n))+\mu_l^{{HP}\ast}=0$ according to (\ref{eq:KKT1}), which is also a contradiction.

For $(r^{{HP}E'}_l(n),r^{{HP}E'}_h(n))=(r_h^{\max}+\bar{R}^{\ast}_{n+2}-\bar{R}^{\ast}_{n+1},r_h^{\max})$, we have
\begin{equation} \label{boundary1}
\left\{
\begin{aligned}
  &\mu_a^{\ast}=p'_l(r_l^{{HP}\ast}(n))r_l^{{HP}\ast}(n)+ p_l(r_l^{{HP}\ast}(n)),\\
  &\mu_h^{R\ast}=-p'_h(r_h^{\max})(-r_l^{{HP}\ast}(n)+p_l(r_l^{{HP}\ast}(n))r_l^{{HP}\ast}(n))+\mu_a^{\ast}.\\
\end{aligned}
\right.
\end{equation}
So, we need to check whether ($\mu_a^{\ast}>0,\mu_{h}^{R\ast}>0$) holds.

\textbf{(ii).} If $r_h^{\max}\geq r_l^{\max}+\bar{R}^{\ast}_{n+1}-\bar{R}^{\ast}_{n+2}$, there are two intersection points $(r^{{HP}E'}_l(n),r^{{HP}E'}_h(n))=(r_l^{\max},r_h^{\max})$, $(r_l^{\max},r_l^{\max}+\bar{R}^{\ast}_{n+1}-\bar{R}^{\ast}_{n+2})$ and the dual variables are ($\mu_{l}^{R\ast}>0,\mu_{h}^{R\ast}>0$), ($\mu_{a}^{\ast}>0,\mu_{l}^{R\ast}>0$).

For $(r^{{HP}E'}_l(n),r^{{HP}E'}_h(n))=(r_l^{\max},r_h^{\max})$, this is not possible a solution since $p_l(r_l^{\max})=p_l(r_l^{\max})=0$. Then we have $p'_l(r_l^{\max}r_l^{\max})=\mu_{l}^{R\ast}<0$ according to (\ref{eq:KKT1}), which is a contradiction.

For $(r^{{HP}E'}_l(n),r^{{HP}E'}_h(n))=(r_l^{\max},r_l^{\max}+\bar{R}^{\ast}_{n+1}-\bar{R}^{\ast}_{n+2})$, note that $p_l(r_l^{\max})=0$. We have $(1 - p_h(r_h^{{HP}\ast}(n)))\cdot(p'_l(r_l^{{HP}\ast}(n))r_l^{{HP}\ast}(n))-\mu_a-\mu_{l}^{R\ast}=0$ according to (\ref{eq:KKT1}), which is impossible since the three terms are all negative.

\textbf{Conclusion:} We thus conclude that the only possible extreme point solution with two active constraints is $(r^{{HP}E'}_l(n),r^{{HP}E'}_h(n))=(r_h^{\max}+\bar{R}^{\ast}_{n+2}-\bar{R}^{\ast}_{n+1},r_h^{\max})$ such that (\ref{boundary1}) holds.

(2) We consider only one active constraint.

\textbf{(i).} If $\bar{R}^{\ast}_{n+1}-\bar{R}^{\ast}_{n+2}< r_h^{\max}< r_l^{\max}+\bar{R}^{\ast}_{n+1}-\bar{R}^{\ast}_{n+2}$, $\mathcal{F}$ is a triangle.

\emph{Case 1}: $\mu_a^{\ast}>0$ and others are all $0$. The solution $(r^{{HP}E}_l(n),r^{{HP}E}_h(n))$ is given by
\begin{equation} \label{boundary2}
\left\{
  \begin{aligned}
   & (1 - p_h(r_h^{{HP}\ast}(n)))\cdot(p'_l(r_l^{{HP}\ast}(n))r_l^{{HP}\ast}(n)\\
   &~~~~~~~~~~~~~~~~~~~~~~~~~~~~~~+p_l(r_l^{{HP}\ast}(n)))-\mu_a^{\ast}=0,\\
   & -p'_h(r_h^{{HP}\ast}(n))\bar{R}_{n+1}^{\ast}+p'_h(r_h^{{HP}\ast}(n))\bar{R}_{n+2}^{\ast}\\
   &~~+p_l(r_l^{{HP}\ast}(n))r_l^{{HP}\ast}(n)(-p'_h(r_h^{{HP}\ast}(n)))\\
   &~~+p'_h(r_h^{{HP}\ast}(n))r_h^{{HP}\ast}(n)+p_h(r_h^{{HP}\ast}(n))+\mu_a^{\ast}=0,\\
   & r^{{HP}\ast}_h(n)+\bar{R}^{\ast}_{n+2}=r^{{HP}\ast}_l(n)+\bar{R}^{\ast}_{n+1},\\
   & 0<r_l^{{HP}\ast}(n)<r_l^{\max},0<r_h^{{HP}\ast}(n)<r_h^{\max}.\\
  \end{aligned}
\right.
\end{equation}

\emph{Case 2}: $\mu_{l}^{L\ast}>0$ and others are all $0$. This is impossible.

\emph{Case 3}: $\mu_{h}^{R\ast}>0$ and others are all $0$. $r_h^{{HP}E}(n)=r_h^{\max}$, and $r_l^{{HP}E}(n)$ is given by
\begin{equation}\label{boundary3}
\left\{
\begin{aligned}
  &p'_l(r_l^{{HP}\ast}(n))r_l^{{HP}\ast}(n)+ p_l(r_l^{{HP}\ast}(n))=0,\\
  &\mu_h^{R\ast}=-p'_h(r_h^{\max})(\bar{R}_{n+1}^{\ast}-\bar{R}_{n+2}^{\ast}\\
  &~~~~~~~~~~~~~~~+p_l(r_l^{{HP}\ast}(n))r_l^{{HP}\ast}(n)-r_h^{\max})>0.\\
\end{aligned}
\right.
\end{equation}
Then $\bar{R}_{n+1}^{\ast}-\bar{R}_{n+2}^{\ast}>r_h^{\max}-p_l(r_l^{{HP}\ast}(n))r_l^{{HP}\ast}(n)$, however, $\bar{R}_{n+1}^{\ast}-\bar{R}_{n+2}^{\ast}<r_h^{\max}-r_l^{{HP}\ast}(n)$. This is a contradiction.

\textbf{(ii).} If $r_h^{\max}\geq r_l^{\max}+\bar{R}^{\ast}_{n+1}-\bar{R}^{\ast}_{n+2}$, $\mathcal{F}$ is a trapezoid. Then other than the above three cases, there is another case.

\emph{Case 4}: $\mu_{l}^{R\ast}>0$ and others are all $0$. Then $p_l(r_l^{\max})=0$, we have $(1 - p_h(r_h^{{HP}\ast}(n)))\cdot(p'_l(r_l^{{HP}\ast}(n))r_l^{{HP}\ast}(n))-\mu_{l}^{R\ast}=0$ according to (\ref{eq:KKT1}), which is impossible since the three terms are all negative.

\textbf{Conclusion:} We conclude that the only possible extreme point solution with one active constraints $(r^{{HP}E}_l(n),r^{{HP}E}_h(n))$ is given by (\ref{boundary2}).

According to the above analysis, we conclude that the solution has THREE cases: $r_l^{{HP}I}(n),r_h^{HPI}(n)$ given by (\ref{eq:interior}), $(r^{{HP}E'}_l(n),r^{{HP}E'}_h(n))$ given by (\ref{boundary1}), and $(r^{{HP}E}_l(n),r^{{HP}E}_h(n))$ given by (\ref{boundary2}) such that the corresponding condition holds. This completes the proof.
\end{IEEEproof}

\subsection{Proof for Proposition 3}\label{appendthprop}
\setcounter{proposition}{2}
\begin{proposition}\label{prop:HTAAppendix}
The optimal pricing in time slot $n$ under the HP strategy is summarized in Table IV, which depends on the values of $\bar{R}_{n+1}^{\ast}-\bar{R}_{n+2}^{\ast}$ and $k_h/k_l$. The closed-form optimal pricing solutions $I_0^{HP}$, $E_1^{HP}$, and $E_2^{HP}$ in Table IV are explicitly given as follows, respectively.
\end{proposition}
\setcounter{table}{3}
\begin{table}[!t]
\setlength{\tabcolsep}{2pt}
\renewcommand{\arraystretch}{1.1}
\caption{Optimal Pricing under Heavy-Priority Strategy}
\label{tab:longdynamicAppendix}
\centering
\begin{tabular}{|c|c|c|c|}
\hline
\multirow{2}{*}{} &\multicolumn{3}{c|}{$\bar{R}^{\ast}_{n+1}-\bar{R}^{\ast}_{n+2}$}\\
\cline{2-4}
 &$\leq\frac{4k_l-3k_h}{4k_hk_l}$ & $\left(\frac{4k_l-3k_h}{4k_hk_l},\frac{2-\sqrt{1+\frac{k_h}{k_l}}}{k_h}\right)$ & $\geq\frac{2-\sqrt{1+\frac{k_h}{k_l}}}{k_h}$\\
\hline
$k_h<\frac{4k_l}{3}$ & $I_0^{HP}$ & $E_1^{HP}$ & $E_2^{HP}$\\
\hline
$\frac{4k_l}{3}\leq k_h<3k_l$ & N./A.  & $E_1^{HP}$ & $E_2^{HP}$\\
\hline
$k_h\geq3k_l$ & N./A. & N./A.  & $E_2^{HP}$\\
\hline
\end{tabular}
\end{table}
\begin{IEEEproof}
We substitute $p_l(r_l^{{HP}\ast}(n))=1-k_lr_l^{{HP}\ast}(n)$ and $p_h(r_h^{{HP}\ast}(n))=1-k_hr_h^{{HP}\ast}(n)$ into the above KKT conditions. Then the KKT conditions become
\begin{equation}
\left\{
\begin{aligned}
& k_hr_h^{{HP}\ast}(n)(1-2k_lr_l^{{HP}\ast}(n))-\mu_a^{\ast}+\mu_l^{{HP}\ast}-\mu_l^{R\ast}=0,\\
& k_h(\bar{R}^{\ast}_{n+1}-\bar{R}^{\ast}_{n+2}+(1-k_lr_l^{{HP}\ast}(n))r_l^{{HP}\ast}(n))\\
&~~+1-2k_hr_h^{L\ast}(n)+\mu_a^{\ast}+\mu_h^{L\ast}-\mu_h^{R\ast}=0,\\
&\mu_a^{\ast}\cdot\left(r_h^{{HP}\ast}(n)+\bar{R}_{n+2}^{\ast}-r_l^{{HP}\ast}(n)-\bar{R}_{n+1}^{\ast}\right)=0,\\
&\mu_{l}^{L\ast}r^{{HP}\ast}_l(n)=0, \mu_{l}^{R\ast}(r_l^{\max}-r^{{HP}\ast}_l(n))=0,\\
&\mu_{h}^{L\ast}r^{{HP}\ast}_h(n)=0, \mu_{h}^{R\ast}(r_h^{\max}-r^{{HP}\ast}_h(n))=0,\\
&\boldsymbol{\mu}^{\ast}\succeq \mathbf{0},\\
& r^{{HP}\ast}_h(n)+\bar{R}_{n+2}^{\ast}\geq r^{{HP}\ast}_l(n)+\bar{R}_{n+1}^{\ast},\\
&0\leq r^{{HP}\ast}_l(n)\leq r_l^{\max}, 0\leq r^{{HP}\ast}_h(n)\leq r_h^{\max}.\\
\end{aligned}
\right.
\end{equation}

The solutions can be summarized as follows.

(1) The interior point solution is given by
\begin{equation}\label{Specialinterior}
\left\{
  \begin{aligned}
   &\textstyle r^{{HP}I}_l(n)=\frac{1}{2k_l},\\
   &\textstyle r^{{HP}I}_h(n)=\frac{1}{2}\left(\frac{1}{4k_l}+\frac{1}{k_h}+\bar{R}^{\ast}_{n+1}-\bar{R}^{\ast}_{n+2}\right),\\
   &\textstyle\boldsymbol{\mu}^{\ast}=\mathbf{0}.\\
  \end{aligned}
\right.
\end{equation}
We check the feasibility conditions. $ r^{{HP}I}_h(n)+\bar{R}_{n+2}^{\ast}\geq r^{{HP}I}_l(n)+\bar{R}_{n+1}^{\ast}\Rightarrow \bar{R}^{\ast}_{n+1}-\bar{R}^{\ast}_{n+2}<\frac{1}{k_h}-\frac{3}{4k_l}$. $r^{{HP}I}_h(n)<r_h^{\max}\Rightarrow \bar{R}^{\ast}_{n+1}-\bar{R}^{\ast}_{n+2}<\frac{1}{k_h}-\frac{1}{4k_l}$. So the feasibility condition is $\bar{R}^{\ast}_{n+1}-\bar{R}^{\ast}_{n+2}<\frac{1}{k_h}-\frac{3}{4k_l}$. We check $-2k_lk_hr_h^{{HP}I}(n)<0, |\nabla^2\mathcal{L}|=4k_lk_h^2r_h^{{HP}I}(n)-k_h^2(1-2k_lr_l^{{HP}I}(n))^2>0$, so the interior point solution is a maximum solution (if exists).

(2) The extreme point solutions are given as follows.

The extreme point solution $E_2^{HP}$: For $(r^{{HP}E'}_l(n),r^{{HP}E'}_h(n))=(r_h^{\max}+\bar{R}^{\ast}_{n+2}-\bar{R}^{\ast}_{n+1},r_h^{\max})$,
\begin{equation}\label{Specialboundary1}
\left\{
\begin{aligned}
  &\mu_a^{\ast}=1-2k_lr^{{HP}E'}_l(n)>0,\\
  &\mu_h^{R\ast}=-k_lk_h(r^{{HP}E'}_l(n))^2-2k_lr^{{HP}E'}_l(n)+1>0.\\
\end{aligned}
\right.
\end{equation}
which holds iff $0<r^{{HP}E'}_l(n)<\frac{-1+\sqrt{1+\frac{k_h}{k_l}}}{k_h}$, since $\frac{-1+\sqrt{1+\frac{k_h}{k_l}}}{k_h}< \frac{1}{2k_l}$ holds strictly. This implies that $r_h^{\max}+\bar{R}^{\ast}_{n+2}-\bar{R}^{\ast}_{n+1}<\frac{-1+\sqrt{1+\frac{k_h}{k_l}}}{k_h}$, i.e., $\bar{R}^{\ast}_{n+1}-\bar{R}^{\ast}_{n+2}>r_h^{\max}-\frac{-1+\sqrt{1+\frac{k_h}{k_l}}}{k_h}=\frac{2-\sqrt{1+\frac{k_h}{k_l}}}{k_h}$.

The extreme point solution $E_1^{HP}$: For $(r^{{HP}E}_l(n),r^{{HP}E}_h(n))$, we have
\begin{equation}\label{Specialboundary2}
\left\{
  \begin{aligned}
   &\textstyle r^{{HP}E}_l(n)=\frac{-(\bar{R}^{\ast}_{n+1}-\bar{R}^{\ast}_{n+2})+\sqrt{(\bar{R}^{\ast}_{n+1}-\bar{R}^{\ast}_{n+2})^2+\frac{3}{k_lk_h}}}{3},\\
   &\textstyle r^{{HP}E}_h(n)=\frac{2(\bar{R}^{\ast}_{n+1}-\bar{R}^{\ast}_{n+2})+\sqrt{(\bar{R}^{\ast}_{n+1}-\bar{R}^{\ast}_{n+2})^2+\frac{3}{k_lk_h}}}{3},\\
   &\textstyle \mu_a^{\ast}=k_hr^{{HP}E}_h(n)(1-2k_lr^{{HP}E}_l(n))>0,\\
  \end{aligned}
\right.
\end{equation}
where the last one $\Leftrightarrow \bar{R}^{\ast}_{n+1}-\bar{R}^{\ast}_{n+2}>\frac{1}{k_h}-\frac{3}{4k_l}$. We check the feasibility conditions. $0<r^{{HP}E}_l(n)<r_l^{\max}=\frac{1}{k_l}\Rightarrow \bar{R}^{\ast}_{n+1}-\bar{R}^{\ast}_{n+2}>\frac{1}{2}\left(\frac{1}{k_h}-\frac{3}{k_l}\right)$. $0<r^{{HP}E}_h(n)<r_h^{\max}=\frac{1}{k_h}\Rightarrow \bar{R}^{\ast}_{n+1}-\bar{R}^{\ast}_{n+2}<\frac{2-\sqrt{1+\frac{k_h}{k_l}}}{k_h}$. Since $\frac{1}{k_h}-\frac{3}{4k_l}>\frac{1}{2}\left(\frac{1}{k_h}-\frac{3}{k_l}\right)$, we conclude that the feasibility condition is $\frac{1}{k_h}-\frac{3}{4k_l}<\bar{R}^{\ast}_{n+1}-\bar{R}^{\ast}_{n+2}<\frac{2-\sqrt{1+\frac{k_h}{k_l}}}{k_h}$.

In conclusion, the optimal solution is given by (\ref{Specialinterior}), (\ref{Specialboundary1}), and (\ref{Specialboundary2}) with $\bar{R}^{\ast}_{n+1}-\bar{R}^{\ast}_{n+2}<\frac{1}{k_h}-\frac{3}{4k_l}$ ($k_h<\frac{4k_l}{3}$), $\bar{R}^{\ast}_{n+1}-\bar{R}^{\ast}_{n+2}>\frac{2-\sqrt{1+\frac{k_h}{k_l}}}{k_h}$, and $\frac{1}{k_h}-\frac{3}{4k_l}<\bar{R}^{\ast}_{n+1}-\bar{R}^{\ast}_{n+2}<\frac{2-\sqrt{1+\frac{k_h}{k_l}}}{k_h}$ ($k_h<3k_l$), respectively. Further, we have the following conclusion. If $k_h>3k_l$, the solution is given by (\ref{Specialboundary1}) only, and if $\frac{4k_l}{3}<k_h<3k_l$, the solution is given by (\ref{Specialboundary1}) and (\ref{Specialboundary2}), and if $k_h<\frac{4k_l}{3}$, the solution is given by (\ref{Specialinterior}), (\ref{Specialboundary1}), and (\ref{Specialboundary2}). This completes the proof.
\end{IEEEproof}

Besides, the optimal values corresponding to each solution are as follows. For (\ref{Specialinterior}), the optimal value is $\bar{R}^{\ast}_{n}=\bar{R}^{\ast}_{n+2}+k_h\cdot(r_h^{{HP}I}(n))^2$; For (\ref{Specialboundary1}), the optimal value is $\bar{R}^{\ast}_{n}=\bar{R}^{\ast}_{n+1}+(1-k_l\cdot r_l^{{HP}E'}(n))r_l^{{HP}E'}(n)$; For (\ref{Specialboundary2}), the optimal value is $\bar{R}^{\ast}_{n}=\bar{R}^{\ast}_{n+2}-k_lk_h\cdot r_h^{{HP}E}(n)(r_l^{{HP}E}(n))^2+r_h^{{HP}E}(n)$.

%% file: Section-Problem5.tex
\section{The Solution to Problem 5}
\subsection{Problem 5: Optimal Pricing for time slot $n$ under LP}
Based on the analysis in Problem \textbf{P4}, in Problem 5 (i.e., \textbf{P5}), we also introduce the dual variable  $\boldsymbol{\mu}^{LP}=(\mu_a^{LP},\mu_b^{LP},\mu_l^{{LP}L},\mu_l^{{LP}R},\mu_h^{{LP}L},\mu_h^{{LP}R})$ ($\boldsymbol{\mu}^{LP} \succeq \mathbf{0}$), and the associated Lagrangian $\mathcal{L}:\mathbb{R}\times\mathbb{R}\times\mathbb{R}^6\rightarrow\mathbb{R}$ is given by
\begin{align}
&\mathcal{L}(r_l^{LP}(n),r_h^{LP}(n),\boldsymbol{\mu}^{LP})\notag\\
&=\bar{R}_n^{LP}(r_l^{LP}(n),r_h^{LP}(n))\notag\\
&~~+\mu_a^{LP}\cdot\left(r_l^{LP}(n)+\bar{R}^{\ast}_{n+1}-r_h^{LP}(n)-\bar{R}^{\ast}_{n+2}\right)\notag\\
&~~+\mu_b^{LP}\cdot\left(r_h^{LP}(n)+\bar{R}^{\ast}_{n+2}-\bar{R}^{\ast}_{n+1}\right)\notag\\
&~~+\mu_{l}^{{LP}L}r_l(n)+\mu_{l}^{{LP}R}(r_l^{\max}-r_l(n))\notag\\
&~~+\mu_{h}^{{LP}L}r_h(n)+\mu_{h}^{{LP}R}(r_h^{\max}-r_h(n)).
\end{align}
The KKT conditions are given by (\ref{eq:KKTH1}), (\ref{eq:KKTH2}), (\ref{eq:KKTH3}), (\ref{eq:KKTH4}), and (\ref{eq:KKTH5}).
\begin{gather}
\label{eq:KKTH1}
\begin{aligned}
&\frac{\partial \mathcal{L}(r_l^{LP}(n),r_h^{LP}(n),\boldsymbol{\mu}^{LP})}{\partial r_l^{LP}(n)}\left| _{(r_l^{{LP}\ast}(n),r_h^{{LP}\ast}(n),\boldsymbol{\mu}^{{LP}\ast})}\right.=0,\\
\end{aligned}\\
\label{eq:KKTH2}
\begin{aligned}
&\frac{\partial \mathcal{L}(r_l^{LP}(n),r_h^{LP}(n),\boldsymbol{\mu}^{LP})}{\partial r_h^{LP}(n)}\left| _{(r_l^{{LP}\ast}(n),r_h^{{LP}\ast}(n),\boldsymbol{\mu}^{{LP}\ast})}\right.=0,\\
\end{aligned}\\
\label{eq:KKTH3}
\begin{aligned}
&\textstyle\mu_a^{{LP}\ast}\cdot\left(r_l^{{LP}\ast}(n)+\bar{R}^{\ast}_{n+1}-r_h^{{LP}\ast}(n)-\bar{R}^{\ast}_{n+2}\right)=0,\\
&\textstyle\mu_b^{{LP}\ast}\cdot\left(r_h^{{LP}\ast}(n)+\bar{R}^{\ast}_{n+2}-\bar{R}^{\ast}_{n+1}\right)=0,\\
&\mu_{l}^{{LP}L\ast}r^{{LP}\ast}_l(n)=0, \mu_{l}^{{LP}R\ast}(r_l^{\max}-r^{{LP}\ast}_l(n))=0,\\
&\mu_{h}^{{LP}L\ast}r^{{LP}\ast}_h(n)=0, \mu_{h}^{{LP}R\ast}(r_h^{\max}-r^{{LP}\ast}_h(n))=0,\\
\end{aligned}\\
\label{eq:KKTH4}
\boldsymbol{\mu}^{{LP}\ast}\succeq \mathbf{0},\\
\label{eq:KKTH5}
\begin{aligned}
& \bar{R}^{\ast}_{n+1}\leq r_h^{{LP}\ast}(n)+\bar{R}^{\ast}_{n+2}\leq r_l^{{LP}\ast}(n)+\bar{R}^{\ast}_{n+1},\\
& 0\leq r^{{LP}\ast}_l(n)\leq r_l^{\max}, 0\leq r^{{LP}\ast}_h(n)\leq r_h^{\max}.\\
\end{aligned}
\end{gather}

We now analyze the solutions that satisfying KKT conditions. The results are shown in \emph{Lemma 4} and \emph{Proposition 4}. We provide the detailed proofs as follows, respectively.
\subsection{Proof for Lemma 4}\label{appendth3}
\begin{lemma}\label{lemma4}
Only three possible solutions satisfy the KKT conditions of \textbf{P5}, i.e., the interior point solution $I_0^{LP}:\big(r^{{LP}I}_l(n),r^{{LP}I}_h(n)\big)$, and two extreme point solutions $E_1^{LP}:\big(r^{{LP}E}_l(n),r^{{LP}E}_h(n)\big)$ and $E_2^{LP}:\big(r^{{LP}E'}_l(n),r^{{LP}E'}_h(n)\big)$.
\end{lemma}
\begin{IEEEproof}
We first consider the interior point solution. In this case, no constraints are active, i.e., all constraints are strictly inequalities. The solution $r^{{LP}I}_l(n),r^{{LP}I}_h(n)$ can be given by
\begin{equation} \label{eq:interiorH}
\left\{
  \begin{aligned}
   & p'_l(r_l^{{LP}\ast}(n))p_h(r_h^{{LP}\ast}(n))(\bar{R}^{\ast}_{n+1}\!\!-\!\!\bar{R}^{\ast}_{n+2})\!\!+\!\!r_l^{{LP}\ast}\!(n)p'_l(r_l^{{LP}\ast}\!(n))\\
   &~~+p_l(r_l^{{LP}\ast}(n))-p_h(r_h^{{LP}\ast}(n))r_h^{{LP}\ast}(n)p'_l(r_l^{{LP}\ast}(n))=0,\\
   &-(1-p_l(r_l^{{LP}\ast}(n)))p'_h(r_h^{{LP}\ast}(n))(\bar{R}^{\ast}_{n+1}-\bar{R}^{\ast}_{n+2})\\
   &+\!\!(1\!\!-\!\!p_l(r_l^{{LP}\ast}(n)))(p'_h(r_h^{{LP}\ast}(n))r_h^{{LP}\ast}(n)\!+\!p_h(r_h^{{LP}\ast}(n)))\!\!=\!\!0,\\
   &\boldsymbol{\mu}^{{LP}\ast}=\mathbf{0}.\\
  \end{aligned}
\right.
\end{equation}
The condition is that all constraints are strictly inequalities.

Then we consider the extreme point solutions. In this case, at least one constraint is active (one, two, or three here). We further look at the feasible region $\mathcal{F}^{LP}$, which is composed of six line segments. Then $\mathcal{F}^{LP}\neq \emptyset$ iff $r_h^{\max}\geq \bar{R}^{\ast}_{n+1}-\bar{R}^{\ast}_{n+2}\geq0$.

\begin{tikzpicture}[scale=1.3]
    \draw [<->,thick] (0,2.5) node (yaxis) [above] {$r_h^{LP}(n)$}
        |- (3.0,0) node (xaxis) [right] {$r_l^{LP}(n)$};
    \draw (1,0) coordinate (a_1) -- (1,0) coordinate (a_2);
    \draw (0,0.3) coordinate (a_1) -- (1.5,1.8) coordinate (b_2) node[right, text width=14em] {$r_h^{LP}(n)=r_l^{LP}(n)+\bar{R}^{\ast}_{n+1}-\bar{R}^{\ast}_{n+2}$};
    \draw (0,0.3) coordinate (a_2) -- (2.3,0.3) coordinate (b_3) node[right, text width=10em] {$r_h^{LP}(n)=\bar{R}^{\ast}_{n+1}-\bar{R}^{\ast}_{n+2}$};
    \draw (0,1.3) coordinate (a_2) -- (2.3,1.3) coordinate (b_3) node[right, text width=10em] {$r_h^{LP}(n)=r_l^{\max}$};
    \draw (1.3,0) coordinate (a_2) -- (1.3,2.5) coordinate (b_3) node[right, text width=10em] {$r_l^{LP}(n)=r_l^{\max}$};
    \draw (0,1.3)  node[left] {$r_h^{\max}$};
    \draw (1.3,0)  node[below]  {$r_l^{\max}$};
    \draw (0,0)  node[left] {$0$};
    \draw (0.9,0.8) node {$\mathcal{F}^{LP}$};
    \path[draw,thick,fill=blue!80](0,0.3)--(1,1.3)--(1.3,1.3)--(1.3,0.3)--cycle;
    \fill[red] (0,0.3) circle (2pt);
    \fill[red] (1.3,0.3) circle (2pt);
    \fill[red] (1.3,1.3) circle (2pt);
    \fill[red] (1.0,1.3) circle (2pt);
    \fill[blue] (1,0.7) circle (2pt);
    \node(c) at (2.0,0.8) [circle,draw,fill=blue!20] {$\mathcal{F}^{LP}$};
    \draw[blue,<-] (1,0.7) .. controls +(0.5,0.5) .. (c);
\end{tikzpicture}

\begin{tikzpicture}[scale=1.3]
    \draw [<->,thick] (0,2.5) node (yaxis) [above] {$r_h^{LP}(n)$}
        |- (3.0,0) node (xaxis) [right] {$r_l(n)$};
    \draw (1,0) coordinate (a_1) -- (1,0) coordinate (a_2);
    \draw (0,0.3) coordinate (a_1) -- (1.5,1.8) coordinate (b_2) node[right, text width=14em] {$r_h^{LP}(n)=r_l^{LP}(n)+\bar{R}^{\ast}_{n+1}-\bar{R}^{\ast}_{n+2}$};
    \draw (0,0.3) coordinate (a_2) -- (2,0.3) coordinate (b_3) node[right, text width=10em] {$r_h^{LP}(n)=\bar{R}^{\ast}_{n+1}-\bar{R}^{\ast}_{n+2}$};
    \draw (0,2.2) coordinate (a_2) -- (2,2.2) coordinate (b_3) node[right, text width=10em] {$r_h^{LP}(n)=r_h^{\max}$};
    \draw (1.3,0) coordinate (a_2) -- (1.3,2.5) coordinate (b_3) node[right, text width=10em] {$r_l^{LP}(n)=r_l^{\max}$};
    \draw (0,2.2)  node[left] {$r_h^{\max}$};
    \draw (1.3,0)  node[below]  {$r_l^{\max}$};
    \draw (0,0)  node[left] {$0$};
    \draw (0.9,0.8) node {$\mathcal{F}^{LP}$};
    \path[draw,thick,fill=blue!80](0,0.3)--(1.3,1.6)--(1.3,0.3)--cycle;
    \fill[red] (0,0.3) circle (2pt);
    \fill[red] (1.3,1.6) circle (2pt);
    \fill[red] (1.3,0.3) circle (2pt);
    \fill[blue] (1,1) circle (2pt);
    \node(c) at (2.5,1) [circle,draw,fill=blue!20] {$\mathcal{F}^{LP}$};
    \draw[blue,<-] (1,1) .. controls +(1,1) and +(-1,-1) .. (c);
\end{tikzpicture}

(1) We claim that there exist at most five intersection points on the boundary. So there exist at most five extreme point solutions with two or three active constraints.

\textbf{(i).} If $\bar{R}^{\ast}_{n+1}-\bar{R}^{\ast}_{n+2}< r_h^{\max}< r_l^{\max}+\bar{R}^{\ast}_{n+1}-\bar{R}^{\ast}_{n+2}$, there exist four intersection points. $(r^{{LP}E'}_l(n),r^{{LP}E'}_h(n))=(0,\bar{R}^{\ast}_{n+1}-\bar{R}^{\ast}_{n+2})$, $(r_l^{\max},\bar{R}^{\ast}_{n+1}-\bar{R}^{\ast}_{n+2})$, $(r_l^{\max},r_h^{\max})$, and $(r_h^{\max}+\bar{R}^{\ast}_{n+2}-\bar{R}^{\ast}_{n+1},r_h^{\max})$, the corresponding dual variables are ($\mu_a^{{LP}\ast}>0,\mu_b^{{LP}\ast}>0,\mu_{l}^{{LP}L\ast}>0$), ($\mu_b^{{LP}\ast}>0,\mu_{l}^{{LP}R\ast}>0$), ($\mu_l^{{LP}R\ast}>0,\mu_{h}^{{LP}R\ast}>0$), and ($\mu_a^{{LP}\ast}>0,\mu_{h}^{{LP}R\ast}>0$).

For $(r^{{LP}E'}_l(n),r^{{LP}E'}_h(n))=(0,\bar{R}^{\ast}_{n+1}-\bar{R}^{\ast}_{n+2})$, ($\mu_a^{H\ast}>0,\mu_b^{H\ast}>0,\mu_{l}^{{LP}L\ast}>0$), we have $p_l(0)+\mu_a^{{LP}\ast}+\mu_{l}^{{LP}L\ast}=0$ according to (\ref{eq:KKTH1}), which is impossible.

For $(r^{{LP}E'}_l(n),r^{{LP}E'}_h(n))=(r_l^{\max},\bar{R}^{\ast}_{n+1}-\bar{R}^{\ast}_{n+2})$, ($\mu_b^{{LP}\ast}>0,\mu_{l}^{{LP}R\ast}>0$), we have $p'_l(r_l^{\max})-\mu_{l}^{{LP}R\ast}=0$ according to (\ref{eq:KKTH1}), which is also impossible.

For $(r^{{LP}E'}_l(n),r^{{LP}E'}_h(n))=(r_l^{\max},r_h^{\max})$, ($\mu_l^{{LP}R\ast}>0,\mu_{h}^{{LP}R\ast}>0$), we have $r_l^{\max}p'_l(r_l^{\max})=\mu_l^{{LP}R\ast}>0$, which is a contradiction.

For $(r^{{LP}E'}_l(n),r^{{LP}E'}_h(n))=(r_h^{\max}+\bar{R}^{\ast}_{n+2}-\bar{R}^{\ast}_{n+1},r_h^{\max})$, ($\mu_a^{{LP}\ast}>0,\mu_{h}^{{LP}R\ast}>0$), we have $(1-p_l(r^{\ast}_l(n)))p'_h(r_h^{\max})r^{\ast}_l(n)-\mu_a^{{LP}\ast}-\mu_{h}^{{LP}R\ast}=0$ according to (\ref{eq:KKTH2}), which is also impossible.

\textbf{(ii).} If $r_h^{\max}\geq r_l^{\max}+\bar{R}^{\ast}_{n+1}-\bar{R}^{\ast}_{n+2}$, there is another intersection point $(r^{{LP}E'}_l(n),r^{{LP}E'}_h(n))=(r_l^{\max},r_l^{\max}+\bar{R}^{\ast}_{n+1}-\bar{R}^{\ast}_{n+2})$, and the dual variables are ($\mu_{a}^{{LP}\ast}>0,\mu_{l}^{{LP}R\ast}>0$). We need to check the following condition
\begin{equation}\label{boundaryH1}
\left\{
  \begin{aligned}
   &\mu_{a}^{{LP}\ast}=p'_hr_l^{\max}+p_h>0,\\
   &\mu_{l}^{{LP}R\ast}=p'_lr_l^{\max}(1-p_h)+\mu_{a}^{{LP}\ast}>0.\\
  \end{aligned}
\right.
\end{equation}

\textbf{Conclusion:} We conclude that the only possible extreme point solution with more than one active constraints is $(r^{{LP}E'}_l(n),r^{{LP}E'}_h(n))=(r_h^{\max}+\bar{R}^{\ast}_{n+2}-\bar{R}^{\ast}_{n+1},r_h^{\max})$ such that (\ref{boundaryH1}) holds.

(2) We consider only one active constraint.

\textbf{(i).} If $\bar{R}^{\ast}_{n+1}-\bar{R}^{\ast}_{n+2}\leq r_h^{\max}< r_l^{\max}+\bar{R}^{\ast}_{n+1}-\bar{R}^{\ast}_{n+2}$, $\mathcal{F}^{LP}$ is a trapezoid.

\emph{Case 1}: $\mu_a^{{LP}\ast}>0$ and others are all $0$. The solution is given by
\begin{equation}\label{boundaryH2}
\left\{
  \begin{aligned}
   &-p'_l(r_l^{{LP}\ast}(n))r_l^{{LP}\ast}(n)p_h(r_h^{{LP}\ast}(n))\\
   &~~+p'_l(r_l^{{LP}\ast}(n))r_l^{{LP}\ast}(n)+p_l(r_l^{{LP}\ast}(n))+\mu_{a}^{{LP}\ast}=0,\\
   &(1-p_l(r_l^{{LP}\ast}(n)))(r_l^{{LP}\ast}(n)p'_h(r_h^{{LP}\ast}(n))\\
   &~~+p_h(r_h^{{LP}\ast}(n))-\mu_{a}^{{LP}\ast}=0,\\
   & r_h^{{LP}\ast}(n)=r_l^{{LP}\ast}(n)+\bar{R}^{\ast}_{n+1}-\bar{R}^{\ast}_{n+2}.\\
  \end{aligned}
\right.
\end{equation}

\emph{Case 2}: $\mu_b^{{LP}\ast}>0$ and others are all $0$. Then $r^{{LP}\ast}_l(n)=\bar{R}^{\ast}_{n+1}-\bar{R}^{\ast}_{n+2}$, we have $(1-p_l(r_l^{{LP}\ast}(n)))p_h(r_h^{{LP}\ast}(n))+\mu_b^{{LP}\ast}=0$ according to (\ref{eq:KKTH2}), which is impossible.

\emph{Case 3}: $\mu_{l}^{{LP}R\ast}>0$ and others are all $0$. Then $r^{{LP}\ast}_l(n)=r_l^{\max}$, we have $\mu_{l}^{{LP}R\ast}=-p_h^2\cdot\frac{p'_l}{p'_h}+p'_l\cdot r_l^{\max}<0$ by the following equations.
\begin{equation}
\left\{
  \begin{aligned}
   &-p'_lp_h(\bar{R}^{\ast}_{n+1}-\bar{R}^{\ast}_{n+2}-r^{\ast}_h(n))+p'_lr_l^{\max}-\mu_{l}^{{LP}R\ast}=0,\\
   &-p'_h(\bar{R}^{\ast}_{n+1}-\bar{R}^{\ast}_{n+2}-r^{\ast}_h(n))+p_h=0.\\
  \end{aligned}
\right.
\end{equation}
Hence, it is not a solution.

\emph{Case 4}: $\mu_{h}^{{LP}R\ast}>0$ and others are all $0$. Then $r^{{LP}\ast}_h(n)=r_h^{\max}$, we have $-(1-p_l)p'_h(\bar{R}^{\ast}_{n+1}-\bar{R}^{\ast}_{n+2}-r_h^{\max})-\mu_{h}^{{LP}R\ast}=0$ according to (\ref{eq:KKTH2}), which is impossible.

\textbf{(ii).} If $r_l^{\max}\geq r_s^{\max}+\bar{R}^{\ast}_{n+1}-\bar{R}^{\ast}_{n+2}$, $\mathcal{F}^{LP}$ is a triangle. Then only Cases 1, 2, and 3 need to be considered.

\textbf{Conclusion:} We conclude that the only possible extreme point solution with one active constraints $(r^{{LP}E}_l(n),r^{{LP}E}_h(n))$ is given by (\ref{boundaryH2}).

According to the above analysis, we conclude that the solution has THREE cases: $(r_l^{{LP}I}(n),r_h^{{LP}I}(n))$ given by (\ref{eq:interiorH}), $(r^{{LP}E'}_l(n),r^{{LP}E'}_h(n))$ given by (\ref{boundaryH1}), and $(r^{{LP}E}_l(n),r^{{LP}E}_h(n))$ given by (\ref{boundaryH2}). This completes the proof.
\end{IEEEproof}

\subsection{Proof for Proposition 4}\label{appendth4}
\begin{proposition}\label{prop:{LP}TS}
The optimal solution to Problem \textbf{P5} can also be summarized in a table as in Table \ref{tab:longdynamic} (i.e., Table \ref{tab:hybriddynamic}), only with different conditions in the rows and the columns and expressions of $I_0^{LP}$, $E_1^{LP}$ and $E_2^{LP}$.
\end{proposition}
\begin{table}[!t]
\setlength{\tabcolsep}{2pt}
\renewcommand{\arraystretch}{1.1}
\caption{Optimal Pricing under Light-Priority Strategy}
\label{tab:hybriddynamic}
\centering
\begin{tabular}{|c|c|c|c|}
\hline
\multirow{2}{*}{} &\multicolumn{3}{c|}{$\bar{R}^{\ast}_{n+1}-\bar{R}^{\ast}_{n+2}$}\\
\cline{2-4}
 &$\geq\frac{2\sqrt{1-\frac{k_h}{k_l}}-1}{k_h}$ & $\left(\frac{k_l-3k_h}{2k_hk_l},\frac{2\sqrt{1-\frac{k_h}{k_l}}-1}{k_h}\right)$ & $\leq\frac{k_l-3k_h}{2k_hk_l}$\\
\hline
$k_h<\frac{k_l}{3}$ & $I_0^{LP}$ & $E_1^{LP}$ & $E_2^{LP}$\\
\hline
$\frac{k_l}{3}\leq k_h<\frac{3k_l}{4}$ & $I_0^{LP}$ & $E_1^{LP}$ & N./A.\\
\hline
$k_h\geq\frac{3k_l}{4}$ & $I_0^{LP}$ & N./A. & N./A.\\
\hline
\end{tabular}
\vspace{-10pt}
\end{table}
\begin{IEEEproof}
We substitute $p_l(r_l^{{LP}\ast}(n))=1-k_lr_l^{{LP}\ast}(n)$ and $p_h(r_h^{{LP}\ast}(n))=1-k_hr_h^{{LP}\ast}(n)$ into the above KKT conditions. Then the KKT conditions become
\begin{equation}
\left\{
\begin{aligned}
& -k_l(1-k_hr_h^{{LP}\ast}(n))(\bar{R}^{\ast}_{n+1}-\bar{R}^{\ast}_{n+2}-r_h^{{LP}\ast}(n))\\
&~~+(1-2k_lr_l^{{LP}\ast}(n))+\mu_a^{{LP}\ast}+\mu_l^{{LP}L\ast}-\mu_l^{{LP}R\ast}=0,\\
&\textstyle k_lk_hr_l^{{LP}\ast}(n)[\bar{R}^{\ast}_{n+1}-\bar{R}^{\ast}_{n+2}+\frac{1}{k_h}-2r_h^{{LP}\ast}(n)]\\
&~~-\mu_a^{{LP}\ast}+\mu_b^{{LP}\ast}+\mu_h^{{LP}L\ast}-\mu_h^{{LP}R\ast}=0,\\
&\mu_a^{{LP}\ast}\cdot\left(r_l^{{LP}\ast}(n)+\bar{R}^{\ast}_{n+1}-r_h^{H\ast}(n)-\bar{R}^{\ast}_{n+2}\right)=0,\\
&\mu_b^{{LP}\ast}\cdot\left(r_h^{{LP}\ast}(n)+\bar{R}^{\ast}_{n+2}-\bar{R}^{\ast}_{n+1}\right)=0,\\
&\mu_{l}^{{LP}L\ast}r^{{LP}\ast}_l(n)=0, \mu_{l}^{{LP}R\ast}(r_l^{\max}-r^{{LP}\ast}_l(n))=0,\\
&\mu_{h}^{{LP}L\ast}r^{{LP}\ast}_h(n)=0, \mu_{h}^{{LP}R\ast}(r_h^{\max}-r^{{LP}\ast}_h(n))=0,\\
&\boldsymbol{\mu}^{{LP}\ast}\succeq \mathbf{0},\\
&\bar{R}^{\ast}_{n+1}\leq r_h^{{LP}\ast}(n)+\bar{R}^{\ast}_{n+2}\leq r_l^{{LP}\ast}(n)+\bar{R}^{\ast}_{n+1},\\
& 0\leq r^{{LP}\ast}_l(n)\leq r_l^{\max}, 0\leq r^{{LP}\ast}_h(n)\leq r_h^{\max}.\\
\end{aligned}
\right.
\end{equation}
The solution can be summarized as follows.

(1) The interior point solution is given by
\begin{equation}\label{SpecialinteriorH}
\left\{
  \begin{aligned}
   &\textstyle r^{{LP}I}_l(n)=\frac{k_h}{8}\left(\bar{R}^{\ast}_{n+1}-\bar{R}^{\ast}_{n+2}-\frac{1}{k_h}\right)^2+\frac{1}{2k_l},\\
   &\textstyle r^{{LP}I}_h(n)=\frac{1}{2}\left(\bar{R}^{\ast}_{n+1}-\bar{R}^{\ast}_{n+2}+\frac{1}{k_h}\right).\\
  \end{aligned}
\right.
\end{equation}
We check the feasibility conditions. $ r^{{LP}I}_h(n)+\bar{R}_{n+2}^{\ast}\geq \bar{R}_{n+1}^{\ast}\Rightarrow \bar{R}^{\ast}_{n+1}-\bar{R}^{\ast}_{n+2}<\frac{1}{k_h}$. $r^{{LP}I}_h(n)+\bar{R}_{n+2}^{\ast}\leq r^{{LP}I}_l(n)+\bar{R}_{n+1}^{\ast}\Rightarrow \bar{R}^{\ast}_{n+1}-\bar{R}^{\ast}_{n+2}>-\frac{1}{k_h}+2\sqrt{\frac{1}{k_h^2}-\frac{1}{k_lk_h}}$. $r^{{LP}I}_h(n)<r_h^{\max}\Rightarrow \frac{1}{k_h}-2\sqrt{\frac{1}{k_lk_h}}<\bar{R}^{\ast}_{n+1}-\bar{R}^{\ast}_{n+2}<\frac{1}{k_h}+2\sqrt{\frac{1}{k_lk_h}}$. Besides, $-\frac{1}{k_h}+2\sqrt{\frac{1}{k_h^2}-\frac{1}{k_lk_h}}>\frac{1}{k_h}-2\sqrt{\frac{1}{k_lk_h}}$ with $k_h<k_l$. So the feasibility condition is:
\begin{equation}\label{feasibilityint}
\left\{
\begin{aligned}
&\textstyle-\frac{1}{k_h}+2\sqrt{\frac{1}{k_h^2}-\frac{1}{k_lk_h}}<\bar{R}^{\ast}_{n+1}-\bar{R}^{\ast}_{n+2}<\frac{1}{k_h}, \text{ if } k_h\leq k_l,\\
&\textstyle\frac{1}{k_h}-2\sqrt{\frac{1}{k_lk_h}}<\bar{R}^{\ast}_{n+1}-\bar{R}^{\ast}_{n+2}<\frac{1}{k_h}, \text{ if } k_h>k_l.\\
\end{aligned}
\right.
\end{equation}
We check $-2k_l<0, |\nabla^2\mathcal{L}|>0$, so the interior point solution is a maximum solution (if exists).

(2) The extreme point solutions are given as follows.

The extreme point solution $E_2^{LP}$: $(r^{{LP}E'}_l(n),r^{{LP}E'}_h(n))=(r_l^{\max},r_l^{\max}+\bar{R}^{\ast}_{n+1}-\bar{R}^{\ast}_{n+2})$
\begin{equation}\label{SpecialboundH1}
\left\{
  \begin{aligned}
   &\mu_{a}^{{LP}\ast}=-k_hr_l^{\max}+1-k_hr^{{LP}E'}_h(n)>0,\\
   &\mu_{l}^{{LP}R\ast}=1-k_hr_l^{\max}>0.\\
  \end{aligned}
\right.
\end{equation}
$\Rightarrow \bar{R}^{\ast}_{n+1}-\bar{R}^{\ast}_{n+2}<\frac{1}{2k_h}-\frac{3}{2k_l}$, and $\bar{R}^{\ast}_{n+1}-\bar{R}^{\ast}_{n+2}<\frac{1}{k_h}-\frac{2}{k_l}$. Besides, $r_l^{\max}+\bar{R}^{\ast}_{n+1}-\bar{R}^{\ast}_{n+2}\leq r_h^{\max}\Rightarrow \bar{R}^{\ast}_{n+1}-\bar{R}^{\ast}_{n+2}\leq \frac{1}{k_h}-\frac{1}{k_l}\Rightarrow k_h<k_l$. Under this condition, we have $\frac{1}{2k_h}-\frac{3}{2k_l}<\frac{1}{k_h}-\frac{2}{k_l}$. We conclude that the feasibility condition is $\bar{R}^{\ast}_{n+1}-\bar{R}^{\ast}_{n+2}<\frac{1}{2k_h}-\frac{3}{2k_l}$ ($k_h<\frac{k_l}{3}$).

The extreme point solution $E_1^{LP}$: For  $r^{{LP}E}_l(n),r^{{LP}E}_h(n)$, we have
\begin{equation}\label{SpecialboundH2}
\left\{
  \begin{aligned}
   &\textstyle r^{{LP}E}_l(n)=\frac{-(\bar{R}^{\ast}_{n+1}-\bar{R}^{\ast}_{n+2})+\sqrt{(\bar{R}^{\ast}_{n+1}-\bar{R}^{\ast}_{n+2})^2+\frac{3}{k_lk_h}}}{3},\\
   &\textstyle r^{{LP}E}_h(n)=\frac{2(\bar{R}^{\ast}_{n+1}-\bar{R}^{\ast}_{n+2})+\sqrt{(\bar{R}^{\ast}_{n+1}-\bar{R}^{\ast}_{n+2})^2+\frac{3}{k_lk_h}}}{3},\\
   &\textstyle \mu_a^{{LP}\ast}=-k_hk_lr^{{LP}E}_l(n)(r^{{LP}E}_l(n)+r^{{LP}E}_h(n))\\
   &~~~~~~~~~~~+k_lr^{{LP}E}_l(n)>0,\\
  \end{aligned}
\right.
\end{equation}
where the last one $\Rightarrow \bar{R}^{\ast}_{n+1}-\bar{R}^{\ast}_{n+2}<-\frac{1}{k_h}+2\sqrt{\frac{1}{k_h^2}-\frac{1}{k_lk_h}}$, and $k_h<k_l$. Then $-\frac{1}{k_h}+2\sqrt{\frac{1}{k_h^2}-\frac{1}{k_lk_h}}>0\Rightarrow k_h<\frac{3}{4}k_l$. We thus have $k_h<\frac{3}{4}k_l$. We check the feasibility conditions. $0<r^{{LP}E}_l(n)<r_l^{\max}=\frac{1}{k_l}\Rightarrow \bar{R}^{\ast}_{n+1}-\bar{R}^{\ast}_{n+2}>\frac{1}{2k_h}-\frac{3}{2k_l}$. $0<r^{{LP}E}_h(n)<r_h^{\max}=\frac{1}{k_h}\Rightarrow \bar{R}^{\ast}_{n+1}-\bar{R}^{\ast}_{n+2}<\frac{2}{k_h}-\frac{1}{k_h}\sqrt{1+\frac{k_h}{k_l}}$. Since $\frac{2}{k_h}-\frac{1}{k_h}\sqrt{1+\frac{k_h}{k_l}}>-\frac{1}{k_h}+2\sqrt{\frac{1}{k_h^2}-\frac{1}{k_lk_h}}$, we conclude that the feasibility condition is $\frac{1}{2k_h}-\frac{3}{2k_l}<\bar{R}^{\ast}_{n+1}-\bar{R}^{\ast}_{n+2}<-\frac{1}{k_h}+2\sqrt{\frac{1}{k_h^2}-\frac{1}{k_lk_h}}$.

In conclusion, the optimal solution is given by (\ref{SpecialinteriorH}), (\ref{SpecialboundH1}), and (\ref{SpecialboundH2}) with the condition (\ref{feasibilityint}), $\bar{R}^{\ast}_{n+1}-\bar{R}^{\ast}_{n+2}<\frac{1}{2k_h}-\frac{3}{2k_l}$ ($k_h<\frac{k_l}{3}$), and $\frac{1}{2k_h}-\frac{3}{2k_l}<\bar{R}^{\ast}_{n+1}-\bar{R}^{\ast}_{n+2}<-\frac{1}{k_h}+2\sqrt{\frac{1}{k_h^2}-\frac{1}{k_lk_h}}$ ($k_h<\frac{3}{4}k_l$), respectively. If $k_h>\frac{3k_l}{4}$, the solution is given by (\ref{SpecialinteriorH}) only, and if $\frac{k_l}{3}<k_h<\frac{3k_l}{4}$, the solution is given by (\ref{SpecialinteriorH}) and (\ref{SpecialboundH2}), and if $k_h<\frac{k_l}{3}$, the solution is given by (\ref{SpecialinteriorH}), (\ref{SpecialboundH1}), and (\ref{SpecialboundH2}). This completes the proof.
\end{IEEEproof}

Besides, the optimal values corresponding to each solution are as follows. For (\ref{SpecialinteriorH}), the optimal value is $\bar{R}^{H\ast}_{n}=\bar{R}^{\ast}_{n+1}+k_lr_l^{HI}(n)^2$; For (\ref{SpecialboundH1}), the optimal value is $\bar{R}^{H\ast}_{n}=\bar{R}^{\ast}_{n+2}+(1-\frac{k_h}{k_l})r_h^{HE'}(n)$; For (\ref{SpecialboundH2}), the optimal value is $\bar{R}^{H\ast}_{n}=\bar{R}^{\ast}_{n+1}+(1-k_lk_hr_h^{HE}(n)r_l^{HE}(n))\cdot r_l^{HE}(n)$.

%% file: Section-Problem6.tex
\section{The Solution to Problem 6}
\subsection{Problem 6: Optimal Pricing for time slot $n$ under LD}
Similarly, in Problem 6 (i.e., \textbf{P6}), we introduce the dual variable  $\boldsymbol{\mu}^{LD}=(\mu_a^{LD},\mu_l^{{LD}{L}},\mu_l^{{LD}R},\mu_h^{{LD}{L}},\mu_h^{{LD}R})$ ($\boldsymbol{\mu}^{LD} \succeq \mathbf{0}$), and the associated Lagrangian $\mathcal{L}:\mathbb{R}\times\mathbb{R}\times\mathbb{R}^5\rightarrow\mathbb{R}$ is given by
\begin{align}
&\mathcal{L}(r_l^{LD}(n),r_h^{LD}(n),\boldsymbol{\mu}^{LD})\notag\\
&=\bar{R}_n^{S}(r_l^{LD}(n),r_h^{LD}(n))\notag\\
&~~+\mu_a^{LD}\cdot (\bar{R}^{\ast}_{n+1}-r_h^{LD}(n)-\bar{R}^{\ast}_{n+2})\notag\\
&~~+\mu_{l}^{{LD}L}r_l(n)+\mu_{l}^{{LD}R}(r_l^{\max}-r_l(n))\notag\\
&~~+\mu_{h}^{{LD}L}r_h(n)+\mu_{h}^{{LD}R}(r_h^{\max}-r_h(n)).
\end{align}

The KKT conditions are given by (\ref{eq:KKTS1}), (\ref{eq:KKTS2}), (\ref{eq:KKTS3}), (\ref{eq:KKTS4}), and (\ref{eq:KKTS5}).
\begin{gather}
\label{eq:KKTS1}
\frac{\partial \mathcal{{L}}(r_l^{LD}(n),r_h^{LD}(n),\boldsymbol{\mu}^{LD})}{\partial r_l^{LD}(n)}\left| _{(r_l^{{LD}\ast}(n),r_h^{{LD}\ast}(n),\boldsymbol{\mu}^{{LD}\ast})}\right.=0,\\
\label{eq:KKTS2}
\frac{\partial \mathcal{{L}}(r_l^{LD}(n),r_h^{LD}(n),\boldsymbol{\mu}^{LD})}{\partial r_h^{LD}(n)}\left|_{(r_l^{{LD}\ast}(n),r_h^{{LD}\ast}(n),\boldsymbol{\mu}^{{LD}\ast})}\right.=0,\\
\label{eq:KKTS3}
\begin{aligned}
&\textstyle\mu^{{LD}\ast}_a\cdot\left(\bar{R}^{\ast}_{n+1}-r_h^{{LD}\ast}(n)-\bar{R}^{\ast}_{n+2}\right)=0,\\
&\textstyle\mu_{l}^{{LD}{L}\ast}r^{{LD}\ast}_l(n)=0, \mu_{l}^{{LD}R\ast}(r_l^{\max}-r^{{LD}\ast}_l(n))=0,\\
&\mu_{h}^{{LD}{L}\ast}r^{{LD}\ast}_h(n)=0, \mu_{h}^{{LD}R\ast}(r_h^{\max}-r^{{LD}\ast}_h(n))=0,\\
\end{aligned}\\
\label{eq:KKTS4}
\boldsymbol{\mu}^{{LD}\ast}\succeq \mathbf{0},\\
\label{eq:KKTS5}
\begin{aligned}
&r_h^{{LD}\ast}(n)+\bar{R}^{\ast}_{n+2}\leq \bar{R}^{\ast}_{n+1},\\
&0\leq r_l^{{LD}\ast}(n)\leq r_l^{\max}, 0\leq r_h^{{LD}\ast}(n)\leq r_h^{\max}.\\
\end{aligned}
\end{gather}

Now we are ready to analyze the solutions that satisfying the above KKT conditions. We will show that the solution is unique by \emph{Proposition 5}.
\begin{proposition}\label{prop:LTSAppendix}
The optimal prices in time slot $n$ under the {L}D strategy are given by the interior point solution $I_0^{{LD}}:$
\begin{equation}\label{SpecialinteriorSAppendix}
 r_l^{{LD}}(n)=\frac{1}{2k_l},r_h^{{LD}}(n)=\min(\bar{R}^{\ast}_{n+1}-\bar{R}^{\ast}_{n+2}, r_h^{\max}).
\end{equation}
\end{proposition}

We first prove the following lemma, and then proceed to prove the proposition.
\begin{lemma}\label{lemma:LTS}
The optimal pricing under the {LD} strategy is uniquely given by the feasible interior point solution.
\end{lemma}
\begin{IEEEproof}
We first consider the interior point solution. In this case, no constraints are active, i.e., all constraints are strictly inequalities. The solution $r_l^{LD\ast}(n),r_h^{LD\ast}(n)$ can be given by
\begin{equation}\label{eq:interiorS}
\left\{
  \begin{aligned}
   &p'_l(r_l^{{LD}\ast}(n))r_l^{{LD}\ast}(n)+p_l(r_l^{{LD}\ast}(n))=0,\\
   &r^{{LD}\ast}_h(n)=\min \{\bar{R}^{\ast}_{n+1}-\bar{R}^{\ast}_{n+2}, r_h^{\max}\},\\
   &\boldsymbol{\mu}^{{LD}\ast}=\mathbf{0}.\\
  \end{aligned}
\right.
\end{equation}

Then we consider the extreme point solutions. In this case, at least one constraint is active (one or two here). We further look at the feasible region $\mathcal{F}^{LD}$, which is composed of four line segments.

\begin{tikzpicture}[scale=1.1]
    \draw [<->,thick] (0,2.5) node (yaxis) [above] {$r_h^{LD}(n)$}
        |- (3.0,0) node (xaxis) [right] {$r_l^{LD}(n)$};
    \draw (1,0) coordinate (a_1) -- (1,0) coordinate (a_2);
    \draw (0,1.3) coordinate (a_2) -- (2,1.3) coordinate (b_3) node[right, text width=10em] {$r_h^{LD}(n)=\bar{R}^{\ast}_{n+1}-\bar{R}^{\ast}_{n+2}$};
    \draw (0,1.7) coordinate (a_2) -- (2,1.7) coordinate (b_3) node[right, text width=10em] {$r_h^{LD}(n)=r_l^{\max}$};
    \draw (1.3,0) coordinate (a_2) -- (1.3,2.5) coordinate (b_3) node[right, text width=10em] {$r_l^{LD}(n)=r_l^{\max}$};
    \draw (0,1.7)  node[left] {$r_h^{\max}$};
    \draw (1.3,0)  node[below]  {$r_l^{\max}$};
    \draw (0,0)  node[left] {$0$};
    \draw (0.9,0.8) node {$\mathcal{F}^{LD}$};
    \path[draw,thick,fill=blue!80](0,0)--(0,1.3)--(1.3,1.3)--(1.3,0)--cycle;
    \fill[red] (0,0) circle (2pt);
    \fill[red] (0,1.3) circle (2pt);
    \fill[red] (1.3,1.3) circle (2pt);
    \fill[red] (1.3,0) circle (2pt);
    \fill[blue] (1,0.7) circle (2pt);
    \node(c) at (2.5,0.6) [circle,draw,fill=blue!20] {$\mathcal{F}^{LD}$};
    \draw[blue,<-] (1,0.7) .. controls +(1,1) and +(-1,-1) .. (c);
\end{tikzpicture}

\begin{tikzpicture}[scale=1.1]
    \draw [<->,thick] (0,2.5) node (yaxis) [above] {$r_h^{LD}(n)$}
        |- (3.0,0) node (xaxis) [right] {$r_l^{LD}(n)$};
    \draw (1,0) coordinate (a_1) -- (1,0) coordinate (a_2);
    \draw (0,2.1) coordinate (a_2) -- (2,2.1) coordinate (b_3) node[right, text width=10em] {$r_h^{LD}(n)=\bar{R}^{\ast}_{n+1}-\bar{R}^{\ast}_{n+2}$};
    \draw (0,1.7) coordinate (a_2) -- (2,1.7) coordinate (b_3) node[right, text width=10em] {$r_h^{LD}(n)=r_l^{\max}$};
    \draw (1.3,0) coordinate (a_2) -- (1.3,2.5) coordinate (b_3) node[right, text width=10em] {$r_l^{LD}(n)=r_l^{\max}$};
    \draw (0,1.7)  node[left] {$r_h^{\max}$};
    \draw (1.3,0)  node[below]  {$r_l^{\max}$};
    \draw (0,0)  node[left] {$0$};
    \draw (0.9,0.8) node {$\mathcal{F}^{LD}$};
    \path[draw,thick,fill=blue!80](0,0)--(0,1.7)--(1.3,1.7)--(1.3,0)--cycle;
    \fill[red] (0,0) circle (2pt);
    \fill[red] (0,1.7) circle (2pt);
    \fill[red] (1.3,1.7) circle (2pt);
    \fill[red] (1.3,0) circle (2pt);
    \fill[blue] (1,0.7) circle (2pt);
    \node(c) at (2.5,0.8) [circle,draw,fill=blue!20] {$\mathcal{F}^{LD}$};
    \draw[blue,<-] (1,0.7) .. controls +(1,1) and +(-1,-1) .. (c);
\end{tikzpicture}

(1) We claim that there exist at most four intersection points on the boundary. So there exist at most four extreme point solutions with two active constraints.

\textbf{(i).} If $\bar{R}^{\ast}_{n+1}-\bar{R}^{\ast}_{n+2}< r_h^{\max}$, there exist four intersection points. $(r^{{LD}E'}_l(n),r^{{LD}E'}_h(n))=(0,0)$, $(r_l^{\max},0)$, $(r_l^{\max},\bar{R}^{\ast}_{n+1}-\bar{R}^{\ast}_{n+2})$, and $(0,\bar{R}^{\ast}_{n+1}-\bar{R}^{\ast}_{n+2})$, the corresponding dual variables are ($\mu_l^{{LD}{L}\ast}>0,\mu_h^{{LD}{L}\ast}>0$), ($\mu_l^{{LD}R\ast}>0,\mu_h^{{LD}{L}\ast}>0$),  ($\mu_a^{{LD}\ast}>0,\mu_{l}^{{LD}R\ast}>0$), and ($\mu_l^{{LD}{L}\ast}>0,\mu_{a}^{{LD}\ast}>0$).

For $(r^{{LD}E'}_l(n),r^{{LD}E'}_h(n))=(0,0)$, ($\mu_l^{{LD}{L}\ast}>0,\mu_h^{{LD}{L}\ast}>0$), we have $\bar{R}^{\ast}_{n+1}=\bar{R}^{\ast}_{n+2}$, which is a contradiction.

For $(r^{{LD}E'}_l(n),r^{{LD}E'}_h(n))=(r_l^{\max},0)$, ($\mu_l^{{LD}R\ast}>0,\mu_h^{{LD}{L}\ast}>0$), we have $p'_l(r_l^{\max})\cdot r_l^{\max}=\mu_l^{{LD}R \ast}>0$ according to (\ref{eq:KKTS1}), which is also a contradiction.

For $(r^{{LD}E'}_l(n),r^{{LD}E'}_h(n))=(r_l^{\max},\bar{R}^{\ast}_{n+1}-\bar{R}^{\ast}_{n+2})$, ($\mu_a^{{LD}\ast}>0,\mu_{l}^{{LD}R\ast}>0$), we have $p'_l(r_l^{\max})\cdot r_l^{\max}=\mu_l^{{LD}R \ast}>0$ according to (\ref{eq:KKTS1}), which is also a contradiction.

For $(r^{{LD}E'}_l(n),r^{{LD}E'}_h(n))=(0,\bar{R}^{\ast}_{n+1}-\bar{R}^{\ast}_{n+2})$, ($\mu_l^{{LD}{L}\ast}>0,\mu_{a}^{{LD}\ast}>0$), we have $p_l(0)+\mu_l^{{LD}{L}\ast}=0$ according to (\ref{eq:KKTS1}), which is a contradiction.

\textbf{(ii).} If $\bar{R}^{\ast}_{n+1}-\bar{R}^{\ast}_{n+2}>r_h^{\max}$, another two intersection points $(r^{{LD}E'}_l(n),r^{{LD}E'}_h(n))=(r_l^{\max},r_h^{\max})$, $(0,r_h^{\max})$, and the dual variables are ($\mu_{l}^{{LD}R\ast}>0,\mu_{h}^{{LD}R\ast}>0$), ($\mu_{l}^{{LD}{L}\ast}>0,\mu_{h}^{{LD}R\ast}>0$).

For $(r^{{LD}E'}_l(n),r^{{LD}E'}_h(n))=(r_l^{\max},r_h^{\max})$, ($\mu_{l}^{{LD}R\ast}>0,\mu_{h}^{{LD}R\ast}>0$), we have $-\mu_{h}^{{LD}R\ast}=0$ according to (\ref{eq:KKTS2}), which is a contradiction.

For $(r^{{LD}E'}_l(n),r^{{LD}E'}_h(n))=(0,r_h^{\max})$, ($\mu_{l}^{{LD}{L}\ast}>0,\mu_{h}^{{LD}R\ast}>0$), we have $p_l(0)+\mu_{l}^{{LD}{L}\ast}=0$, which is a contradiction.

\textbf{Conclusion:} We conclude that no possible extreme point solution exists with more than one active constraints.

(2) We consider only one active constraint.

\textbf{(i).} If $\bar{R}^{\ast}_{n+1}-\bar{R}^{\ast}_{n+2}< r_h^{\max}$, $\mathcal{F}^{LD}$ is a rectangle.

\emph{Case 1}: $\mu_l^{{LD}{L}\ast}>0$ and others are all $0$. Then we have $r^{{LD}\ast}_l(n)=0$, we have $p_l(0)+\mu_l^{{LD}{L}\ast}=0$, no solution.

\emph{Case 2}: $\mu_h^{{LD}{L}\ast}>0$ and others are all $0$. This is impossible for (\ref{eq:KKTS2}).

\emph{Case 3}: $\mu_{l}^{{LD}R\ast}>0$ and others are all $0$. This is impossible for (\ref{eq:KKTS1}) $p'_l(r_l^{\max})\cdot r_l^{\max}=\mu_l^{{LD}R \ast}>0$.

\emph{Case 4}: $\mu_{a}^{{LD}\ast}>0$ and others are all $0$. This is impossible for (\ref{eq:KKTS2}).

\textbf{(ii).} If $r_h^{\max}<\bar{R}^{\ast}_{n+1}-\bar{R}^{\ast}_{n+2}$, $\mathcal{F}^{LD}$ is a rectangle. Then Cases 1, 2, and 3 are the same.

\emph{Case 4'}: $\mu_{h}^{{LD}R\ast}>0$ and others are all $0$. This is also impossible for (\ref{eq:KKTS2}).

\textbf{Conclusion:} We conclude that there is also no possible extreme point solution with one active constraint.

According to the above analysis, we conclude that the solution is unique, which is given by (\ref{eq:interiorS}). This completes the proof.
\end{IEEEproof}
\subsection{Proof for Proposition 5}
Given the above lemma ({L}emma \ref{lemma:LTS}), we now proceed to prove Proposition 5.
\begin{IEEEproof}
We substitute $p_l(r_l^{{LD}\ast}(n))=1-k_lr_l^{{LD}\ast}(n)$ and $p_h(r_h^{{LD}\ast}(n))=1-k_hr_h^{{LD}\ast}(n)$ into the above KKT conditions. Then the KKT conditions become
\begin{equation}
\left\{
\begin{aligned}
& (1-2k_lr_l^{{LD}\ast}(n))+\mu_l^{{LD}{L}\ast}-\mu_l^{{LD}R\ast}=0,\\
&-\mu_a^{{LD}\ast}+\mu_h^{{LD}{L}\ast}-\mu_h^{{LD}R\ast}=0,\\
&\mu_a^{{LD}\ast}\cdot\left(\bar{R}^{\ast}_{n+1}-r_h^{{LD}\ast}(n)-\bar{R}^{\ast}_{n+2}\right)=0,\\
&\mu_{l}^{{LD}{L}\ast}r^{{LD}\ast}_l(n)=0, \mu_{l}^{{LD}R\ast}(r_l^{\max}-r^{{LD}\ast}_l(n))=0,\\
&\mu_{h}^{{LD}{L}\ast}r^{{LD}\ast}_h(n)=0, \mu_{h}^{{LD}R\ast}(r_h^{\max}-r^{{LD}\ast}_h(n))=0,\\
&\boldsymbol{\mu}^{{LD}\ast}\succeq \mathbf{0},\\
& r_h^{{LD}\ast}(n)+\bar{R}^{\ast}_{n+2}\leq \bar{R}^{\ast}_{n+1},\\
& 0\leq r_l^{{LD}\ast}(n)\leq r_l^{\max}, 0\leq r_h^{{LD}\ast}(n)\leq r_h^{\max}.\\
\end{aligned}
\right.
\end{equation}
The solution can be given by the first two equations in the above KKT conditions.
\begin{equation}\label{SpecialinteriorSproof}
\left\{
  \begin{aligned}
   &r^{{LD}\ast}_l(n)=\frac{1}{2k_l},\\
   &r^{{LD}\ast}_h(n)=\min \{\bar{R}^{\ast}_{n+1}-\bar{R}^{\ast}_{n+2}, r_h^{\max}\},\\
   &\boldsymbol{\mu}^{{LD}\ast}=\mathbf{0}.\\
  \end{aligned}
\right.
\end{equation}
Note that $r^{{LD}\ast}_l(n)$ does not influence the objective value. This completes the proof.
\end{IEEEproof}

Besides, the objective function is $\bar{R}_n^{LD}(r_l^{LD}(n))=\bar{R}_{n+1}^{\ast}+p_l(r_l^{LD}(n))\cdot r_l^{LD}(n)=\bar{R}_{n+1}^{\ast}+(1-k_lr_l^{LD}(n))\cdot r_l^{LD}(n)$, where $0\leq r_l^{LD}(n)\leq r_l^{\max}$.
The second-order condition is $\ddot{\bar{R}}^{LD}_n(r_l^{LD}(n))=-2k_l<0$, so the objective function is concave for $0\leq r_l(n)\leq r_l^{\max}$. Besides, the feasible region is convex. As such, the problem is a convex problem. The optimal solution is given by the solution (\ref{SpecialinteriorSproof}). Besides, the optimal value is
$\bar{R}^{{LD}\ast}_{n}=\bar{R}^{\ast}_{n+1}+\left(1-k_l\cdot\frac{1}{2k_l}\right)\cdot \frac{1}{2k_l}=\bar{R}^{\ast}_{n+1}+\frac{1}{4k_l}$.
\section{Proof for Theorem 2}
\begin{IEEEproof}
The dynamic prices $\boldsymbol{r}^{\ast}$ and the dynamic admission policy $\boldsymbol{\pi}^{\ast}$ are derived by Algorithm 2, which is built upon the backward induction algorithm for solving dynamic programming problems. That is, we first optimally solve the revenue maximization problem in a time slot, and then determine the optimal pricing and admission solution backwards from time sot $N$ to 1, by comparing all candidate solutions in a time slot. According to the \emph{principle of optimality} [24], we know that Algorithm 2 is optimal for solving our Problem \textbf{P2}.
\end{IEEEproof}

%% file: Section-ExtensionAppendix.tex
\section{Proof for Proposition 6}
\begin{proposition}\label{prop:generalAppendix}
The optimal policy for solving the revenue maximization Problem \textbf{P1} degenerates to the heavy-priority stationary admission policy when price ratio between the heavy-traffic SU and the light-traffic SU is larger than a threshold, i.e.,
\begin{equation}\label{generalthresholdAppendix}
r_h/r_l>\theta_{th}^{HP}(p_l,p_h),
\end{equation}
where the threshold ratio $\theta_{th}^{HP}(p_l,p_h)$ can be determined by solving the following:
\begin{equation}\label{generalconditionAppendix}
r_h +\bar{R}^{\ast}_{n+M}=r_l+\bar{R}^{\ast}_{n+1},\forall n\in\{1,2,\cdots,N-M+1\}.
\end{equation}
\end{proposition}
\begin{proof}
Eq. (\ref{generalconditionAppendix}) is similar to (\ref{eq:con2}) in Section IV-C. 
By using the heavy-priority admission policy in all time slots, we can derive the expected total revenue $\bar{R}_n^{\ast}$ given $\bar{R}_{n+1}^{\ast}$ and $\bar{R}_{n+M}^{\ast}$. The latter two are the expected future revenues corresponding to the admission actions $a_n=0,1$ and $a_n=M$, respectively. We next derive the closed-form expression of $\bar{R}_n^{\ast}$.

Note that at the beginning of time slots $N-M+2$ to $N$, we cannot admit a heavy-traffic SU as it requires the consecutive $M$ time slots' occupancy. As such, we consider two choices $a_n=0$ and $a_n=1$ in the last two terms of (16), respectively, and thus have the following initial conditions for $n=N-M+2,\cdots,N+1$, i.e.,
$\bar{R}_{N+1}^{\ast}=0,$ $\bar{R}_{N}^{\ast}=p_lr_l,$ $\cdots,$ $\bar{R}_{N-M+2}^{\ast}=(M-1)p_lr_l.$
For $1\leq n\leq N-M+1$, we have
\begin{align}
\bar{R}_n^{\ast}&=(1-p_h)\bar{R}_{n+1}^{\ast}+p_h\bar{R}_{n+M}^{\ast}+p_lr_l(1-p_h)+p_hr_h\notag\\
&\textstyle=\bar{R}_{n+1}^{\ast}+\frac{p_lr_l(1-p_h)+p_hr_h}{1+(M-1)p_h}\notag\\
&\textstyle~~~-p_h\left(\bar{R}_{n+1}^{\ast}-\bar{R}_{n+M}^{\ast}-(M-1)\frac{p_lr_l(1-p_h)+p_hr_h}{1+(M-1)p_h}\right).
\end{align}
Subtracting and adding $\bar{R}_{n+2}^{\ast},\bar{R}_{n+3}^{\ast},\cdots,\bar{R}_{n+M-1}^{\ast}$ in the last term and rearranging terms will lead to
\begin{align}
&\textstyle \bar{R}_n^{\ast}-\bar{R}_{n+1}^{\ast}-\frac{p_lr_l(1-p_h)+p_hr_h}{1+(M-1)p_h}\notag\\
&\textstyle =-p_h\left(\bar{R}_{n+1}^{\ast}-\bar{R}_{n+M}^{\ast}-(M-1)\frac{p_lr_l(1-p_h)+p_hr_h}{1+(M-1)p_h}\right)\notag\\
&\textstyle =-p_h\left(\bar{R}_{n+1}^{\ast}-\bar{R}_{n+2}^{\ast}-\frac{p_lr_l(1-p_h)+p_hr_h}{1+(M-1)p_h}\right.\notag\\
&\textstyle~~~~~~~~+\bar{R}_{n+2}^{\ast}-\bar{R}_{n+3}^{\ast}-\frac{p_lr_l(1-p_h)+p_hr_h}{1+(M-1)p_h}\notag\\
&\textstyle~~~~~~~~ +\cdots+\left.\bar{R}_{n+M-1}^{\ast}-\bar{R}_{n+M}^{\ast}-\frac{p_lr_l(1-p_h)+p_hr_h}{1+(M-1)p_h}\right).
\end{align}
This becomes a difference equation, and the corresponding characteristic equation is $1=-p_h(\alpha+\alpha^2+\cdots+\alpha^{M-1})$.
By solving this equation for $\bar{R}_n^{\ast}$, we derive $\theta_{th}^{HP}(p_l,p_h,n)$ by plugging $\bar{R}_{n+1}^{\ast}-\bar{R}_{n+M}^{\ast}$ into
$r_h-r_l=\bar{R}_{n+1}^{\ast}-\bar{R}_{n+M}^{\ast}$. The final threshold $\theta_{th}^{HP}(p_l,p_h)$ is determining by optimizing $\theta_{th}^{HP}(p_l,p_h,n)$ over $n\in\{1,2,\cdots,N-M+1\}$. This shows that the derived threshold $\theta_{th}^{HP}(p_l,p_h)$ guarantees $r_h +\bar{R}^{\ast}_{n+M}\geq r_l+\bar{R}^{\ast}_{n+1}$ for all time slots $n\in\{1,2,\cdots,N-M+1\}$, which is the optimality condition for the heavy-priority admission policy.
\end{proof}

\section{Proof for Proposition 7}
\begin{proposition}
Given an arbitrary value of spectrum occupancy $M$, the optimal dynamic pricing under the heavy-priority strategy is the same as that in Proposition 3 and Table IV, once we replace  $\bar{R}_{n+1}^{\ast}-\bar{R}_{n+2}^{\ast}$ by $\bar{R}_{n+1}^{\ast}-\bar{R}_{n+M}^{\ast}$.
\end{proposition}
\begin{proof}
From the proofs of Propositions 3-5, we know that the analyses apply to an arbitrary value of spectrum occupancy $M$, due to the same problem structures compared with our previous problems (Problems 4-6). More specifically, given an arbitrary $M$, the revenue maximization problem under dynamic pricing and dynamic admission is the same as Problems 2 and 3. Similarly, we can still use the decomposition scheme to decompose the one-slot problem into three subproblems as we did for solving Problem 3. As a result, we solve the subproblem under the heavy-priority strategy similarly as in Proposition 3 and Table IV, except that the conditions in Proposition 3 and Table IV are given in terms of $\bar{R}_{n+1}^{\ast}-\bar{R}_{n+M}^{\ast}$.
\end{proof}

\section{Proof for Proposition 8}
\begin{proposition}\label{prop:GeneralPolicyAppendix}
Given the set $\mathcal{I}$ of $I$ types of SUs, there are $(I+1)!$ admission priorities. For each admission priority, there exist thresholds of the price ratios such that the optimal admission priority for a time slot is optimal for all time slots (corresponding to an optimal stationary admission policy).
\end{proposition}
\begin{proof}
As discussed in this section, in each time slot $n\in\mathcal{N}$, there are a total of $I+1$ admission actions $I\cup\{0\}$, depending on the values of the $I+1$ total revenues $\{0+\bar{R}_{n+1}^{\ast},r_i+\bar{R}_{n+i}^{\ast},\forall i\in\mathcal{I}\}$. For any two revenues $r_i+\bar{R}_{n+i}^{\ast}$ and $r_j+\bar{R}_{n+j}^{\ast}$, we denote the admission priority as $i>j$ iff $r_i+\bar{R}_{n+i}^{\ast}>r_j+\bar{R}_{n+j}^{\ast}$. This shows that we will admit a type-$i$ SU if it requests and only admit an arrived type-$j$ SU when there is no type-$i$ SU. Hence, for all time slots, there are $(I+1)!$ admission priorities that correspond to the value relations of the $I+1$ revenues, since we need to choose $I+1$ ordered priorities from $I+1$ actions.

For each admission priority, we first assume that this admission priority is optimal for all time slots, and we thus derive the expected total revenue $\bar{R}_{n}^{\ast},\forall n\in\mathcal{N}$ by specifying the admission actions in (16). Then, there exists a threshold such that the corresponding revenue conditions (ensuring the optimality of the admission priority) hold. Finally, we optimize the threshold over all time slots $n\in\mathcal{N}$ to obtain the final thresholds (similar as the proof for Proposition 6). These thresholds ensure that the admission priority is optimal for all time slots. Hence, this admission priority becomes an optimal stationary admission policy. This completes the proof.
\end{proof}